\documentclass[hidelinks,11pt,draft]{article}
\pdfoutput=1

\usepackage[utf8]{inputenc}
\usepackage{fullpage,amsfonts,bm}
\usepackage{amsmath,amssymb,amsthm,mathtools}
\usepackage{mleftright}\mleftright
\usepackage[pagebackref,final]{hyperref}
\usepackage{natbib}
\usepackage{doi}
\usepackage{algpseudocode}
\usepackage{float}
\usepackage{algorithm}
\usepackage{braket}
\usepackage[normalem]{ulem}
\usepackage{cleveref}
\usepackage[shortlabels]{enumitem}
\usepackage{graphicx}
\usepackage{subfigure}
\usepackage{tcolorbox}
\usepackage{xcolor}
\usepackage{mdframed}

\usepackage{verbatim}
\usepackage{ifdraft}

\ifdraft{\newcommand{\authnote}[3]{{\color{#3} {\bf  #1:} #2}}}{\newcommand{\authnote}[3]{}}
\usepackage{tikz}
\usetikzlibrary{quantikz}
\usetikzlibrary{calc,arrows,quotes,angles}

\usepackage{subcaption}

\usepackage{caption}

\usepackage{thmtools}
\usepackage{thm-restate}
\usepackage{tcolorbox}

\floatstyle{boxed}

\newfloat{algorithm}{t}{lop}
\floatname{algorithm}{Algorithm}

\newfloat{metaalgorithm}{t}{lop}
\floatname{metaalgorithm}{Meta-Algorithm}

\newtheorem{theorem}{Theorem}
\newtheorem{proposition}{Proposition}
\newtheorem{lemma}{Lemma}

\newtheorem{corollary}{Corollary}
\newtheorem{definition}{Definition}
\newtheorem{oracle}{Oracle}

\newmdenv[
backgroundcolor=gray!10, 
linecolor=black, 
linewidth=1pt, 
roundcorner=30pt, 
innertopmargin=3pt,        % Margin inside the box (top)
innerbottommargin=10pt,     % Margin inside the box (bottom)
innerrightmargin=10pt,      % Margin inside the box (right)
innerleftmargin=10pt        % Margin inside the box (left)
]{mybox2}

\mathchardef\mhyphen="2D

\allowdisplaybreaks
\newcommand{\ignore}[1]{}

\newcommand{\Tr}{\mbox{\rm Tr}}
\newcommand{\tr}[1]{\Tr\left(#1\right)}

\newcommand{\polylog}{\mbox{\rm polylog}}

\newcommand{\nrm}[1]{\left\lVert#1\right\rVert}

\def\01{\{0,1\}}

\newcommand{\bigO}{{\mathcal O}}
\newcommand{\bigOt}{{\widetilde{\mathcal O}}}
\newcommand{\pp}[1]{{(#1)}}

\DeclareExpandableDocumentCommand{\multiplexer}{O{}m}{|[multiplexor,#1]| {} \qw}
\def\mltplex#1{\controlslash{}\vqw{#1}}

\numberwithin{equation}{section}
\algrenewcommand\algorithmicensure{\hspace{3.2em}\textbf{Ensure:}}
\algrenewcommand\algorithmicindent{1.5em}
\algtext*{EndIf}   % hide only “EndIf”
\algtext*{EndFor}  % hide only “EndFor”

\begin{document}

\title{Quantum and classical algorithms for SOCP based on the multiplicative weights update method}
\author{M. Isabel Franco Garrido\thanks{Institute for Quantum Information and Matter, California Institute of Technology. mfrancog@caltech.edu } , \quad
Alexander M. Dalzell\thanks{AWS Center for Quantum Computing} , \quad Sam McArdle\footnotemark[2]}
% \affiliation{}
% \affiliation{}
% \author{}
% \affiliation{AWS Center for Quantum Computing, Pasadena, CA, USA}

\date{}
\maketitle

\begin{abstract}
We give classical and quantum algorithms for approximately solving second-order cone programs (SOCPs) based on the multiplicative weights (MW) update method. Our approach follows the MW framework previously applied to semidefinite programs (SDPs), of which SOCP is a special case. We show that the additional structure of SOCPs can be exploited to give better runtime with SOCP-specific algorithms. For an SOCP with $m$ linear constraints over $n$ variables partitioned into $r \leq n$ second-order cones, our quantum algorithm requires $\bigOt(\sqrt{r}\gamma^5 + \sqrt{m}\gamma^4)$ (coherent) queries to the underlying data defining the instance, where $\gamma$ is a scale-invariant  parameter proportional to the inverse precision. This nearly matches the complexity of solving linear programs (LPs), which are a less expressive subset of SOCP. It also outperforms (especially if $n \gg r$) the naive approach that applies existing SDP algorithms onto SOCPs, which has complexity $\bigOt(\gamma^{4}(n + \gamma \sqrt{n} + \sqrt{m}))$. Our classical algorithm for SOCP has complexity $\bigOt(n\gamma^4 + m \gamma^6)$ in the sample-and-query model.

\end{abstract}

\newpage

\tableofcontents

\newpage

\section{Introduction}

\paragraph{Motivation} Second-order cone programming (SOCP) \cite{alizadeh2003second} is a prominent optimization framework that extends linear (LP) and quadratic (QP) programming and can be seen as a subset of semidefinite programming (SDP). LPs and QPs lack the expressiveness to model the nonlinear cone constraints inherent to SOCP, making them inadequate for solving such problems. Although SDP-based methods can be used to solve SOCPs, the resulting computational complexity would be unnecessarily high. This underscores the need for specialized algorithms tailored specifically to SOCP, balancing computational efficiency with the structural complexity of the problem.

A particularly relevant approach to solving convex optimization problems is the multiplicative weights (MW) method, which has been successfully applied to LPs and SDPs in conjunction with both classical \cite{Kale07} and quantum \cite{brandão2019quantumsdpsolverslarge, vanApeldoorn2020, zerosumLP} algorithms. However, despite its effectiveness, the MW framework remains largely unexplored for SOCP. Developing MW-based algorithms for SOCPs could offer a promising alternative to interior-point methods, particularly in scenarios where the problem size is large and high accuracy is not required.

SOCP has found applications in diverse fields, including machine learning (e.g., training support vector machines \citep{Debnath2005}), computational finance (e.g., portfolio optimization \citep{Kolbert2010RobustPO,Gambeta2020}) and a number of engineering \cite{Lobo1998} areas, such as control \citep{Akbari2024,Papusha2015}, engineering design \citep{Fakhari2021,Chen2020,Kanno2018, AiminJiang2010} and electrical power systems optimization \citep{ 10016686, Kocuk2018, Jabr2006}. In general, the enhanced expressivity of SOCPs can be used to solve ``robust'' variants of simpler convex problems like least squares regression and LPs, that is, the setting where some input data is subject to some uncertainty or stochastic variation \cite{alizadeh2003second,Wang2022}. SOCPs can also be used as a subroutine for solving SDPs \citep{9682942,Ahmadi2019}.  Furthermore, in the area of quantum optimization, recent work has shown that the Quantum Max-Cut problem can be relaxed to SOCP such that an approximate solution to the Quantum Max-Cut problem can be obtained by rounding the optimal SOCP solution \citep{Huber:2024hth}. 

\paragraph{Our results}
This work explores classical and quantum MW-based methods for solving primal SOCPs. It is inspired by the MW-based quantum algorithms for solving SDPs \citep{ Brandao2017,vanapeldoorn19, brandão2019quantumsdpsolverslarge, vanApeldoorn2020} and for solving LPs (via zero-sum games) \citep{zerosumLP}. Specifically, the high-level approach in this paper can be understood as an adaptation of the ``primal oracle'' SDP solver of \cite{Gilyn2019,brandão2019quantumsdpsolverslarge} to the SOCP setting. This metastrategy is related to but differs slightly from the original Arora--Kale framework for SDP solving \cite{arora2012multiplicative, arora2016SDPsolverPublishedVersion}, a framework that has been directly used for quantum SDP algorithms separately in \cite{Brandao2017,vanApeldoorn2020,Gilyn2019}. Exploiting the Euclidean-Jordan-algebra structure of second-order cones, our specialization narrows the complexity gap between SOCPs and LPs.

Our main result is a quantum algorithm for solving an SOCP with $m$ linear constraints over $n$ variables partitioned into $r\leq n$ second-order cones with complexity $\bigOt(\sqrt{r}\gamma^5 +\sqrt{m} \gamma^4)$ queries to the underlying data defining the instance, where $\gamma$ is a scale-invariant inverse-precision parameter, similar to that which appears in related algorithms for SDP, such as  \cite{vanapeldoorn19}. In particular, if $R$ and $\tilde{R}$ are upper bounds on the trace of the primal and dual solutions, respectively, and $\epsilon$ is the additive precision up to which the objective function should be optimized, then $\gamma = R\tilde{R}/\epsilon$. We also give a classical algorithm in the sample-and-query access model \cite{tang2019quantumInspired} (which has been used to dequantize many quantum machine learning algorithms) for which the complexity scales as $\bigOt(n\gamma^4 + m\gamma^6)$. In the regime where $r = \Theta(n)$, the quantum algorithm offers a quadratic speedup, although larger speedups remain possible in the regime where $r \ll n$.

For a program with $m$ constraints over an $n$-dimensional primal optimization variable, the algorithm iteratively updates a sparse vector $\mathbf{y} \in \mathbb{R}^m$ with non-negative entries (initially $\mathbf{y} = \mathbf{0}$). This $\mathbf{y}$ implicitly defines a candidate solution $\mathbf{x}$ to the primal formulation of the program---in fact, $\mathbf{x}$ is a Gibbs distribution depending on the weights in $\mathbf{y}$, which is the key fact that connects the MW method to possible quantum advantage. The update to $\mathbf{y}$ at each iteration is determined by querying a ``violated constraint oracle'' that returns one of the $m$ constraints that is violated by this $\mathbf{x}$, if one exists. The remarkable implication of the MW framework is that only $\bigO(\log(n))$ iterations are needed to (implicitly) obtain a point $\mathbf{x}$ representing a solution to the program up to some fixed constant precision. Our task is then to give the best possible classical or quantum implementation of the violated constraint oracle, which contributes the dominant $\mathrm{poly}(n,m)$ factor to the overall complexity. A key innovation for quantum SDP solvers in this framework was the utilization of the quantum OR lemma \cite{vanapeldoorn19, brandão2019quantumsdpsolverslarge,harrow2017sequential} to separate the $n$ and $m$ dependence additively as $\bigOt(\sqrt{n}+\sqrt{m})$ for the violated constraint oracle. 

In applying this framework specifically to SOCP, we discover that a simpler (essentially classical) version of the quantum OR lemma is required. As a result, both our quantum algorithm and our classical algorithm obtain additive complexity: $\bigOt(\sqrt{r}+\sqrt{m})$ and $\bigOt(n+m)$, respectively (here assuming $\gamma = \bigO(1)$). This additive complexity replicates the complexity for quantumly and classically solving dense linear programs from \cite{zerosumLP}  (up to a factors of $\gamma$).

In these complexity statements, the additive term depending on $r$ (or $n$) derives from the complexity of the Gibbs sampling step. Indeed, in a sense, the MW framework for SDP can be re-interpreted as a reduction from SDP to the task of Gibbs sampling: to solve an SDP, one prepares the Gibbs state of a $n \times n$ Hamiltonian, which changes from iteration to iteration, and in each iteration, one estimates the expectation values of selected observables. By analogy, our MW procedure for second-order-cone programs (SOCPs) establishes an analogous reduction: 
solving an SOCP boils down to iteratively preparing an analogy of the Gibbs state for the Euclidean Jordan algebra of second-order cones. Unlike SDP, where the Gibbs state is inherently a mixed state, the inherent low-rank nature of the second-order cone constraint means that the relevant Gibbs state in our case is pure. In any case, the $\bigO(\sqrt{r})$ quantum complexity required to prepare the state represents a worst-case analysis, and could be reduced in specific cases where fast thermalization is possible.  This viewpoint suggests a pathway along which a larger quantum advantage might emerge.

\paragraph{Significance and commentary on the similarities and differences between LP, SOCP, and SDP} 

Our work provides clarification of the complexity of the MW approach to SOCP, and this can be understood in the context of similar approaches for LP and SDP. A key takeaway is that the structure of SOCP can be exploited so that the complexity of solving an SOCP with $r$ second-order cone constraints becomes almost as good as solving an LP with $r$ positivity constraints \cite{zerosumLP}, and much better than what would be obtained by pursuing a naive embedding of the SOCP into an SDP. 

Linear programs are the simplest kind of conic program, taking the form 
\begin{align} \text{(LP) } \qquad \qquad 
\begin{split}
    \max_{\mathbf{x} \in \mathbb{R}^n} & \;\;\mathbf{c}^\top \mathbf{x} \\
    \text{subject to } A \mathbf{x} &\leq \mathbf{b} \in \mathbb{R}^m\\
    \mathbf{x} &\in \mathcal{C}_n
\end{split}
\end{align}
where $\mathbf{c} \in \mathbb{R}^n$ encodes the objective function, matrix $A$ and vector $\mathbf{b}$ together encode $m$ linear constraints, and $\mathcal{C}_n$ is the positive orthant, the conic subset of $\mathbb{R}^n$ for which all coordinates are non-negative.  SDPs generalize LPs by promoting the optimization variable $\mathbf{x}$, objective vector $\mathbf{c}$, and each row $A_{j,:}$ of $A$ to be an $n \times n$ symmetric matrix instead of a length-$n$ vector. The vector inner product $\mathbf{c}^\top \mathbf{x}$ is generalized to the  Hilbert--Schmidt inner product $\Tr(\mathbf{c} \mathbf{x})$ and the positive orthant is generalized to the cone of semidefinite matrices. Thus, SDP captures more problems than LP, but it also works with more complex objects with $\Theta(n^2)$ degrees of freedom rather than $\Theta(n)$, leading to greater complexity. 

SOCPs sit in between LPs and SDPs and potentially offer advantages of both. The form of SOCP resembles LP; the optimization variables are still vectors $\mathbf{x}$ of length $n$, and there are $m$ linear inequality constraints. The only difference is that the conic positive orthant constraint is replaced with a second-order cone constraint $\mathbf{x} \in \mathcal{L}$, where $\mathcal{L}$ is a product of second-order cones $\mathcal{L} = \mathcal{L}^{\pp{0}} \times \mathcal{L}^\pp{1} \times \cdots \times \mathcal{L}^\pp{r-1}$. The size of cone with index $k$ is $n^\pp{k}$ (with $\sum_k n^\pp{k} = n$), and formally it is given by the set $\mathcal{L}^\pp{k} = \{\mathbf{u} \colon u_0^2 \geq u_1^2+u_2^2 + \cdots + u_{n^\pp{k}-1}^2\}$.  Since second-order cones of dimension 1 (or 2) are equivalent to the positive orthant of dimension 1 (or a rotated positive orthant of dimension 2), LP is recovered when  $n^\pp{k}\leq  2$ for all $k$. SOCP becomes interesting when some of the cones have $n^\pp{k} \geq 3$, in which case they are more expressive than LP. In particular, a point in the second-order cone of size $n^\pp{k}$ can be described by 2 non-negative numbers and a single ``direction''---a unit vector on the sphere of dimension $n^\pp{k}-2$. This contrasts with a point in the semidefinite cone of $n^\pp{k} \times n^\pp{k}$ matrices, which are described by $n^\pp{k}$ non-negative numbers (the eigenvalues of the matrix) and $n^\pp{k}$ directions (the eigenvectors of the matrix). Thus, there is an intuition that second-order cone constraints gain the high-dimensional directionality of semidefinite constraints while maintaining the low-rank essence of positivity constraints. This low-rankness is exploited in our quantum algorithm---the ``direction'' within a second-order cone can be compactly represented as a pure quantum state, whereas representing high-rank SDP variables as quantum states requires using mixed states, which are more complex to prepare. 

It is worth noting that it would be possible to explicitly reformulate SOCP as a special case of SDP by rewriting the second-order cone constraints as semidefinite constraints: specifically, for vector $\mathbf{x} \in \mathbb{R}^n$, one can define an $n \times n$ symmetric ``arrowhead'' matrix $\operatorname{Arw}(\mathbf{x})$ (see \cref{def:arrowhead_def}) for which $\operatorname{Arw}(\mathbf{x})$ is semidefinite if and only if $\mathbf{x} \in \mathcal{L}$. However, enforcing the structure of $\operatorname{Arw}(\mathbf{x})$ within the larger set of all semidefinite matrices would require adding an additional up to $\Theta(n^2)$ linear constraints (e.g., to force most of the matrix elements to be equal to 0), which may lead SDP-based methods to have unnecessarily high complexity for SOCP problems. As a result, the complexity of using the quantum SDP solver of \cite{vanapeldoorn19} directly on SOCP would be $\bigOt(\gamma^{4}(n + \gamma\sqrt{n} +\sqrt{m}))$. The natural low-rank nature of the SOCP is not utilized in this approach; it is valuable to examine algorithms that directly target SOCPs.

\paragraph{Comparison to current methods for SOCP} State-of-the-art approaches to solving SOCPs predominantly rely on interior-point methods (IPMs), which provide polynomial-time complexity and strong convergence guarantees, making them the preferred choice in solvers like MOSEK, Gurobi, and CPLEX. In particular, IPMs achieve a better precision scaling, with a runtime that depends polylogarithmically on the inverse precision target $1/\epsilon$. However, for large programs, IPMs can become intractable since the scaling of their complexity is superlinear in the problem size, for example, classical IPMs for LPs with $n$ variables and $\bigO(n)$ constraints have complexity scaling as $\bigO(n^{\omega})$, where $\omega < 2.37$ is the matrix multiplication exponent \cite{cohen2021interiorPointMethodMatrixMultiplicationTime}. IPMs for SOCPs have received less intensive optimization but rigorously proved complexity scales only slightly worse, as $\bigO(\sqrt{r}n^{\omega})$ \cite{monteiro2000SOCPIPM}. More recently, quantum algorithms based on IPMs \citep{ Kerenidis2021,Dalzell2023,augustino2024quantumcentralpathalgorithm,apers2024quantumspeedupslinearprogramming} have emerged, including for SOCP \cite{Kerenidis2021,Dalzell2023}, and under favorable conditions could offer asymptotic speedups (at most subquadratic in size, see \cite{Dalzell2023}) over their classical counterparts. 

In parallel, \cite{zheng2024primal} have proposed a MW-based primal–dual meta-algorithm for symmetric-cone programs (SCPs), which includes SOCPs and mixed programs possessing both second-order cone and semidefinite constraints---this serves as a complement to our direct primal-focused algorithm. Ref.~\cite{zheng2024primal} generalizes the Arora--Kale SDP framework to all SCPs and presents only a meta-algorithm, where SCPs are solved iteratively, with each iteration requiring execution of a simpler oracle. Ref.~\cite{zheng2024primal} supplies the oracle implementation only for a couple of example applications---Support Vector Machine (SVM) and Smallest Enclosing Sphere (SES) instances---in both cases, the runtime achieved is nearly linear in the size of the input data and is amenable to parallelization. Our work follows a complementary meta-strategy and, in contrast, provides implementation of the relevant oracles in general, yielding end-to-end complexity statements for general SOCPs. Additionally, our classical algorithm operates in a different access model (sample-and-query), a choice we make to ensure fair comparisons between our classical and quantum algorithms. As we show, in this access model,  sublinear classical complexity is achievable (due to the ability to sample).

\paragraph{Structure }In this paper, we first provide in \cref{background} the background on Second-Order Cone Programming (SOCP), ensuring the necessary self-contained understanding. There, we also specify the access model, whereby the quantum and classical algorithms access the underlying data defining the SOCP instance. Then, in \cref{sec:MWSOCP}, we provide the general multiplicative weights algorithm for SOCP, and prove its convergence. This framework is common to both the classical and quantum algorithms, relying only upon a subroutine that checks for violated constraints in the feasibility problem, called the ``violated constraint oracle.'' We describe an approach for implementing the ``violated constraint oracle'' as a two-step process responsible for the additive $\bigOt(\sqrt{r} + \sqrt{m})$ complexity scaling of the quantum algorithm: the first step is a ``cone index Gibbs sampling oracle'' \cref{oracle:cone_index_Gibbs_sampling}, which does the corresponding importance sampling, and the second step is a ``sampled violated constraint search oracle'' \cref{oracle:sampled_constraint_search}, which uses the importance samples to find approximately the violated constraints. In \cref{sec:quantum_implementation}, we provide the quantum implementations of these oracles and the overall complexity of the quantum algorithm, leveraging quantum primitives of quantum singular value transformation (QSVT), amplitude amplification, and Gibbs sampling. Finally, in \cref{section:classical_implementation_violated_constraint} we present the classical implementation of the main oracles in the sample-and-query access model.  

A flow chart depicting how SOCP is reduced and broken down into various quantum subroutines is provided in \cref{fig:flow_chart}. 

% Other classical: augmented Lagrangian method (ALM) to solve convex quadratic second-order cone programming (SOCP) https://arxiv.org/pdf/2010.08772

\section{Background}\label{background}

\subsection{Second-order cones and their Jordan algebra}

This section provides a self-contained collection of definitions of the objects that feature in our algorithm for second-order cone programs, and their key properties~\citep{alizadeh2003second, Kerenidis2021}. We take as input a positive integer $r$ denoting the number of second-order cones, and positive integers $n^\pp{0},\ldots,n^\pp{r-1}$ denoting the size of each cone.
% \footnote{In $n^\pp{k}$, the superscript $k$ indexes the cone number, it does not mean exponentiation.} 
Define $n := \sum_{k=0}^{r-1} n^\pp{k}$.

\begin{definition}[Second-order (Lorentz) cone]\label{def:second_order_cone} For each $k = 0,\ldots, r-1$, define the second-order (``Lorentz'') cone of size $n^\pp{k}$ as the following set:
$$\mathcal{L}^\pp{k} =  
\left\{\mathbf{v}^\pp{k}=\left.\left[\begin{array}{l}
v^\pp{k}_0 \\
\vec v^\pp{k}
\end{array}\right] \right\rvert\, v^\pp{k}_{0} \in \mathbb{R},  \vec v^\pp{k} \in \mathbb{R}^{n^\pp{k}-1},\|\vec v^\pp{k}\| \leqslant v^\pp{k}_0\right\} \subseteq \mathbb{R}^{n^\pp{k}} $$
where \( \| \cdot \| \) denotes the Euclidean norm. We refer to $v^\pp{k}_0$ and $\vec{v}^\pp{k}$ as the scalar and vector parts of $\mathbf{v}^\pp{k}$, respectively.
\end{definition}

We can then consider the Cartesian product of the $r$ second-order (Lorentz) cones
\begin{equation}
    \mathcal{L} = \mathcal{L}^\pp{0} \times \mathcal{L}^\pp{1} \times \cdots \times \mathcal{L}^\pp{r-1} \subset \mathbb{R}^n
\end{equation}
Given a vector $\mathbf{v} \in \mathbb{R}^n$, we may write 
\begin{equation}
    \mathbf{v} = \begin{bmatrix}
        \mathbf{v}^\pp{0} \\
        \mathbf{v}^\pp{1} \\
        \vdots \\
        \mathbf{v}^\pp{r-1}
    \end{bmatrix}
\end{equation}
where $\mathbf{v}^\pp{k} \in \mathbb{R}^{n^\pp{k}}$, with the superscript $k$ signalling that the vector $\mathbf{v}^\pp{k}$ is associated with the $k$th-cone in the Cartesian product. In general, we will also use the  superscript index $k$ to signal the cone number for matrices.
For a vector $\mathbf{v}^\pp{k}$, we use interchangeably $\mathbf{v}^\pp{k}\succeq \mathbf{0}$ and $\mathbf{v}^\pp{k} \in \mathcal{L}^\pp{k}$, to denote that $\mathbf{v}^\pp{k}$ is a vector in the second-order cone.
We see that $\mathbf{v} \in \mathcal{L}$, or $\mathbf{v}\succeq \mathbf{0}$, if and only if  $\mathbf{v}^\pp{k} \in \mathcal{L}^\pp{k}$ for $k=0,\ldots, r-1$.  

\begin{definition}[Arrowhead matrix]\label{def:arrowhead_def}
The Arrowhead matrix $\operatorname{Arw}(\mathbf{v}^\pp{k}) \in \mathbb{R}^{n^\pp{k} \times n^\pp{k}}$ associated to a vector $\mathbf{v}^\pp{k} \in \mathbb{R}^{n^\pp{k}}$ is defined as the following square matrix: 
\begin{equation}
    \operatorname{Arw}(\mathbf{v}^\pp{k}) = \left(
    \begin{array}{cc}
        v^\pp{k}_0    & {\vec v}^{\pp{k}\top}    \\
        \vec v^\pp{k} & v^\pp{k}_0 I  
    \end{array}
    \right)
\end{equation}
     where $I$ denotes the $(n^\pp{k}-1)\times(n^\pp{k}-1)$ identity matrix and $v_0^\pp{k}$, $\vec{v}^\pp{k}$ denote the scalar and vector part of $\mathbf{v}^\pp{k}$ as in \cref{def:second_order_cone}. 
    The name ``arrowhead'' is given due to the observation that the only nonzero entries lie on the diagonal or in the first row and column, resembling an arrow pointing toward the top left. We observe that $\operatorname{Arw}(\mathbf{v}^\pp{k})$ is positive semidefinite if and only if $\mathbf{v}^\pp{k} \in \mathcal{L}^{\pp{k}}$.
\end{definition}

When working with the Cartesian product of $r$ second-order cones, we define the arrowhead matrix of an $n$-dimensional vector as the direct product of the $r$ arrowhead matrices of its constituent parts, as follows. 
\begin{definition} [Arrowhead matrix for Cartesian product of cones] \label{arrw_multi_vector}Given a vector $\mathbf{v} \in \mathbb{R}^{n^\pp{0}} \times \cdots \times \mathbb{R}^{n^\pp{r-1}}$, the arrowhead matrix associated to $\mathbf{v}$ can be written as:
$$\operatorname{Arw}(\mathbf{v})
  = \bigoplus_{k=0}^{r-1} \operatorname{Arw}(\mathbf{v}^{(k)})$$
% \begin{equation}
% \qquad \operatorname{Arw}(\mathbf{v}) = {\scriptstyle \begin{pmatrix}\operatorname{Arw}(\mathbf{v}^\pp{0}) &                                 &   & \cdots & \\
%                                 & \operatorname{Arw}(\mathbf{v}^\pp{1}) & 0 & \cdots  & 0 \\
% 0                                & 0                                & \operatorname{Arw}(\mathbf{v}^\pp{2}) & \cdots & 0 \\
% \vdots & \vdots & \vdots & \ddots & \vdots \\
% 0 & 0 & 0 & \cdots &\operatorname{Arw}(\mathbf{v}^\pp{r-1}) 
%  \end{pmatrix} 
% }
% \end{equation}
We observe that $\operatorname{Arw}(\mathbf{v})$ is positive semidefinite if and only if $\mathbf{v} \in \mathcal{L}$.
\end{definition}

A key framework to study second-order cones is the Euclidean Jordan algebra associated with this class of cone. Before defining the central operation, we define the identity element for the algebra:

\begin{definition}[Identity element $\mathbf{e}^\pp{k}$ for cone $k$] When we consider the single-cone case, the identity element is $\mathbf{e}^\pp{k}:= (1, \vec{0})^\top \in \mathbb{R}^{n^\pp{k}}$, where $\vec{0}$ denotes the all-zeros vector of length $n^{\pp{k}}-1$. For the $r$-cone case, 
\begin{equation}
    \mathbf{e} := (\underbrace{1,\vec{0}}_{\mathbf{e}^\pp{0}},\underbrace{1,\vec{0}}_{\mathbf{e}^\pp{1}},\ldots,\underbrace{1,\vec{0}}_{\mathbf{e}^\pp{r-1}})^\top \in \mathbb{R}^n.
\end{equation}
\end{definition}

This notation is distinguished from the notation $\mathbf{e}_j$, by which we mean the standard basis vector with a 1 in position $j$ and a 0 in all other positions. 

\begin{definition}[Jordan (Circle) product $\circ$]\label{def:jordan_product}
The Jordan product is a commutative (but not associative), bilinear operation that performs the following operation on two vectors $\mathbf{v}^\pp{k} \in \mathbb{R}^{n^\pp{k}}$ and  $\mathbf{w}^\pp{k} \in \mathbb{R}^{n^\pp{k}}$:
\begin{equation}
    \mathbf{v}^\pp{k} \circ \mathbf{w}^\pp{k} := \left(\begin{array}{c}
\mathbf{v}^{\pp{k}\top} \mathbf{w}^\pp{k} \\
v^\pp{k}_0 \vec{w}^\pp{k} +w^\pp{k}_0 \vec{v}^\pp{k}
\end{array}\right)=\operatorname{Arw}(\mathbf{v}^\pp{k}) \mathbf{w}^\pp{k}=\operatorname{Arw}(\mathbf{v}^\pp{k}) \operatorname{Arw}(\mathbf{w}^\pp{k}) \mathbf{e}^\pp{k}
\end{equation}
where $\mathbf{e}^\pp{k} \in \mathbb{R}^{n^\pp{k}}$ is the identity element for the Euclidean Jordan algebra. The circle product can be extended to the multicone case: it acts cone-wise, and the relationship $\mathbf{v} \circ \mathbf{w} = \operatorname{Arw}(\mathbf{v})\mathbf{w}$ is preserved.

\end{definition}

The square matrix $\operatorname{Arw}(\mathbf{v}^\pp{k})$ has $n^\pp{k}$ eigenvalues and eigenvectors, of which two suffice to decompose $\mathbf{v}^\pp{k}$ as 
\begin{align}
\mathbf{v}^\pp{k} = \lambda_{+}(\mathbf{v}^\pp{k})\mathbf{c}_{+}(\mathbf{v}^\pp{k}) + \lambda_{-}(\mathbf{v}^\pp{k})\mathbf{c}_{-}(\mathbf{v}^\pp{k}),
\end{align}
where we define
\begin{equation}
\lambda_{\pm}(\mathbf{v}^\pp{k}) = v^\pp{k}_0 \pm \nrm{\vec v^\pp{k}} ; \quad \mathbf{c}_{\pm}(\mathbf{v}^\pp{k}) = \frac{1}{2}\begin{pmatrix}
1 \\
\pm \frac{\vec{v}^\pp{k}}{\|\vec{v}^\pp{k}\|}
\end{pmatrix}.
\end{equation}
We drop the argument and just write $\lambda_{\pm}$ and $\mathbf{c}_{\pm}$ when context is clear. 
This decomposition is usually referred to as the \textit{Jordan frame of $\mathbf{v}^\pp{k}$}.
The Jordan frame can be used to extend the definition of a real-valued continuous function to the Jordan algebra, as 
\begin{equation}\label{eq:EJA_continuous_function}
    f(\mathbf{v}^\pp{k}) := f(\lambda_+) \mathbf{c_+} + f(\lambda_-)\mathbf{c_-}
\end{equation} \label{func_jordan_algebra}
This defines exponentiation for a vector $\mathbf{v}^\pp{k} \in \mathbb{R}^{n^\pp{k}}$, in the Jordan algebra, as follows:
\begin{equation}\label{eq:Jordan_frame_exponential}
    e^{\mathbf{v}^\pp{k}}=e^{\lambda_+(\mathbf{v}^\pp{k})}\mathbf{c}_+(\mathbf{v}^\pp{k}) + e^{\lambda_{-}(\mathbf{v}^\pp{k})}\mathbf{c}_-(\mathbf{v}^\pp{k})
\end{equation}
% Expanding:
% \begin{align}\label{eq:exponentiation_vector_socp}
% e^{\mathbf{v}^\pp{k}}=&\left(\begin{array}{c}
% \frac{1}{2} e^{v_0^\pp{k}+||\vec v^\pp{k}||}  \\
% \frac{\vec v^\pp{k}}{2||\vec v^\pp{k}||} e^{v^\pp{k}_0+||\vec v^\pp{k}||} 
% \end{array}\right)+\left(\begin{array}{c}
%  \frac{1}{2} e^{v^\pp{k}_0-||\vec v^\pp{k}||} ) \\
%  \frac{-\vec v^\pp{k}}{2||\vec v^\pp{k}||}  e^{v^\pp{k}_0-||\vec v^\pp{k}||} 
% \end{array}\right)\\ =&\left(\begin{array}{c}
% e^{v^\pp{k}_0} (\frac{1}{2} e^{||\vec v^\pp{k}||} + \frac{1}{2} e^{-||\vec v^\pp{k}||} ) \\
% e^{v^\pp{k}_0} (\frac{\vec v^\pp{k}}{2||\vec v^\pp{k}||} e^{||\vec v^\pp{k}||} - \frac{\vec v^\pp{k}}{2||\vec v^\pp{k}||}  e^{||\vec v^\pp{k}||} )
% \end{array}\right)=\left(\begin{array}{c}
% e^{v^\pp{k}_0} \cosh (\|\vec v^\pp{k}\|) \\
% e^{v^\pp{k}_0 }\frac{\sinh (\|\vec v^\pp{k}\|)}{\|\vec v^\pp{k}\|} \vec v^\pp{k}
% \end{array}\right)
% \end{align}
    
% \begin{equation}\label{eq:exp_vector_divided_1st_element}
% \frac{e^{\mathbf{v}^\pp{k}}}{[e^{\mathbf{v}^\pp{k}}]_0} = 
%     \begin{pmatrix}
%     1 \\
%     \tanh(\nrm{\vec{v}^\pp{k}})\frac{\vec{v}^\pp{k}}{\nrm{\vec{v}^\pp{k}}}
%     \end{pmatrix}
% \end{equation}
\noindent We can also define the trace of a vector in the Jordan algebra: \begin{equation}
    \tr{\mathbf{v}^\pp{k}} := \lambda_+(\mathbf{v}^\pp{k}) + \lambda_-(\mathbf{v}^\pp{k}) = 2v^\pp{k}_0
\end{equation}
By the definition of the circle product, the trace of the circle product between two vectors evaluates to twice the inner product between the vectors.
\begin{equation}\label{eq:trace_circ_product}
\Tr(\mathbf{v}^\pp{k} \circ \mathbf{w}^\pp{k}) = 2\mathbf{v}^{\pp{k}\top} \mathbf{w}^\pp{k}
\end{equation}
Similarly, the trace of the exponential of a vector in the Jordan algebra:
\begin{equation}\label{eq:trace_exponential_vector}
    \tr{e^{\mathbf{v}^\pp{k}}} := \tr{e^{\lambda_+(\mathbf{v}^\pp{k})}\mathbf{c}_+(\mathbf{v}^\pp{k}) + e^{\lambda_{-}(\mathbf{v}^\pp{k})}\mathbf{c}_-(\mathbf{v}^\pp{k})} = e^{\lambda_+(\mathbf{v}^\pp{k})}+e^{\lambda_{-}(\mathbf{v}^\pp{k})}
\end{equation}
We can extend the above identities for the multicone case:
\begin{align}
    e^{\mathbf{v}}&=(e^{\mathbf{v}^\pp{0}};\ldots;e^{\mathbf{v}^\pp{r-1}})\\
    \tr{\mathbf{v}} &:= \sum^{r-1}_{k=0} 2v^\pp{k}_0\\
    \Tr(\mathbf{v} \circ \mathbf{w}) &= \sum^{r-1}_{k=0} 2\mathbf{v}^{\pp{k}\top} \mathbf{w}^\pp{k} = 2 \mathbf{v}^\top \mathbf{w}\\
    \tr{e^{\mathbf{v}}} &:= \sum^{r-1}_{k=0} e^{\lambda_+(\mathbf{v}^\pp{k})}+e^{\lambda_{-}(\mathbf{v}^\pp{k})}
\end{align}
Interpreting $\lambda_{\pm}$ as eigenvalues also suggests a definition for a norm $\lVert \cdot \rVert_{\rm soc}$ for vectors:
\begin{align}
    \nrm{\mathbf{v}^\pp{k}}_{\rm soc} = \nrm{\operatorname{Arw}(\mathbf{v}^\pp{k})} =  |v_0| + \nrm{\vec{v}} = \max(|\lambda_+|, |\lambda_-|)
\end{align}
and for multicone vectors, $\nrm{\mathbf{v}}_{\rm soc} = \max_k \nrm{\mathbf{v}^\pp{k}}$.  Note that $\nrm{\mathbf{v}^\pp{k}}_{\rm soc} \geq \nrm{\mathbf{v}^\pp{k}}$.

A remark on notation used widely across the document: we use $\log$ as the natural logarithm. Throughout, we index starting at 0, so, for example, the set $[r] = \{0,1,\ldots,r-1\}$. 

\subsection{Second-order cone programs}

Second-order cone programming solves a convex optimization problem (minimizing a convex cost function subject to a number of convex constraints) over the convex set defined by the Cartesian product of second-order (Lorentz) cones, subject to linear inequality constraints.
\begin{definition} [Second-order cone programs]\label{Def:SOCP}
    Let the number of cones, $r$, the number of constraints, $m$, and the length of each cone, $n^\pp{0},\ldots,n^\pp{r-1}$, be positive integers. Given as input vectors $\mathbf{c}^\pp{k} \in \mathbb{R}^{n^\pp{k}}$ and matrices $A^\pp{k} \in \mathbb{R}^{m \times n^\pp{k}}$ for $k = 0,\ldots, r-1$, as well as vector $ \mathbf{b} \in \mathbb{R}^{m}$, the primal formulation of the second-order cone program is
\begin{equation}
\label{primal} 
\begin{array}{ll}
\text{maximize} & \sum_{k=0}^{r-1} \mathbf{c}^{\pp{k} \top} \mathbf{x}^\pp{k} \\
\text{subject to} & \sum_{k=0}^{r-1} A^\pp{k} \mathbf{x}^\pp{k} \leq \mathbf{b}, \\
& \mathbf{x}^\pp{k} \in \mathcal{L}^\pp{k} \quad \forall k \in [r].
\end{array}
\end{equation}
where $\leq$ denotes element-wise comparison, and maximization is taken over vectors $\mathbf{x}^\pp{k} \in \mathbb{R}^{n^\pp{k}}$.  

The dual formulation can be written as\footnote{The corresponding SOCP duality theory can be found in standard optimization textbooks 
\citep{alizadeh2003second}.}:
\begin{equation}
\label{dual}
\begin{array}{ll}
\text{minimize} & \mathbf{b}^\top \mathbf{z}\\
\text{subject to} & A^{\pp{k} \top} \mathbf{z} - \mathbf{c}^\pp{k} \in \mathcal{L}^\pp{k} \quad \forall k \in [r] \\
& \mathbf{z} \geq \mathbf{0}
\end{array}
\end{equation}
where $\mathbf{0} = (0;0;\ldots;0) \in \mathbb{R}^{m}$, and minimization is over vectors $\mathbf{z} \in \mathbb{R}^{m}$. 
\end{definition}

\begin{definition}[Strong duality]
An SOCP as in \cref{Def:SOCP} is said to satisfy \textit{strong duality} if both the primal and dual programs are feasible (i.e., there exist points that satisfy all of the constraints), and furthermore the optimal objective value for the primal of \cref{primal} is equal to the optimal objective value for the dual of \cref{dual}.
\end{definition}
We will assume strong duality holds, which is true in most cases. For example, it can be shown that if the feasible set for the primal and the dual each have a nonempty \textit{interior}, then strong duality is true \cite{alizadeh2003second}.  

We assume that the inputs to the SOCP, as defined above, satisfy certain normalisation conditions, which are formalized in the following definition. Prior to that, we fix some notation used throughout. We write \([\cdot]_{j,:}\) to denote the \emph{j-th row} of a matrix, meaning the row obtained by fixing the row index to \(j\) and selecting all columns. Similarly, \([\cdot]_{:,j}\) denotes the \emph{j-th column}, obtained by fixing the column index to \(j\) and selecting all rows. This slice notation provides a concise and unambiguous way to refer to specific rows or columns, and will be used throughout to simplify expressions and improve readability.
\begin{definition}[normalisation conditions]\label{def:normalisation_conditions}
Given an SOCP as in \cref{Def:SOCP}, we say that it obeys the normalisation conditions if the following hold: 
    \begin{itemize}
    \item \textbf{Objective}: The vectors $\mathbf{c}^\pp{k}$ are normalized such that $\nrm{\mathbf{c}^\pp{k}}_{\rm soc} \leq 1$ for all $k$. If an SOCP does not satisfy this relation, the vectors $\mathbf{c}^\pp{k}$ can all be scaled down by a constant such that it is satisfied, without changing the feasible set or optimal point(s). 
    \item \textbf{Constraints}:
    % Each individual entry $A^\pp{k}_{ji}$ satisfies $|A^\pp{k}_{ji}|\leq 1$ for all $k \in [r]$, $j \in [m]$, $i \in [n^\pp{k}]$. 
    Each row of each constraint matrix satisfies $\nrm{A^\pp{k}_{j,:}}_{\rm soc} \leq 1$ for all $k \in [r]$ and $j \in [m]$. If this condition is violated for a certain $j$, the value $b_j$ and the row $A^\pp{k}_{j,:}$ can be scaled down (for all values of $k$) by a constant such that it is satisfied, without changing the feasible set or optimal point(s). 
\end{itemize}
\end{definition}
\begin{definition}[$R$-trace constrained]\label{def:R_trace}
        Given an SOCP as in \cref{Def:SOCP}, we say that it is $R$-trace constrained if the value $R$ is known and (i) there exists a feasible, optimal solution $(\mathbf{x}^\pp{0}; \allowbreak \mathbf{x}^\pp{1}; \allowbreak \ldots; \allowbreak \mathbf{x}^\pp{r-1})$
 to the SOCP that also satisfies $\sum_{k=0}^{r-1}\tr{\mathbf{x}^\pp{k}} \leq R$, and (ii) each entry $b_j$ of the vector $\mathbf{b}$ satisfies $|b_j| \leq R$.\footnote{\label{footnote:traceconditions} Note that item (ii) is essentially redundant with item (i). Observe that if item (i) is satisfied and the normalisation conditions are met, then  for the optimal solution $\mathbf{x}$ satisfying the trace bound we have for each $j$, $\sum_{k=0}^{r-1}A^\pp{k}_{j,:}\mathbf{x}^{\pp{k}}\leq \sum_{k=0}^{r-1} \nrm{A^\pp{k}_{j,:}} \nrm{\mathbf{x}} \leq \sum_{k=0}^{r-1} \nrm{A^\pp{k}_{j,:}}_{\rm soc} \nrm{\mathbf{x}}_{\rm soc} \leq \sum_{k=0}^{r-1} \nrm{\mathbf{x}}_{\rm soc} \leq  \sum_{k=0}^{r-1} \Tr(\mathbf{x}^\pp{k}) \leq R$, where the penultimate inequality holds since $\mathbf{x}^\pp{k} \succeq 0$.  Thus, if $|b_j| > R$, then we can replace $b_j$ with $R$ if $b_j>0$ or with  $-R$ if $b_j < 0$, ensuring item (ii) is met, without changing the feasibility of the point $\mathbf{x}$ and the validity of item (i).} 
\end{definition}
\begin{definition}[$\tilde{R}$-dual-trace constrained]\label{def:R_trace_dual}
        Given an SOCP as in \cref{Def:SOCP} with dual as in \cref{dual}, we say that it is $\tilde{R}$-dual-trace constrained if the value $\tilde{R}$ is known and there exists a feasible, optimal solution $\mathbf{z}$ to the dual formulation for which $\sum_{j=0}^{m-1} z_j \leq \tilde{R}. $ 
\end{definition}
For exact optimization of SOCP, it is always possible to find an equivalent SOCP and values of $R, \tilde{R}$, such that the SOCP satisfies the normalisation conditions and trace constraints.  However, we will be solving SOCPs approximately and the error can interplay with the normalisation, so it is important to first transform the SOCP instance so that it satisfies these conditions. 
In particular, we illustrate a couple examples of the interaction between the values $R$, $\tilde{R}$, the objective value, the normalisation conditions, and the ``scale invariant'' quantity $\gamma = R\tilde{R}/\epsilon$. First, if the vectors $\mathbf{c}^{\pp{k}}$ are all scaled down by a factor $F$ in order to meet the normalisation conditions, this leads the objective value (and hence the precision quantity $\epsilon$) to scale down by a factor of $F$ and the dual trace bound $\tilde{R}$ to scale down by a factor $F$, such that $\gamma$ is unchanged. On the other hand, if the constraint data $A^{(k)}$ and $\mathbf{b}$ is all scaled down by a factor $F$, the value of $\tilde{R}$ increases by a factor $F$, and so does $\gamma$. Ultimately, for fixed $r, m, n$, the complexity of our algorithms scales polynomially with $\gamma$. 

We follow prior work on SDPs and, instead of solving the general SOCP written above, our algorithm will instead only solve the question of feasibility for normalized SOCPs, promised that the SOCP is feasible or that it is $\theta$-far from feasible.

\begin{definition}[Unit trace Primal SOCP $\theta$-feasibility]\label{def:SOCP_feasibility}
Let $r$, $m$, $n^\pp{0},\ldots,n^\pp{r-1}$ be positive integers. Given as input an error parameter $\theta$, and matrices $A^\pp{k}$ for $k=0,\ldots,r-1$ satisfying the normalisation condition for constraints (\cref{def:normalisation_conditions}), as well as vector $\mathbf{b} \in \mathbb{R}^{m}$ satisfying $|b_j| \leq 1$ for all $j$, we define the set of $\theta$-feasible points $\mathcal{S}_{\theta}$ to contain all points $(\mathbf{x}^\pp{0}; \ldots;\mathbf{x}^\pp{r-1})$, for which
    \begin{equation}
    \begin{split}
\mathbf{x}^\pp{k} &\in \mathcal{L}^\pp{k} \quad \forall k \in [r]  \\
        \sum_{k=0}^{r-1} A^\pp{k} \mathbf{x}^\pp{k} &\leq \mathbf{b}+ \theta \mathbf{1} \\
       \sum_{k=0}^{r-1} \tr{\mathbf{x}^\pp{k}} &= 1\\
    \end{split}
    \end{equation}
    where $\mathbf{1} = (1;1;\ldots;1) \in \mathbb{R}^{m}$. Thus, $\mathcal{S}_{\theta}$ contains points with unit trace lying within the cones and which are at most $\theta$-far from satisfying each of the inequality constraints.  
    
    For a fixed value of $\theta$, suppose that we are promised that either
    \begin{enumerate}[(i)]
    \item $\mathcal{S}_{0}$ is nonempty (``feasible'')
    \item $\mathcal{S}_{\theta}$ is empty (``infeasible'')
    \end{enumerate}
    % \mathcal{S}_{0}$ is nonempty or that (ii) $\mathcal{S}_{\theta}$ is empty. 
    The $\theta$-approximate feasibility question is to determine whether (i) or (ii) is the case, and in the case of (i) to produce a vector $\mathbf{y} \in \mathbb{R}^m$ for which the set $\mathcal{S}_{\theta}$ contains the point\footnote{The semicolon represents row stacking.} 
    \begin{align}\label{eq:feasibility_output_primal_vector}
        \frac{1}{\sum_{k=0}^{r-1} \tr{e^{-A^{\pp{k}\top}\mathbf{y}}}}\left(e^{-A^{\pp{0}\top}\mathbf{y}}; e^{-A^{\pp{1}\top}\mathbf{y}}; \ldots ; e^{-A^{\pp{r-1}\top}\mathbf{y}} \right)
    \end{align} 
\end{definition}

It is well known that optimization of convex programs like SOCPs can be reduced to the feasibility question above, and prior work often omits this reduction explicitly from their presentation (see \cite{zheng2024primal} for an exception). We include the reduction here for completeness. 

\begin{lemma}[Reduction from general SOCP to primal feasibility problem]\label{feasibility_reduction}

Fix positive numbers $R$, $\tilde{R}$, and $\epsilon$. Let $\mathcal{P}$ be an SOCP with $r$ cones and $m$ constraints, defined by matrices $A^\pp{k}$ and vectors $\mathbf{b}, \mathbf{c}^\pp{k}$,  as in \cref{Def:SOCP}. Suppose that $\mathcal{P}$ satisfies the normalisation constraints and strong duality, and that it is $R$-trace and $\tilde{R}$-dual-trace constrained. 
% For any $g \in \mathbb{R}$, we may define an instance $\hat{\mathcal{P}}_g$ with $r$ cones and $m+1$ constraints of the feasibility SOCP.
Let $g^*$ denote the optimal value of $\mathcal{P}$, which satisfies the constraint $|g^*| \leq \min(R,\tilde{R})$, given the normalisation conditions and trace constraints.\footnote{This fact follows from the constraint on $\mathbf{c}^\pp{k}$, $R$-trace bound of $\mathbf{x}^\pp{k}$ and Cauchy--Schwarz via a similar calculation as in \cref{footnote:traceconditions}.
% : 
% $$\left\lVert \sum_{k=0}^{r-1} \mathbf{c}^\pp{k} \mathbf{x}^\pp{k}\right\rVert  \leq \sum_{k=0}^{r-1} \|\mathbf{c}^\pp{k} \| \|\mathbf{x}^\pp{k}\|\leq \sum_k \tr{\mathbf{x}^\pp{k}}$$
} Suppose that one has access to an oracle $\mathcal{O}_{\theta}$ that solves the $\theta$-feasibility question with probability at least 2/3 (when the promise is satisfied) for any SOCP $\hat{\mathcal{P}}$ formulated as in \cref{def:SOCP_feasibility}, where $\hat{\mathcal{P}}$ has $r+1$ cones (where the first $r$ cones are the same sizes as those of $\mathcal{P}$, and the final cone is of size 1) and $\hat{\mathcal{P}}$ has $m+1$ constraints.

Then, with $\bigOt(\log(R/\theta))$ calls to $\mathcal{O}_{\theta}$ with $\theta = \epsilon/(4R\tilde{R})$, one can determine a value $g$, and a vector $\mathbf{y}\in \mathbb{R}^{m+1}$
(satisfying $\mathbf{y} \geq \mathbf{0}$) such that, with probability at least $2/3$, the following are satisfied:
\begin{itemize}
    \item $g^* \in [g, g+\epsilon]$
    \item When we define the vectors:
    \begin{align}\label{eq:implicit_x_from_y_full_SOCP}
        \mathbf{x}^\pp{k}:= \frac{R\cdot 
        e^{-A^{\pp{k}\top} \mathbf{y}_{[0:m-1]}+\mathbf{c}^\pp{k}y_{m}}}{1+ \sum_{k=0}^{r-1} \tr{e^{-A^{\pp{k} \top}\mathbf{y}_{[0:m-1]}+\mathbf{c}^\pp{k} y_{m}}}}\,
    \end{align}
    We use \( \mathbf{y}_{[0:m-1]} \) to denote the vector \( (y_0, y_1, \dots, y_{m-1})^\top \).
    The vector $\mathbf{x} = (\mathbf{x}^\pp{0};\ldots;\mathbf{x}^\pp{r-1})$ achieves an objective value $\sum_{k=0}^{r-1} \mathbf{c}^{\pp{k}\top}\mathbf{x}^\pp{k} \geq g-\epsilon/4\tilde{R}$, and satisfies constraints \newline $\sum_{k=0}^{r-1} A^{\pp{k} \top} \mathbf{x}^\pp{k} \leq \mathbf{b} + (\epsilon/4\tilde{R}) \mathbf{1}$. 
\end{itemize}
\end{lemma}

\begin{proof}
We are given as input the SOCP $\mathcal{P}$, as in \cref{Def:SOCP}. It is specified by matrices $A^\pp{k}$ and vectors $\mathbf{b}, \mathbf{c}^\pp{k}$, and we assume that it is $R$-trace and $\tilde{R}$-dual-trace constrained and also that it satisfies the normalisation constraints. The proof idea begins by considering an appropriate normalisation of the inputs and incorporating the objective vectors \( \mathbf{c}^\pp{k} \) and the guess \( g \) into the matrices \( A^\pp{k} \) and vector \( \mathbf{b} \), respectively. We then adjust the guess \( g \) to test for feasibility: if it is sufficiently low, a feasible point exists. This procedure enables a binary search to identify the optimal value and a corresponding feasible solution. 

Let $g^* \in \mathbb{R}$ be the optimal (unknown) value of $\mathcal{P}$, which satisfies $|g^*|\leq R$. Let $g \in [-R,R]$ be a tunable guess for $g^*$. We can define new variables that act as inputs for an instance denoted $\hat{P}_g$ of the feasibility problem of \cref{def:SOCP_feasibility}:\begin{align}\label{eq:redefinition_A}
\hat{A}^\pp{k} &:= (   A^\pp{k}_{0,:} \quad ;\quad A^\pp{k}_{1,:} \quad;\quad  \ldots \quad ; \quad A^\pp{k}_{m-1,:} \quad;\quad - \mathbf{c}^\pp{k}) \in \mathbb{R}^{(m+1) \times n^\pp{k}} \text{ for } k =0,\ldots,r-1 \\
\hat{A}^\pp{r} &:= \left(0;0;\cdots; 0\right)\in \mathbb{R}^{(m+1) \times 1}\\
\hat{\mathbf{b}} &:= \left(\frac{b_0}{R};\ldots;\frac{b_{m-1}}{R}; -\frac{g}{R}\right) \in \mathbb{R}^{m+1}. \label{eq:redefinition_b} 
\end{align}

Since we have assumed that the normalisation conditions (\cref{def:normalisation_conditions}) are in place, it should be noted that $\hat{A}^\pp{k}_{m,:}=-\mathbf{c}^\pp{k}$ satisfies $\nrm{\hat{A}^\pp{k}_{m,:}}_{\rm soc} \leq 1$ for all $k$.  Furthermore, since $\mathcal{P}$ is $R$-trace constrained and $g \in [-R,R]$, we have $|b_j| \leq R$ and hence $|\hat{b}_j|\leq 1$. Thus, $\hat{\mathcal{P}}_g$ can be taken as a valid instance of the SOCP feasibility problem defined in \cref{def:SOCP_feasibility}, with  $\hat{m}:= m+1$ constraints and $\hat{r} = r+1$ cones, where cone $k$ has size $n^\pp{k}$ for $k=0,\ldots,r-1$, and cone $r$ has size 1. The additional cone is introduced to enforce that the unity-trace condition is satisfied.

The instance $\hat{\mathcal{P}}_g$ can be fed as input to the oracle  $\bigO_{\theta}$ for any $g \in [-R,R]$. Let $\mathbf{x}^* = (\mathbf{x}^{*\pp{0}};\ldots; \mathbf{x}^{*\pp{r-1}})$ denote the optimal feasible point of $\mathcal{P}$ that achieves $\sum_{k=0}^{r-1}\mathbf{c}^{\pp{k}\top} \mathbf{x}^{*\pp{k}}= g^*$, while satisfying $\sum_{k=0}^{r-1} A^\pp{k} \mathbf{x}^{*\pp{k}}\leq \mathbf{b}$. Define
\begin{align}
    \hat{\mathbf{x}}^* = \left(\frac{\mathbf{x}^{*\pp{0}}}{R};\frac{\mathbf{x}^{*\pp{1}}}{R};\ldots;\frac{\mathbf{x}^{*(r-1)}}{R}; 1-\frac{\sum_{k=0}^{r-1} \Tr(\mathbf{x}^{*\pp{k}})}{R}\right) \in \mathbb{R}^{n^\pp{0}}\times \mathbb{R}^{n^\pp{1}} \times \cdots \times \mathbb{R}^{n^\pp{k-1}} \times \mathbb{R}^1
\end{align}

First, by inspection we observe that $g^* \geq g$ implies that $\sum_{k=0}^{r} \hat{A}^\pp{k} \hat{\mathbf{x}}^{*\pp{k}} \leq \hat{\mathbf{b}}$, and hence that the set $\mathcal{S}_0$ (as defined in \cref{def:SOCP_feasibility}) for the SOCP $\hat{\mathcal{P}}_g$ contains the point $\hat{\mathbf{x}}^*$. 

Next, we consider the case that $g^* < g$. Consider the modification of $\mathcal{P}$ where $\mathbf{b}$ is replaced by $\mathbf{b}+R\theta \mathbf{1}$, and compute its dual as in \cref{dual}: $\min \mathbf{b}^\top \mathbf{z} + R\theta\mathbf{1}^{\top} \mathbf{z}$ , subject to $A^{\pp{k}\top} \mathbf{z} - \mathbf{c}^\pp{k} \in \mathcal{L}^\pp{k}, \mathbf{z} \geq \mathbf{0}$. The fact that the unmodified SOCP $\mathcal{P}$ is $\tilde{R}$-dual-trace constrained implies that $\mathcal{P}$ has an optimal dual solution $\mathbf{z}^*\geq \mathbf{0}$ satisfying $\mathbf{1}^\top \mathbf{z}^* \leq \tilde{R}$. As a consequence, replacing $\mathbf{b}$ with $\mathbf{b}+R\theta \mathbf{1}$ can cause the optimal objective value of the dual to increase by at most $R\tilde{R}\theta$. By strong duality, the optimal value of the modified primal is equal to the optimal value of the modified dual. We conclude that if in fact $g^* < g-R\tilde{R}\theta$, then the optimal value of the modified primal program will still be less than $g$, and there is no point $\mathbf{x} \in \mathcal{L}$ for which $\sum_{k=0}^{r} \hat{A}^\pp{k} \hat{\mathbf{x}}^\pp{k} \leq \hat{\mathbf{b}}+\theta\mathbf{1}$. This implies that $\mathcal{S}_\theta = \varnothing$. 

In summary, we have shown
\begin{align}
    g \leq g^* &\implies \mathcal{S}_0 \neq \varnothing  \qquad \text{case (i) of \cref{def:SOCP_feasibility}} \\
    g > g^* + R\tilde{R}\theta  &\implies \mathcal{S}_\theta = \varnothing \qquad \text{case (ii) of \cref{def:SOCP_feasibility}}
\end{align}
These two statements enable a binary search of the interval $[-R,R]$ for the value of $g^*$. Given an interval $[a,b]$ (initially with $a=-R$ and $b=R$), we can choose $g=(a+b)/2$ to be the midpoint of the interval and run the oracle $\mathcal{O}_{\theta}$ on $\hat{P}_g$. If $g^* \geq g$, then the oracle will output ``feasible'' with probability at least 2/3 and if $g^* \leq g-R\tilde{R}\theta$, it will output ``infeasible'' with probability at least 2/3. We may boost this probability to $1-\zeta$ by repeating the oracle call $\mathcal{O}(\log(1/\zeta))$ times and taking the majority output. If $g^* \in [g-R\tilde{R}\theta,g]$, then we have no guarantees on the output of the oracle. Thus, if we obtain the output ``feasible'' we can update the search interval from $[a,b]$ to $[g-R\tilde{R}\theta, b]$, and if we obtain ``infeasible'' we can update the search interval to $[a,g]$---as long as the oracle call succeeded, we will not have eliminated $g^*$ from the search interval. Each step cuts off nearly half the size of the search interval (binary search). In particular, if the interval has size at least $4R\tilde{R}\theta$, then each iteration reduces the interval size by a factor between 1/2 and 3/4. Since we aim to reduce the search interval to a length of at most \( 4 R \tilde{R} \theta \), we determine the required number of iterations as follows. After \( T_{\rm bs} \) iterations, the interval length is at most \( 2R(3/4)^{T_{\rm bs}}\). 
Therefore, \(T_{\rm bs} =  \log_{4/3}\left( \frac{1}{2\tilde{R}\theta} \right) \) iterations are sufficient.

By outputting the value $g$ to be the lower point of the final search interval after $T_{\rm bs}$ iterations, we guarantee that the output $g$ satisfies $g^* \in [g, g + 4 R\tilde{R} \theta]$. By choosing 
\begin{equation}
    \theta = \frac{\epsilon}{4R \tilde{R}}
\end{equation}
we can achieve precision $\epsilon$ on the objective value, as stated in the theorem statement. We take $\zeta = 1/(3 T_{\rm bs})$---meaning that we need $\bigO(\log(1/\zeta)) = \bigO(\log(\log(1/\tilde{R}\theta)))$ repetitions of the oracle at each binary search step---which ensures all steps succeed and the overall error probability is bounded by $1/3$. 

Furthermore, once we have determined the output $g\in [g^*-4R\tilde{R}\theta, g^*]$, we may complete the procedure by running the oracle $\bigO_\theta$ one final time on the SOCP $\hat{\mathcal{P}}_g$. For this value of $g$, the analysis above guarantees that $\mathcal{S}_0 \neq \varnothing$. Thus, the oracle produces a $\mathbf{y} \in \mathbb{R}^{\hat{m}} = \mathbb{R}^{m+1}$ that implicitly generates a vector $\hat{\mathbf{x}}$ via \cref{eq:feasibility_output_primal_vector}, given here by
\begin{equation}
    \hat{\mathbf{x}} = \frac{1}{\sum_{k=0}^{r} \tr{e^{-\hat{A}^{\pp{k}\top}\mathbf{y}}}}\left(e^{-\hat{A}^{\pp{0}\top}\mathbf{y}}; e^{-\hat{A}^{\pp{1}\top}\mathbf{y}}; \ldots ; e^{-\hat{A}^{\pp{r}\top}\mathbf{y}} \right)
\end{equation}
which, by noting that $\hat{A}^\pp{r}$ is 0 and utilizing the structure of the other $\hat{A}^\pp{k}$, can be rewritten as
\begin{equation}
    \hat{\mathbf{x}} = \frac{1}{1+\sum_{k=0}^{r-1} \tr{e^{-A^{\pp{k}\top}\mathbf{y}_{[0:m-1]} + \mathbf{c}^\pp{k} y_{m}}}}\left(e^{-A^{\pp{1}\top}\mathbf{y}_{[0:m-1]} + \mathbf{c}^\pp{1} y_{m}}; \ldots ; e^{-A^{\pp{r-1}\top}\mathbf{y}_{[0:m-1]} + \mathbf{c}^\pp{r-1} y_{m}} ; 1\right)
\end{equation}
This $\hat{\mathbf{x}}$ lies in $\mathcal{S}_\theta$ and thus
satisfies  $\sum_{k=0}^{r} \hat{A}^\pp{k} \hat{\mathbf{x}}^\pp{k} \leq \hat{\mathbf{b}} + \theta \mathbf{1}$. The same $\mathbf{y}$ can be used to implicitly generate $\mathbf{x} = (R\hat{\mathbf{x}}^\pp{0}; \ldots; R\hat{\mathbf{x}}^\pp{r-1})$ satisfying $ \sum_{k=0}^{r-1} A^\pp{k} \mathbf{x}^\pp{k} \leq \mathbf{b} + R\theta \mathbf{1}$ and $\sum_{k=0}^{r-1}\mathbf{c}^{\pp{k}\top} \mathbf{x}^\pp{k} \geq g-R\theta$, verifying the theorem statement. 

We conclude that the total number of queries to $\mathcal{O}_{\theta}$ is given by $$\bigO(T_{\rm bs}\log(1/\zeta)) = \bigOt(\log(R/\theta)).$$

\end{proof}

\subsection{Access model}\label{section:access_model}

\subsubsection{Quantum registers and notation}
In the presentation so far, we allowed the $r$ cones to be of varying sizes $n^\pp{0}, \ldots, n^\pp{r-1}$ with $n = \sum_k n^\pp{k}$. To organize the notation for the quantum implementation, we assume all cone sizes are equal to $\bar n = \max_k  n^\pp{k} $ in the following quantum access oracles. This is achieved by appropriately padding the vectors $\mathbf{x}^\pp{k}$ and $\mathbf{c}^\pp{k}$, as well as the matrices $A^\pp{k}$, with zeros for all $k$ where $n^\pp{k}<\bar n$. Hence the total number of variables of the (padded) SOCP is $r\bar n$, where $r$ is the number of cones. This assumption can be taken without loss of generality since, as will show, the final query complexity of our quantum algorithm scales with $r$ and $\bar{n}$ as $\bigO(\sqrt{r}) \cdot \mathrm{polylog}(r,\bar{n})$ and thus the padding does not affect the core polynomial scaling of the algorithm. The padding may lead to a mild increase in the space requirement by at most a constant factor. 

The quantum algorithm acts on a set of quantum registers. The computational basis states of these registers index parts of the input data. Namely, we consider a computational basis state of the form 
\begin{equation}
    \underbrace{\ket{j}_{\rm row}\ket{k}_{\rm cone} \ket{i}_{\rm col}}_{\text{for indexing entry } A^\pp{k}_{ji}} \;\; \ket{f}_{\rm flag} \ket{h}_{\rm sampindex}\ket{a}_{\rm anc}
\end{equation}
where the first register holds the row index ($j$), the second register the cone index ($k$), and the third register the column index within the cone ($i$), relevant for indexing the input value $A^\pp{k}_{ji}$. These registers must contain at least $\log_2(m)$, $\log_2(r)$, and $\log_2(\bar{n})$ qubits, respectively, in order for the number of computational basis states to be greater than the number of index values in each case.
  There is also a single-qubit ``flag'' register holding $\ket{f}$, where $f=1$ corresponds to indication of a failure of some kind. One subroutine utilizes an additional register we call the sample index register of $\lceil \log_2(T')\rceil $ qubits, where $T'$ is an integer to be specified later (``number of samples''). 
Some subroutines also utilize additional ancilla registers storing $\ket{a}$, where the size can vary. Above, the subscripts are included for convenience but generally we leave them off except where helpful for clarity.  

For any of these registers, we use $\ket{\bar 0}$ to describe the state of a multiqubit register where all the qubits are initialised to $\ket{0}$. The precise number of qubits can vary and can be inferred from context. 
% of polynomial order, that will likely vary for each use.
The state $\ket{ \rm garbage}$ refers to an unspecified pure state that is not useful for the intended computational or informational purpose. Similar to $\ket{\bar{0}}$, the number of qubits in $\ket{\rm garbage}$ varies and can be computed given the context. \newline

\subsubsection{Quantum access model}

Given an instance of the unit-trace feasibility problem of \cref{def:SOCP_feasibility}, defined by input data $A^\pp{0},\ldots,A^\pp{r-1}, \mathbf{b}$, we assume the quantum algorithm has access to the data in $A^\pp{k}$ and $\mathbf{b}$ by the following oracles.

\begin{oracle}[Row-prep oracle $O_R$]\label{O_R}
The quantum algorithm can access the data in the matrices $A^\pp{k}$ through a \textit{row-prep oracle} $O_R$, which prepares a quantum state encoding a given row of the matrix $A^\pp{k}$, controlled on the row and cone as follows
  \begin{align}
    O_R:\ket{j}_{\rm row}\ket{k}_{\rm cone}\ket{\overline{0}}_{\rm col} \ket{0}_{\rm flag} \mapsto\ket{j}_{\rm row }\ket{k}_{\rm cone}\Bigg( & \sum^{\bar n-1}_{i=0} A^{\pp{k} }_{ji}\ket{i}_{\rm col}\ket{0}_{\rm flag}+  \\
    &\sum^{\bar n-1}_{i=0}\sqrt{1-|A^{\pp{k} }_{ji}|^2}\ket{i}_{\rm col}\ket{1}_{\rm flag}\Bigg)\notag
\end{align}
This is well defined since we have assumed that the matrices $A^\pp{k}$ obey the normalisation conditions (\cref{def:normalisation_conditions}), and thus $\nrm{A^\pp{k}_{j,:}} \leq 1$. We assume the ability to implement both the gate and its adjoint. Each of them is also equipped with an additional control on an ancilla qubit.
\end{oracle}

A generalization of this access model might re-define $O_R$ in terms of a constant $\alpha \geq 1$ as:
\begin{align}
    O_R:\ket{j,k,\overline{0},0}\mapsto\ket{j,k}\left(\frac{1}{\alpha}\sum^{\bar n-1}_{i=0} A^{\pp{k}}_{ji}\ket{i}\ket{0}+\sum^{\bar n-1}_{i=0} \sqrt{1- \frac{1}{\alpha^2}|A^{\pp{k}}_{ji}|^2}\ket{i}\ket{1}\right)
\end{align}
Here, $\alpha$ would be analogous to the normalisation factor one uses when working with block-encodings, and it must satisfy $\alpha \geq \nrm{A^\pp{k}_{j,:}}$ for all $j,k$. Since we have assumed that the matrices $A^\pp{k}$ satisfy the normalisation conditions, we may take $\alpha=1$. If the input data was not normalised or normalising wasn't an option, one can follow the analysis of the algorithm carrying the $\alpha$ coefficient. The coefficient \(\alpha\) would appear in the normalisation of the block-encodings of the considered Arrowhead matrices and, consequently, would impact the algorithm's complexity.

\begin{oracle}[$O_\mathbf{b}$ oracle for $\mathbf{b}$ vector]\label{U_b}
The quantum algorithm can access the entries of the vector $\mathbf{b}$ through an oracle that encodes a given entry into the amplitude of the quantum state, controlled on the row register:
\begin{align}
O_\mathbf{b}: \ket{j}_{\rm row}\ket{0}_{\rm flag} \mapsto b_j\ket{j}_{\rm row} \ket{0}_{\rm flag}+\sqrt{1-|b_j|^2}\ket{j}_{\rm row} \ket{1}_{\rm flag}
\end{align}
This is well defined since we have assumed that the vector $\mathbf{b}$ obeys the condition $|b_j|\leq 1$ for all $j$ (\cref{def:SOCP_feasibility}). We assume the ability to implement both the gate and its adjoint. Each of them is also equipped with an additional control on an ancilla qubit.
\end{oracle}

Our complexity analysis focuses on the number of queries made to oracles. The oracle $O_R$ is analogous to an access model for instances of SDPs, where on input $j$, the oracle produces a block-encoding of the $j$-th constraint matrix---here, the constraints are encoded in the row vectors $A^\pp{k}_{j,:}$ rather than as matrices, so it makes sense to assume an oracle that can prepare quantum states encoding these vectors. The oracle $O_R$ can be implemented, for example, as a quantum circuit with $\polylog(mn)$ circuit depth using the quantum accessible data structure studied in \cite{Kerenidis2016QuantumSystems, Chakraborty2018TheSimulation}. 

The oracle $O_{\mathbf{b}}$ can be implemented as a quantum circuit with $\polylog_2(m)$ depth by loading the binary representation of $b_j$ into an ancilla register (e.g., with a log-depth circuit for quantum random access memory (QRAM)), using the ancilla register to control a rotation of the flag qubit, and then unloading the binary representation of $b_j$ to restore the ancilla register. Thus, in a cost model where circuit depth is the relevant metric, these access oracles can be implemented cheaply. This is essentially equivalent to an assumption of cheap QRAM \cite{jaques2023qramsurveycritique,dalzell2025distillation}. 

Our quantum algorithm will also involve a step where it computes some intermediate classical data it stores in a classical database, and then later accesses this dataset coherently. We require two versions of this, which are related. First, for a classically stored vector $\mathbf{y} \in \mathbb{R}^m$, we require the ability to prepare a quantum state encoding the vector into its amplitudes, via the oracle $O_{\mathbf{y}}$. Second, for a list of $T'$ cone indices $\mathcal{T} = (k_0,k_1,\ldots,k_{T'-1})$, we require the ability to prepare an equal superposition over $\ket{k_h}\ket{h}$ for $h = 0,\ldots, T'-1$, accomplished by the oracle $O_{\mathcal{T}}$. 

\begin{oracle}[State-prep oracle for classical data]\label{O_y} Let $\mathbf{y} = (y_0,\ldots,y_{m-1})^\top \in \mathbb{R}^m$ be a vector for which $y_j \geq 0$ for all $j$. The state-prep oracle $O_\mathbf{y}$ prepares a state encoding the entries of $\mathbf{y}$ into its amplitudes
\begin{align}
O_\mathbf{y}\colon \ket{\overline{0}}_{\rm row} \mapsto \frac{1}{\sqrt{\nrm{\mathbf{y}}_1}}\sum_{j=0}^{m-1} \sqrt{y_j}\ket{j}_{\rm row}
\end{align}
Similarly, for an integer $T'$, let $\mathcal{T} = (k_0,k_1,\ldots,k_{T'-1})$ where each $k_h$ is an integer in $[r]$. The state prep oracle $O_{\mathcal{T}}$ prepares a state
\begin{align}
    O_{\mathcal{T}}\colon     \ket{\bar{0}}_{\mathrm{cone}}\ket{\bar{0}}_{\mathrm{sampindex}} \mapsto \frac{1}{\sqrt{T'}}\sum_{h=0}^{T'-1}\ket{k_h}_{\mathrm{cone}} \ket{h}_{\mathrm{sampindex}}
\end{align}
\end{oracle}
We note that both of these oracles could be viewed under the same framework by thinking of $O_{\mathcal{T}}$ as preparing the state associated with the vector in $\mathbb{R}^{T'r}$ with a 1 in $T'$ of the entries. We also assume access to the corresponding adjoint. Both $\mathbf{y}$ and $O_{\mathcal{T}}$ can be implemented as quantum circuits with depth $\mathrm{polylog}(mn)$, although the total circuit size is at least $s$ and $T'$, respectively, where $s$ is the number of nonzero entries of $\mathbf{y}$.

% Second, for a classical databases of numbers $\mathbf{\lambda} \in \mathbb{R}^r$ and $\mathbf{\mathcal{T}} \in [r]^{T'}$, we assume the ability to read the numbers into an ancilla register via $O_{\mathbf{\lambda}}$ and $O_{\mathbf{\mathcal{T}}}$.
% \begin{oracle}
%     Given $\mathbf{\lambda} \in \mathbb{R}^r$, where the entries of $\lambda$ are stored to some finite number of bits of precision,  and $\mathbf{\mathcal{T}} \in [r]^{T'}$, define
%     \begin{align}
%         O_{\mathbf{\lambda}}\colon& \ket{k}_{\rm cone} \ket{\bar{0}}_{\rm anc} \mapsto \ket{k}_{\rm cone}  \ket{\lambda_k}_{\rm anc}\\
%         O_{\mathbf{\mathcal{T}}}\colon& \ket{0}_{\rm cone} \ket{h}_{\rm anc} \mapsto \ket{k_h}_{\rm cone}  \ket{h}_{\rm anc}
%     \end{align}
% \end{oracle}

\subsubsection{Classical access model}

We assume that our classical algorithm has an analogous kind of access to the data defining the SOCP instance. Namely, we assume the sample-and-query access model from \cite{tang2019quantumInspired, Chia20}. 

\begin{oracle}
    Given inputs $A^\pp{0},\ldots,A^\pp{r-1}$ and $\mathbf{b}$, we assume we have the ability to perform the following operations.
    \begin{itemize}
        \item Query access: for any $j,k,i$ we can query the value $A^\pp{k}_{ji}$, or the value $b_j$. 
        \item Sample access: for any $j,k$, we can sample a value $i$ with probability equal to $\frac{|A^\pp{k}_{ji}|^2}{\nrm{A^\pp{k}_{j,:}}^2}$. 
        \item Norm access: given $j,k$, we can query the value of the norm $\nrm{A^\pp{k}_{j,:}}$. 
    \end{itemize}
\end{oracle}

\section{Multiplicative Weights approach to SOCP} \label{sec:MWSOCP}

We will now discuss our proposed algorithm to solve the unit-trace SOCP $\theta$-feasibility problem, as described in \cref{def:SOCP_feasibility}, where there are $m$ constraints and $r$ cones of lengths $n^\pp{0}, n^\pp{1},\ldots, n^\pp{r-1}$. The framework of this section applies to both the quantum and the classical algorithm, which diverge only on implementation of the subroutines introduced here.

In \cref{fig:flow_chart}, we provide a flow chart organizing the layers of abstraction in our analysis, where the full SOCP is first reduced to the feasibility SOCP, and then the violated constraint oracle. The rest of the flow chart depicts how the violated constraint oracle is decomposed further into (quantum) subroutines and eventually into queries to the data access oracles of \cref{section:access_model}. 

\begin{figure}[h]
    \centering
    \includegraphics[draft=false, width=0.97\linewidth]{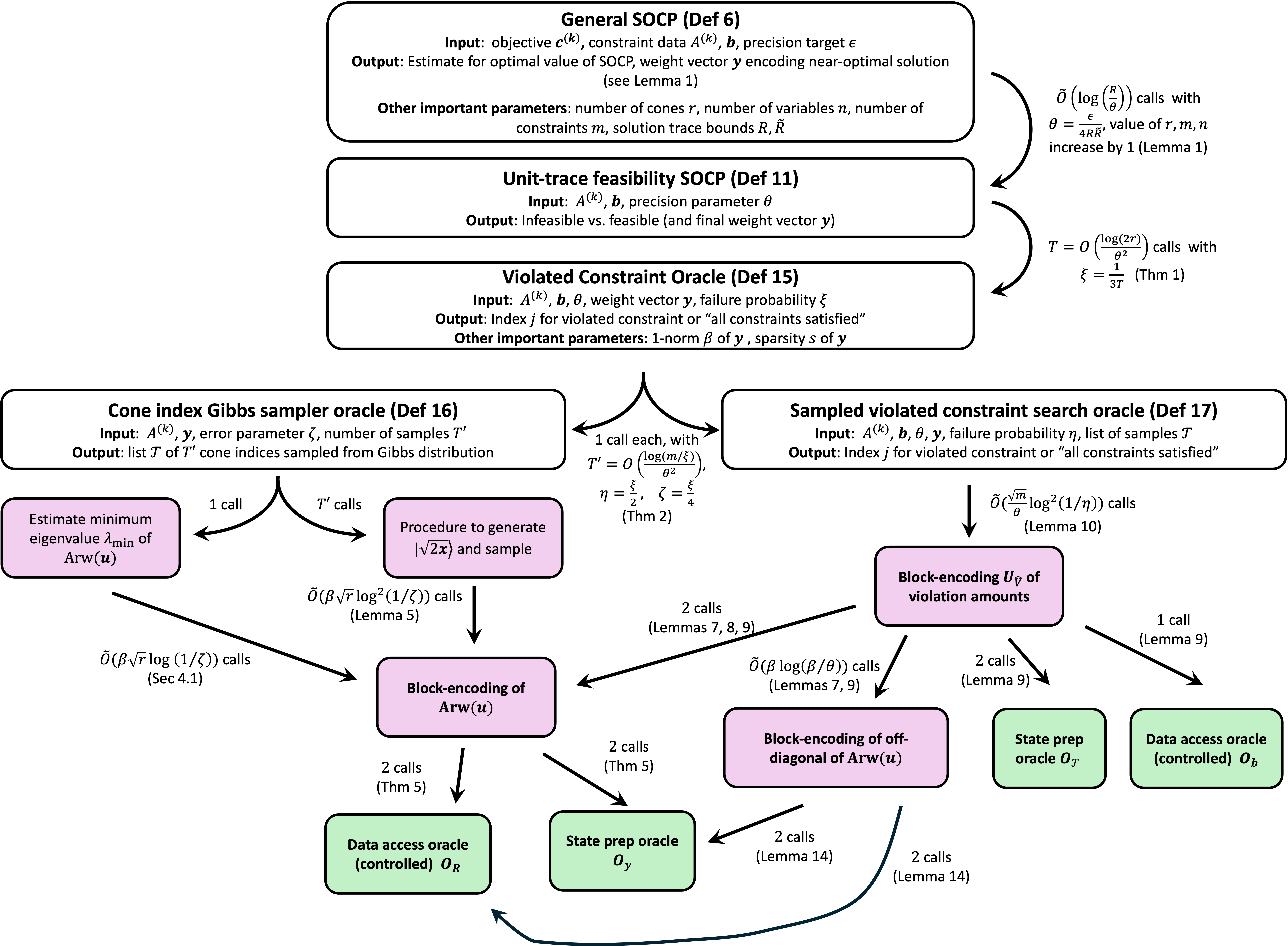}
    \caption{Summary of the subroutines in our analysis that are used to approximately solve an SOCP. White boxes represent subroutines which are agnostic to classical vs.~quantum, purple boxes represent subroutines specific to the quantum implementation, and green boxes represent the quantum data access oracles. }
    \label{fig:flow_chart}
\end{figure}

\subsection{Violated constraint oracle}
The main subroutine of the algorithm is “violated constraint oracle'', which, given an implicit representation of a point $\mathbf{x} \in \mathbb{R}^n$, either finds a constraint index $j \in [m]$ corresponding to a violated constraint or else returns that all constraints are satisfied. As stated in \cref{def:SOCP_feasibility},  we must be able to distinguish two scenarios: either $\mathcal{S}_0 \neq \varnothing$, or $\mathcal{S}_\theta = \varnothing$. The violated constraint oracle, roughly speaking, will test whether $\mathbf{x} \in \mathcal{S}_0$, or else find an index $j$ associated with a constraint that $\mathbf{x}$ violates by at least $\Omega(\theta)$. 

Core to this approach is a quantification of the amount of ``violation'' of each constraint. Given a point $\mathbf{x} = (\mathbf{x}^\pp{0};\ldots; \mathbf{x}^\pp{r-1})$, for each $j =0,1,\ldots,m-1$, we define
\begin{align}
    v_j = \sum_{k=0}^{r-1} A_{j,:}^\pp{k} \mathbf{x}^\pp{k} - b_j
\end{align}
to be the amount by which the $j$-th constraint is violated---here $A_{j,:}^{\pp{k}} \mathbf{x}^\pp{k}$ is the scalar quantity equal to the inner product between $\mathbf{x}^{\pp{k}}$ and the $j$-th row of $A^{\pp{k}}$.  We partition the set $[m]$ into subsets based on how much violation they have, as illustrated in \cref{fig:ThetaRegion} and the definitions below. 

\begin{definition}[Violated Constraints]
Let \(\theta \geq 0\) be a given threshold parameter. A constraint is said to be \textbf{violated} if its violation exceeds \(\theta\), that is, $v_j = \sum_{k=0}^{r-1} A_{j,:}^\pp{k} \mathbf{x}^\pp{k} - b_j > \theta $. The set of all such constraints is denoted by \(V_{> \theta} \subseteq [m]\) .
\end{definition}

% Next, we define the allowable regime of constraint violations for our algorithm. Specifically, in the case of an exact search, we have specified in \cref{def:SOCP_feasibility}, violations of the form  
% \[
% \sum_{k=0}^{r-1} A^\pp{k} \mathbf{x}^\pp{k} > \mathbf{b} + \theta \mathbf{1}.
% \]
% However, to account for errors in the algorithm, we introduce the notion of \textbf{\(\theta\)-violation}, which imposes a stricter bound to ensure we guarantee that it is signaled with certainty if all constraints are satisfied. The following definitions detail then that the algorithm outputs ``all constraints satisfied'' if  \[
% \sum_{k=0}^{r-1} A^\pp{k} \mathbf{x}^\pp{k} \leq \mathbf{b} + \frac{\theta}{2} \mathbf{1}.
% \]

\begin{definition}[\(\theta\)-Violated Constraints]
A constraint is said to be \textbf{\(\theta\)-violated} if its violation is within the range \((\theta/2, \theta]\), that is $\theta/2 < \sum_{k=0}^{r-1} A_{j,:}^\pp{k}\mathbf{x}^\pp{k} -b_j \leq \theta$. The set of all such constraints is denoted by \(V_{\theta} \subseteq [m]\).
\end{definition}

\begin{definition}[Extended Violated Constraints set V]
The set of all constraints that are either violated or \(\theta\)-violated is denoted by \(V\), i.e.,  
\[
V = V_{> \theta} \cup V_{\theta}.
\]
\end{definition}

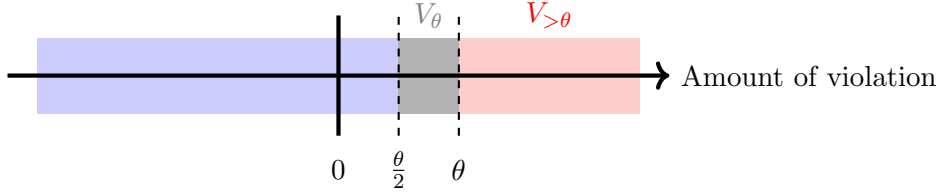
\begin{figure}
    \centering
    \begin{tikzpicture}
        % Define dimensions
        \def\rectWidth{8}
        \def\rectHeight{1}
        \def\thetaPos{1.6} 
        
        % Draw the rectangle
        \fill[blue!20] (0,0) rectangle (\rectWidth/2+\thetaPos/2,\rectHeight);
        \fill[gray!60] (\rectWidth/2+\thetaPos/2,0) rectangle (\rectWidth/2+\thetaPos,\rectHeight);
        \fill[red!20] (\rectWidth/2+\thetaPos,0) rectangle (\rectWidth,\rectHeight);

        \node[gray] at (\rectWidth/2 + \thetaPos*3/4,\rectHeight*1.3) {$V_{\theta}$};
        \node[red] at (\rectWidth*3/4 + \thetaPos/2,\rectHeight*1.3) {$V_{>\theta}$};
        
        % Draw the arrow 
        \draw[->, ultra thick] (-\rectWidth*0.05,\rectHeight*0.5) -- (\rectWidth*1.05,\rectHeight*0.5) node[right] {Amount of violation};
        
        % Draw vertical line at center (0)
        \draw[ultra thick] (\rectWidth/2,-0.3) -- (\rectWidth/2,\rectHeight+0.3);
        \node at (\rectWidth/2,-0.75) {0};
        
        % Draw vertical line for theta/2
        \draw[thick, dashed] (\rectWidth/2 + \thetaPos/2,-0.3) -- (\rectWidth/2 + \thetaPos/2,\rectHeight+0.3);
        \node at (\rectWidth/2 + \thetaPos/2,-0.75) {$\frac{\theta}{2}$};

        % Draw vertical line for theta
        \draw[thick, dashed] (\rectWidth/2 + \thetaPos,-0.3) -- (\rectWidth/2 + \thetaPos,\rectHeight+0.3);
        \node at (\rectWidth/2 + \thetaPos,-0.75) {$\theta$};
        
    \end{tikzpicture}
    \caption{Given a point $\mathbf{y} \in \mathbb{R}^m$, which implicitly defines $(\mathbf{x}^\pp{0};\ldots;\mathbf{x}^\pp{r-1})$ via \cref{eq:x^\pp{k}_implicit_VCO}, each constraint $j \in [m]$ is violated by an amount $\sum_{k=0}^{r-1} A^\pp{k}_{j,:}\mathbf{x}^\pp{k} - b_j$ (negative numbers indicate the constraint is satisfied). The set $V_{>\theta}$ contains values of $j$ for which the violation is more than $\theta$, and the set $V_{\theta}$ contains values of $j$  the violation is in the interval $(\theta/2,\theta]$.  }
    \label{fig:ThetaRegion}
\end{figure}

To understand why we have defined the sets this way, notice that we will never be able to ``measure'' $v_j$ perfectly; higher precision will come at higher cost, and we are trying to minimize the cost, hence we require buffer zones around the critical values of $v_j$ at 0 and $\theta$. 
% Suppose we are able to measure $v_j$ to precision $\theta/2$; then we will (i) never mistakenly assign a $j$ for which $v_j \leq 0$ into the set $V$, (ii) never mistakenly assign a $j$ for which $j \in V$ into the set $[m]\setminus V$, and (iii) never mistakenly assign a  That is, if we are promised that a $j$ satisfies $v_j \leq 0$ or $v_j > \theta$, then determining whether $j \in V$ allows us to perfectly distinguish these cases. Furthermore, if $v_j \in $
Specifically, recall from \cref{def:SOCP_feasibility} that we will be promised that our situation is one of two cases. The first case is that there exists a point $\mathbf{x}$ for which all constraints are fully satisfied ($v_j \leq 0$ for all $j$), which also implies that $V = \varnothing$ for that point. Importantly, in this case, we need only produce a vector $\mathbf{y}$ that defines a point $\mathbf{x}$ for which $V_{>\theta}$ is empty. We don't actually need to produce a point where all constraints are fully satisfied; thus, it will be fine to consider constraints violated by $\theta/2$ or less as effectively satisfied, as in \cref{fig:ThetaRegion}. The second case is that for all $\mathbf{x}$, $V_{>\theta} \neq \varnothing$. In this case, as long as we can ``measure'' $v_j$ to precision, say, $\theta/4$, then for at least one $j$, our estimate of $j$ will be at least $3\theta/4$, and conversely any $j$ that satisfies this criteria must be in $V$. Thus, we can find a point with $\Omega(\theta)$ violation using only precision $\Omega(\theta)$. We distill our desiderata into the following ``violated constraint oracle.''

\begin{definition}[Violated constraint oracle]\label{def:violated_constraint_oracle}
    A violated constraint oracle is a (classical or quantum) subroutine that satisfies the following input-output criteria. It takes as input 
    \begin{itemize}
        \item An instance of the unit-trace SOCP feasibility problem of \cref{def:SOCP_feasibility}, specified by an error parameter $\theta$, matrices $A^\pp{0}, \ldots, A^\pp{r-1}$ and vector $\mathbf{b}$
        \item A vector $\mathbf{y} \in \mathbb{R}^m$ 
        \item A maximum failure probability $\xi$
    \end{itemize}
    The inputs implicitly define a point $\mathbf{x} = (\mathbf{x}^\pp{0};\ldots;\mathbf{x}^\pp{r-1})$ where
    \begin{equation}\label{eq:x^\pp{k}_implicit_VCO}
        \mathbf{x}^\pp{k} = \frac{e^{-A^{\pp{k} \top} \mathbf{y}}}{\sum_{k'=0}^{r-1} \Tr(e^{-A^{\pp{k'} \top} \mathbf{y}})}
    \end{equation}
    and sets $V_{>\theta}, V_{\theta}, V \subset [m]$. The output of the oracle is:
    \begin{enumerate}[label=(\roman*)]
        \item If $V_{>\theta}$ is nonempty, then with probability at least $1-\xi$, output a value $j \in V = V_{>\theta} \cup V_\theta$.
        \item If $V$ is empty, then with probability at least $1-\xi$, output ``all constraints satisfied''.
        \item Otherwise, $V_{>\theta}$ is empty but $V_{\theta}$ is not empty. Then, with probability at least $1-\xi$ output either ``all constraints satisfied'' or output a value $j \in V_{\theta}$.
    \end{enumerate}
    % The probability of an invalid output should be at most $\xi$. 
\end{definition}

\subsection{Main algorithm}
The main algorithm makes repeated calls to the violated constraint oracle, and consistent with \cref{def:SOCP_feasibility}, it either outputs (i) (``Feasible'', $\mathbf{y} \in \mathbb{R}_{+}^m$), which means $\mathbf{y}$ builds a vector $\mathbf{x}$ via \cref{eq:feasibility_output_primal_vector} that is $\theta$-feasible; or (ii) (“Infeasible”), which means that the SOCP is infeasible. The probability of incorrect output provided the promise of \cref{def:SOCP_feasibility} is satisfied is taken to be at most 1/3. The multiplicative weights algorithm for the SOCP feasibility problem is given below.\\
\begin{algorithm}[H]
\caption{Pseudo-code for SOCP Feasibility MW algorithm.}\label{Alg1:description}
\begin{algorithmic}[1]
\Statex \textbf{Input:} $ A^\pp{0},\ldots,A^\pp{r-1},\mathbf{b}, \theta$ as in \cref{def:SOCP_feasibility}
\Statex \textbf{Output:} If set $\mathcal{S}_{\theta}$ is empty, outputs ``infeasible'' with probability at least 2/3. If set $\mathcal{S}_0$ is not empty, outputs  $\mathbf{y}$ that (implicitly) constructs $\theta$-feasible $\mathbf{x}^\pp{k}$ $\forall k$ with probability at least 2/3. (See \cref{def:SOCP_feasibility}.) 
\Statex
\Procedure{FeasibilitySOCP}{}
%{$A^\pp{0},\ldots,A^\pp{r-1}, \mathbf{b}, \theta$}
    \State $T \gets \frac{36\log (2r)}{\theta^2}$
    \State $\xi \gets 1/(3T)$ 
    \State $\mathbf{y}^{(0)} \gets \mathbf{0} \in \mathbb{R}^m$ 
    \For{$t =0$ to $T-1$}
    \State $j^{(t)} \gets \mathrm{ViolationConstraintOracle}(\mathbf{y}^{(t)},A^\pp{0},\ldots,A^\pp{r-1},\mathbf{b},\theta,\xi )$\label{line:VCO}
    \If{$j^{(t)} = $ ``all constraints satisfied''}
    \State  \Return (``Feasible'', $\mathbf{y} = \mathbf{y}^{(t)}$) \Comment{Conclude that $S_{0}\neq \emptyset$}
    \Else \Comment{$j^{(t)} \in V$ is a violated constraint}
    \State $\mathbf{y}^{(t+1)} \gets \mathbf{y}^{(t)} + \frac{\theta}{6} \mathbf{e}_{j^{(t)}}$
    \EndIf
    \EndFor
    \State \Return $\text{(“Infeasible”)}$. \Comment{If $j^{(t)} \in V$ found in all $T$ iterations, conclude that $S_{\theta}=\emptyset$}
\EndProcedure
\end{algorithmic}
\end{algorithm}

\subsection{Convergence of the main algorithm}

First, we outline some key facts of SOCP and lemmas that will help us prove \cref{theorem:MW_algo}. At a high level, this section is a translation of the SDP case from \cite{Kale07} to the specific case of SOCP. Notably, equivalent proofs for some of these lemmas can be found in \cite{Canyakmaz2023MultiplicativeUF}.

We first define notation used in the literature and in \cref{Alg1:description}. For $t=0,\ldots,T-1$, let $j^{(t)}$ denote the constraint index returned by the violated constraint oracle in line \ref{line:VCO} of \cref{Alg1:description} on iteration $t$. Here we can assume the situation when an index is returned directly. Let $\delta > 0$ be a tunable parameter (later we will identify $\delta = \theta/6$)
\begin{itemize}
    \item For each cone $k = 0,\ldots, r-1$ and each $t = 0,\ldots, T-1$, define
\begin{align}
\mathbf{m}^{(k,t)}& = \frac{1}{2}(\mathbf{e}^\pp{k}-(A^{\pp{k}\top})_{:,j^{(t)}})\text{, where $j^{(t)} \in [m]$},
% \text{ where }  \mathbf{e}^\pp{k}=(1,0,\ldots,0)\in \mathbb{R}^{n^\pp{k}} \text{ is the identity vector}  
\\
\phi^{(k,t)} &= \Tr(\exp(\delta \sum_{\tau=0}^{t-1} \mathbf{m}^{(k,\tau)})) \\
\mathbf{p}^{(k,t)} &= \frac{\exp(\delta \sum_{\tau=0}^{t-1} \mathbf{m}^{(k,\tau)})}{\phi^{(k,t)}}.
\end{align}

    \item To extend to the multicone case, we concatenate the cone vectors, forming:
\begin{align}
\mathbf{M}^{(t)}& = \left(\mathbf{m}^{(0,t)};\ldots;
\mathbf{m}^{(r-1,t)}\right)\in \mathbb{R}^{n^\pp{0}}\times\ldots\times\mathbb{R}^{n^\pp{r-1}} \\
\Phi^{(t)} &= \Tr(\exp(\delta \sum_{\tau=0}^{t-1} \mathbf{M}^{(\tau)})) = \sum^{r-1} _{k=0}\Tr(\exp(\delta \sum_{\tau=0}^{t-1} \mathbf{m}^{(k,\tau)})) = \sum^{r-1}_{k=0} \phi^{(k,t)} \label{eq:Phi^{(t)}}\\
\mathbf{P}^{(t)} &= \frac{\exp(\delta \sum_{\tau=0}^{t-1} \mathbf{M}^{(\tau)})}{\Phi^{(t)}}= \frac{1}{\Phi^{(t)}}\left(\exp(\delta \sum_{\tau=0}^{t-1} \mathbf{m}^{(0,\tau)});\ldots;\exp(\delta \sum_{\tau=0}^{t-1} \mathbf{m}^{(r-1,\tau)})\right)
\end{align}

\end{itemize}

We drop the superscripts $k, t$ (e.g., $\mathbf{m}^{(k,t)} \mapsto \mathbf{m}$) for the following lemmas:

\begin{lemma}[Exponential inequality]\label{Lemma:exponential_inequality}
    Given $\delta\leq 1$, let $\delta_1 :=e^{\delta}-1 $ and $\delta_2 :=1-e^{-\delta} $. If $\| \mathbf{m}\|_{\rm soc}\leq1$, then
    \begin{align}
             e^{\delta \mathbf{m}} &\preceq  \mathbf{e} +   \delta_1 \mathbf{m}  
            \quad \text { if }  \mathbf{m} \succeq \mathbf{0} \text{ i.e. all eigenvalues}\in [0,1]\\
             e^{\delta \mathbf{m}} &\preceq  \mathbf{e} +   \delta_2 \mathbf{m}      \quad \text { if }  \mathbf{m} \preceq \mathbf{0} \text{ i.e. all eigenvalues}\in [-1,0]
\end{align}
    and if $\forall k$ $\| \mathbf{m}^\pp{k}\|_{\rm soc} \leq 1$: 
        \begin{equation}
        e^{\delta \mathbf{M}} =(e^{\delta \mathbf{m}^\pp{0}};\ldots;e^{\delta \mathbf{m}^\pp{r-1}})   \preceq  
        \begin{cases}
        (\mathbf{e}^\pp{0} +\delta_1\mathbf{m}^\pp{0};\ldots;\mathbf{e}^\pp{r-1} +\delta_1\mathbf{m}^\pp{r-1})  & \text{if $\mathbf{m}^\pp{k} \succeq 0$ for all $k$}\\
        (\mathbf{e}^\pp{0} + \delta_2\mathbf{m}^\pp{0};\ldots;\mathbf{e}^\pp{r-1} + \delta_2\mathbf{m}^\pp{r-1})  & \text{if $\mathbf{m}^\pp{k} \preceq 0$ for all $k$}        
        \end{cases}
    \end{equation}
\end{lemma}
    \begin{proof} 
        Note that by the definition of exponentiation for second-order cones, the left-hand side $e^{\delta \mathbf{m}}$ and the right-hand side $ \mathbf{e}+\delta_i \mathbf{m}$ share the same Jordan frame. Thus, it is sufficient to show that the  eigenvalues of the left-hand side are each less than or equal to the corresponding eigenvalues of the right-hand side. Let $\lambda_{\pm}$ denote two eigenvalues of $\mathbf{m}$, noting that $|\lambda_{\pm}|\leq 1$ by assumption. The two eigenvalues of the left-hand side are given by $e^{\delta \lambda_{\pm}}$. We note the following inequalities over real numbers, which follow from the convexity of the exponential function:
\begin{align}
     \exp (\delta \lambda_{\pm}) \leq \begin{cases}1+ (e^{\delta}-1) \lambda_{\pm}  = 1+\delta_1 \lambda_{\pm}  & \text { if } \lambda_{\pm} \in[0,1] \text{ and }\delta \leq 1 \\
      1+ (1-e^{-\delta}) \lambda_{\pm} =1+\delta_2 \lambda_{\pm}   &  \text { if } \lambda{\pm} \in[-1,0] \text{ and }\delta \leq 1
     \end{cases}
\end{align}
where we have substituted $\delta_1 := e^{\delta}-1$ and $\delta_2 := 1-e^{-\delta}$. This proves the single-cone statement. The multicone statement follows immediately, since each cone can be treated independently. 

    \end{proof}

\begin{lemma}\label{lem:trace_replace}
    If $\mathbf{A} \succeq 0$ and $\mathbf{B} \preceq \mathbf{C}$, then 
    \begin{align}
        \Tr(\mathbf{A} \circ \mathbf{B}) \leq \Tr(\mathbf{A} \circ \mathbf{C})
    \end{align}
\end{lemma}
\begin{proof}
    We have $\mathbf{C}-\mathbf{B} \succeq 0$, hence the state $\mathbf{D} = \sqrt{\mathbf{C}-\mathbf{B}}$ is well defined and satisfies $\mathbf{D} \circ \mathbf{D} = \mathbf{C} - \mathbf{B}$. We have
    \begin{align}
        \Tr(\mathbf{A} \circ (\mathbf{C} - \mathbf{B})) &= \Tr(\mathbf{A} \circ (\mathbf{D} \circ \mathbf{D})) \\
        &= 2\mathbf{A}^{\top}\operatorname{Arw}(\mathbf{D})\mathbf{D} \\
        &= 2\mathbf{D}^{\top}\operatorname{Arw}(\mathbf{A})\mathbf{D} \\
        &\geq 2 \min_{\mathbf{x}: \nrm{\mathbf{x}}=\nrm{\mathbf{D}}} \mathbf{x}^{\top}\operatorname{Arw}(\mathbf{A}) \mathbf{x} \\
        &= 2 \nrm{\mathbf{D}}^2 \lambda_{\min}(\operatorname{Arw}(\mathbf{A})) \\
        &\geq 0
    \end{align}
    where the final inequality follows from the fact that $\mathbf{A} \succeq 0$. This implies the lemma statement.
\end{proof}

\begin{lemma}[From Golden-Thompson inequality \citep{Tao2021}]\label{lemma:Golden-Thomson}
Let \( V \) be any Euclidean Jordan algebra. Then for \( \mathbf{m}, \mathbf{q} \in V \),
\begin{equation}
\Tr(e^{\mathbf{m} + \mathbf{q}}) \leq \Tr(e^\mathbf{m} \circ e^\mathbf{q}). 
\end{equation}
where equality holds if and only if \( \mathbf{m} \) and \( \mathbf{q} \) share the same Jordan frame. 
\end{lemma}

Moreover, for the Euclidean Jordan algebra associated to the second-order cone,
\begin{equation}
\Tr(e^{\mathbf{m} + \mathbf{q}}) \leq \Tr(e^\mathbf{m} \circ e^\mathbf{q}) = 2 (e^{\mathbf{m}})^{\top} e^\mathbf{q}. 
\end{equation}

Then for the multicone case, 
\begin{equation}
\Tr(e^{\mathbf{M}+\mathbf{Q}}) \leq \sum^{r-1}_{k=0}\Tr(e^{\mathbf{m}^\pp{k}} \circ e^{\mathbf{q}^\pp{k}}) = \sum^{r-1}_{k=0}2(e^{\mathbf{m}^\pp{k}})^\top e^{\mathbf{q}^\pp{k}}\end{equation}

In the following proposition, we bound the relationship between our candidate solution $\mathbf{P}^{(t)}$ and any vector $\mathbf{Q}$ that has unit trace, and lies in the Cartesian product of second‐order cones, and could potentially be a solution to the feasibility problem. This comparison ensures that $\mathbf{P}^{(t)}$ reliably distinguishes between the sets $\mathcal{S}_0$ and $\mathcal{S}_{\theta}$, serving as an effective certificate of feasibility (see \cref{def:SOCP_feasibility}).

To prove convergence we use a potential function $\Phi$, which is a common tool in the MW literature \cite{Kale07}. Intuitively, the potential function measures the cumulative effect of past choices and the corresponding updates on the algorithm’s confidence, where each suboptimal update multiplicatively shrinks it.

\begin{proposition}\label{prop:tr_m_p}
    Suppose $\mathbf{M}^{(0)}, \ldots, \mathbf{M}^{(T-1)}$ are vectors satisfying $\mathbf{e} \succeq \mathbf{M}^{(t)} \succeq 0$ for all $t$. For a fixed $0<\delta \leq 1$, define 
    \begin{align}
    \mathbf{P}^{(t)} = \frac{e^{\delta \sum_{\tau=0}^{t-1} \mathbf{M}^{(\tau)}}}{\Tr(e^{\delta \sum_{\tau=0}^{t-1} \mathbf{M}^{(\tau)}})}\,.
    \end{align}
    Then for any $\mathbf{Q} \succeq 0$ satisfying $\Tr(\mathbf{Q}) = 1$, we have
    \begin{align}\label{eq:ineq_prop1}
        (1+\delta)\sum_{t=0}^{T-1}\Tr(\mathbf{M}^{(t)} \circ \mathbf{P}^{(t)}) \geq \sum_{t=0}^{T-1}\Tr(\mathbf{M}^{(t)} \circ \mathbf{Q}) - \frac{\log(2r)}{\delta}
    \end{align}
\end{proposition}

\begin{proof}
Define $\Phi^{(t)}$ as in \cref{eq:Phi^{(t)}}.  We prove \cref{eq:ineq_prop1} by establishing upper and lower bounds on the potential function at step $T$. We begin with the upper bound:
\begin{align}
    \Phi^{(t+1)} &= \Tr(\exp(\delta \sum_{\tau=0}^t \mathbf{M}^{(\tau)})) \\
   &\leq \Tr(e^{\delta \sum_{\tau=0}^{t-1} \mathbf{M}^{(\tau)}} \circ e^{\delta \mathbf{M}^{(t)}}) \quad  \quad  \quad \quad \quad \because  \text{lemma \ref{lemma:Golden-Thomson}} \\
   &\leq \Tr(e^{\delta \sum_{\tau=0}^{t-1} \mathbf{M}^{(\tau)}} \circ (\mathbf{e} + (e^{\delta}-1) \mathbf{M}^{(t)}) )  \qquad \qquad \because \text{lemmas \ref{Lemma:exponential_inequality}, \ref{lem:trace_replace}}\\ 
   &= \Tr(e^{\delta \sum_{\tau=0}^{t-1} \mathbf{M}^{(\tau)}} ) + (e^{\delta}-1) \Tr(e^{\delta \sum_{\tau=0}^{t-1} \mathbf{M}^{(\tau)}} \circ \mathbf{M}^{(t)})  \\
   &= \Phi^{(t)} + (e^{\delta}-1)\Phi^{(t)} \Tr(\mathbf{P}^{(t)} \circ \mathbf{M}^{(t)}) \\
   &= \Phi^{(t)}(1+(e^{\delta}-1)\Tr(\mathbf{P}^{(t)} \circ \mathbf{M}^{(t)})) \\
   &\leq \Phi^{(t)}\exp((e^{\delta}-1)\Tr(\mathbf{P}^{(t)} \circ \mathbf{M}^{(t)})) 
    \end{align}

Then, by induction:
\begin{equation}
    \Phi^{(T)}\leq \Phi^{(0)}e^{(e^{\delta}-1) \sum_{\tau = 0}^{T-1}  \Tr(\mathbf{P}^{(\tau)} \circ \mathbf{M}^{(\tau)}) }
\end{equation}

Observe that $\phi^{(k,0)}=2$ for each cone $k$, and thus $\Phi^{(0)} = 2r$. Hence,
\begin{equation}\label{eq:upper_bound_potential_function}
    \Phi^{(T)}\leq 2re^{(e^{\delta}-1)\sum_{\tau=0}^{T-1} \Tr(\mathbf{P}^{(\tau)} \circ \mathbf{M}^{(\tau)})}
\end{equation}
We now turn to the lower bound on the potential function at step $T$:
\begin{align}
    \Phi^{(T)} &= \Tr(e^{\delta \sum_{\tau=0}^{T-1} \mathbf{M}^{(\tau)}}) =  \sum^{r-1}_{k=0} e^{\delta\lambda_+(\sum_{\tau=0}^{T-1}\mathbf{m}^{(k,\tau)})}+e^{\delta\lambda_{-}(\sum_{\tau=0}^{T-1}\mathbf{m}^{(k,\tau)})}\\
    &\geq \sum^{r-1}_{k=0} e^{\delta\lambda_{\max}(\sum_{\tau=0}^{T-1}\mathbf{m}^{(k,\tau)})} \geq \max_k e^{\delta\lambda_{\max}(\sum_{\tau=0}^{T-1}\mathbf{m}^{(k,\tau)})} \\
    &= e^{\delta \lambda_{\max}(\sum_{\tau=0}^{T-1} \mathbf{M}^{(\tau)})} = e^{\delta \lambda_{\max}(\sum_{\tau=0}^{T-1} \mathbf{M}^{(\tau)})\cdot \Tr(\mathbf{Q})}=e^{\delta \lambda_{\max}(\sum_{\tau=0}^{T-1} \mathbf{M}^{(\tau)})\cdot (\sum^{r-1}_{k=0}2q^\pp{k}_0)}\\
    &\geq e^{\delta \sum^{r-1}_{k=0}2\mathbf{q}^{\pp{k}\top}  (\sum_{\tau=0}^{T-1} \mathbf{m}^{(k,\tau)})} = e^{\delta \Tr(\mathbf{Q} \circ \sum_{\tau=0}^{T-1} \mathbf{M}^{(\tau)})} = e^{\delta \sum_{\tau=0}^{T-1} \Tr(\mathbf{Q} \circ  \mathbf{M}^{(\tau)})}
\end{align}
where $\mathbf{Q} = (\mathbf{q}^{\pp{0}}; \ldots; \mathbf{q}^{\pp{k-1}}$) has trace 1, and $\lambda_{\max}(\sum_{\tau=0}^{T-1} \mathbf{M}^{(\tau)})$ is the maximum eigenvalue of $\operatorname{Arw}(\sum_{\tau=0}^{T-1} \mathbf{M}^{(\tau)})$, which is equivalent to $\max_{k}\lambda_{\rm max}(\operatorname{Arw}(\sum_{\tau=0}^{T-1} \mathbf{m}^{k(\tau)}))$. Combining these inequalities and taking the logarithm of both sides yields 
\begin{align}
    \log(2r) + (e^{\delta}-1) \sum_{\tau=0}^{T-1} \Tr( \mathbf{P}^{(\tau)} \circ \mathbf{M}^{(\tau)}) \geq \delta \sum_{\tau=0}^{T-1} \Tr(\mathbf{Q} \circ \mathbf{M}^{(\tau)})
\end{align}
Subtracting $\log(2r)$ from both sides, dividing by $\delta$, and noting that $(e^{\delta}-1)/\delta \leq 1+ \delta$ for $0 < \delta \leq 1$ yields the theorem statement. 
\end{proof}

 \begin{theorem}[Correctness of main algorithm]\label{theorem:MW_algo} 
 \Cref{Alg1:description} correctly solves the SOCP feasibility problem of \cref{def:SOCP_feasibility} with probability at least 2/3. It uses at most $T = \frac{36\log(2r)}{\theta^2}$ queries to the Violated Constraint oracle (\cref{def:violated_constraint_oracle}), with parameter setting $\xi = 1/(3T)$. 

\end{theorem}
\begin{proof}
The stated complexity of the theorem is a direct consequence of the fact that the algorithm terminates after at most $T$ iterations, and uses parameter setting $\xi = 1/3T$. Moreover, each call to violated constraint oracle fails with probability at most $\xi$, thus by the union bound the probability that all $T$ calls to the oracle succeed is at least 2/3. Hence, what remains to show is that the output of the algorithm is correct whenever the oracle is correct on every iteration, in each of the two cases detailed in \cref{def:SOCP_feasibility}. 
% Note that in either case, the Violated-Constraint oracle is called at most $T$ times, and each time outputs a valid answer with probability at least $1-\xi = 1-1/(3T)$. By the union bound, the probability that all $T$ calls to the Violated-Constraint oracle are a valid answer is at least 2/3. 

\begin{enumerate}[label=(\roman*)]
\item \textbf{ If $\mathcal{S}_0 \neq \varnothing$, the algorithm outputs ``feasible'' and a vector $\mathbf{y}$ implicitly defining a vector  $\mathbf{x} \in \mathcal{S}_\theta$ via \cref{eq:feasibility_output_primal_vector}. } 
First, suppose for contradiction that the algorithm outputs ``infeasible''. This implies that each time the violated constraint oracle is queried, it finds a violated constraint, i.e., there is a sequence $j^{(0)},\ldots,j^{(T-1)}$ of constraint indices such that for each $t=0,\ldots,T-1$,
\begin{align}\label{eq:violated_constraint_t}
\sum_{k=0}^{r-1}A^\pp{k}_{j^{(t)},:} \frac{e^{-\frac{\theta}{6}\sum_{\tau = 0}^{t-1}A^{\pp{k} \top}\mathbf{e}_{j^{(\tau)}}}}{\sum_{k'=0}^{r-1} \Tr(e^{-\frac{\theta}{6}\sum_{\tau = 0}^{t-1}A^{\pp{k'} \top}\mathbf{e}_{j^{(\tau)}}})} \geq b_{j^{(t)}} + \frac{\theta}{2}
\end{align}
For $t = 0,\ldots, T-1$, define multicone vectors
\begin{align}
\mathbf{M}^{(t)} = \frac{1}{2}(\mathbf{e}^\pp{0} - A^{\pp{0} \top} \mathbf{e}_{j^{(t)}}; \ldots;\mathbf{e}^\pp{r} -  A^{\pp{r-1} \top} \mathbf{e}_{j^{(t)}}) =  \frac{\mathbf{e} - ((A^\pp{0}_{j^{(t)},:})^\top;\ldots;(A^\pp{r-1}_{j^{(t)},:})^\top)}{2}
\end{align}
that is, $\mathbf{M}^{(t)}$ is proportional to the identity vector $\mathbf{e}$ minus the $j^{(t)}$-th row of the constraint matrix $A = (A^\pp{1},\ldots,A^\pp{r-1})$. Since the normalisation conditions ensure that $\nrm{A^\pp{k}_{j^{(t)},:}}_{\rm soc} \leq 1$, we may assert that $\mathbf{e} \succeq \mathbf{M}^{(t)} \succeq 0$. For $t=0,\ldots,T-1$, we define
\begin{align}
    \mathbf{P}^{(t)} = \frac{e^{\frac{\theta}{3} \sum_{\tau=0}^{t-1} \mathbf{M}^{(\tau)}}}{\Tr(e^{\frac{\theta}{3} \sum_{\tau=0}^{t-1} \mathbf{M}^{(\tau)}})}
\end{align}
Note that if we add a multiple of the identity vector $\mathbf{e}$ to the exponent in the numerator and the denominator of \cref{eq:violated_constraint_t}, it will cancel and not impact the vector. Also note that $2\mathbf{e}^\top \mathbf{p} = \Tr(\mathbf{p})$ for any $\mathbf{p}$. Thus we may equivalently rewrite \cref{eq:violated_constraint_t} as
\begin{align}
\mathbf{M}^{(t)\intercal} \mathbf{P}^{(t)}  \leq \frac{1-b_{j^{(t)}} - \frac{\theta}{2}}{2}
\end{align}
which is equivalent to
\begin{align}
\Tr(\mathbf{M}^{(t)} \circ \mathbf{P}^{(t)})  \leq 1-b_{j^{(t)}} - \frac{\theta}{2}
\end{align}
Since we have assumed that $\mathcal{S}_0 \neq 0$, there exists a vector $\mathbf{Q}$ that is feasible. This implies that 
\begin{align}
    \Tr(\mathbf{M}^{(t)} \circ \mathbf{Q})\geq 1-b_{j^{(t)}}
\end{align}
for all $t$. We now have
\begin{align}
    (1+\frac{\theta}{6})\sum_{t=0}^{T-1} (1-b_{j^{(t)}}-\frac{\theta}{2}) &\geq  (1+\frac{\theta}{6}) \sum_{t=0}^{T-1} \Tr(\mathbf{M}^{(t)} \circ \mathbf{P}^{(t)}) \\
    &\geq \sum_{t=0}^{T-1}\Tr(\mathbf{M}^{(t)} \circ \mathbf{Q})-\frac{6 \log(2r)}{\theta} \quad \because \cref{prop:tr_m_p} \\
    &\geq \sum_{t=0}^{T-1}(1-b_{j^{(t)}})-\frac{6 \log(2r)}{\theta} 
\end{align}
This is equivalent to
\begin{align}
    -\frac{\theta T}{3} -\frac{\theta^2T}{12} \geq - \frac{6\log(2r)}{\theta }+ \frac{\theta}{6}\sum_{t=0}^{T-1}b_{j^{(t)}} 
\end{align}
and since $b_{j^{(t)}} \geq -1$, it implies 
\begin{align}
 T \leq \frac{6\log(2r)}{\theta(\frac{\theta}{6} + \frac{\theta^2}{12}) } \leq \frac{36\log(2r)}{\theta^2}
\end{align}
Thus, since we have chosen $T > \frac{36 \log(2r)}{\theta^2}$, there is a contradiction, and we may conclude that in this case the algorithm does not output ``infeasible.'' 

Consequently, at some iteration $t <T$ the algorithm calls the violated constraint oracle and receives the output ``all constraints satisfied''. By \cref{def:violated_constraint_oracle}, this output is only possible when the vector $\mathbf{P}^{(t)}$ that is generated from the vector $\mathbf{y}^{(t)}$ satisfies all the constraints, up to an additive tolerance of $\theta$, or in other words, $ \mathbf{P}^{(t)} \in \mathcal{S}_{\theta}$, verifying that the output is correct.

\item \textbf{If $\mathcal{S}_{\theta} =\varnothing$, then the  algorithm outputs ``infeasible''} 
This is true because at each iteration $t$, regardless of $\mathbf{y}^{(t-1)} \in \mathbb{R}^{m}$, there will always exist a constraint that is more than $\theta$-violated (otherwise, $\mathcal{S}_{\theta}$ would be nonempty). Thus, the violated constraint oracle will return a value $j^{(t)}$ at each of the $T$ iterations and correctly return ``infeasible''. 
 \end{enumerate}
\end{proof}

\subsection{Two-step approach to implementing violated constraint oracle}\label{sec:two_step_approach}

A straightforward approach to implementing the violated constraint oracle would be to simply compute the vector $\mathbf{x} \propto e^{-A^\top \mathbf{y}} \in \mathbb{R}^n$ from the input $\mathbf{y} \in \mathbb{R}^m$, and then use $\mathbf{x}$ to compute the amount of violation $v_j$ for each $j$---the matrix multiplication $A^\top \mathbf{y}$ alone would have a classical cost $\bigO(mn)$ in general. However, this strategy is overkill; for example, it also allows us to compute $v_j$ to exact precision for all the $j$, whereas the problem we are solving explicitly allows for some tolerance of size $\theta$ on the violation amount. We can achieve better cost by instead \textit{sampling} some of the cones $k$ biased toward those with larger weight within the vector $\mathbf{x}$ (as measured by the trace). These samples are expensive but they can be re-used; all the $v_j$ are estimated using the same samples. This can be viewed as a simplified version of the quantum OR lemma that featured in multiplicative weights--based quantum algorithms for solving SDPs.

Specifically, both our classical and our quantum algorithms implement the violated constraint oracle through the same two-step approach, which we describe here. Recall from \cref{def:violated_constraint_oracle} that  we are given a vector $\mathbf{y} \in \mathbb{R}^m$, and the goal is to examine whether the implicitly defined point $\mathbf{x} = (\mathbf{x}^\pp{0}; \ldots; \mathbf{x}^\pp{r-1})$ violates any of the constraints. 

To understand this approach, for $k=0,\ldots, r-1$, we define the numbers
\begin{align}\label{eq:Z^\pp{k}}
    \mathcal{Z}^\pp{k} &= \Tr(e^{-A^{\pp{k} \top} \mathbf{y}}) 
\end{align}
and the unit-trace vectors
\begin{align}\label{eq:p^\pp{k}}
    \mathbf{p}^\pp{k} = \frac{e^{-A^{\pp{k} \top} \mathbf{y}}}{\mathcal{Z}^\pp{k}}\,.
\end{align}
We may also define
\begin{align}
    \mathcal{Z} &= \sum_{k=0}^{r-1} \mathcal{Z}^\pp{k} \label{eq:mathcalZ} 
\end{align}
and then write the point $\mathbf{x}$ as
\begin{align}
    \mathbf{x} &= (\mathbf{x}^\pp{0}; \ldots; \mathbf{x}^\pp{r-1}) = \frac{1}{\mathcal{Z}} \left(\mathcal{Z}^\pp{0} \mathbf{p}^\pp{0};\ldots;\mathcal{Z}^{\pp{r-1}} \mathbf{p}^\pp{r-1}\right)\label{eq:p_vector}
\end{align}
This form shows that the vector $\mathbf{x}$ has unit trace and can be understood as analogous to a Gibbs distribution.  

We define a subroutine, the ``cone index Gibbs sampler oracle,'' which repeatedly samples a cone index $k$ with probability equal to $\mathcal{Z}^\pp{k}/\mathcal{Z}$, up to some error tolerance. 

\begin{definition}[Cone index Gibbs sampler oracle, $C_{\rm samp}$]\label{oracle:cone_index_Gibbs_sampling}
    A cone index Gibbs sampler oracle is a (classical or quantum) subroutine that satisfies the following input-output criteria. It takes as input 
    \begin{itemize}
    \item An instance of the unit-trace SOCP feasibility problem of \cref{def:SOCP_feasibility}, specified by matrices $A^\pp{0}, \ldots, A^\pp{r-1}$, where $m$ is the number of constraints, $r$ is the number of cones, and $n$ the total number of variables.
    \item A vector $\mathbf{y} \in \mathbb{R}^m$, where $s$ denotes the sparsity of $\mathbf{y}$, and $\beta = \nrm{\mathbf{y}}_1$
        % \item The sparsity coefficient $s$ of $\mathbf{y}$
        % \item The $\ell_1$-norm of $\mathbf{y}$, denoted by $\beta$  
        \item An error parameter $\zeta$
        \item An integer $T'$
    \end{itemize}
    The inputs implicitly define numbers $\mathcal{Z}^\pp{k}$, $\mathcal{Z}$, as in \cref{eq:Z^\pp{k},eq:mathcalZ}. Let $\mathcal{P}$ be the probability distribution that assigns probability $\mathcal{Z}^\pp{k}/\mathcal{Z}$ to each integer $k \in [r]$. The output of the sampler is:
    \begin{itemize}
        \item $T'$ samples $k_0,k_1,\ldots,k_{T'-1} \in [r]$, where the joint distribution over the $T'$ samples is at most $\zeta$-far in total variation distance from the distribution where each $k_h$ is chosen i.i.d.~from $\mathcal{P}$.
    \end{itemize}
\end{definition}

Given a fixed set of $T'$ samples $k_0, k_1,\ldots, k_{T'-1} \in [r]$, for each $j =0,1,\ldots,m-1$, we may define the violation $v_{j,h}$ of sample $h$, and the overall average violation $v_j$ for the set of samples, by
\begin{align}
    \hat{v}_{j,h} &= A^\pp{k_h}_{j,:}\mathbf{p}^\pp{k_h} - b_j \label{eq:hat_v_jh}\\
    \hat{v}_j &= \frac{1}{T'} \sum_{h=0}^{T'-1} v_{j,h} \label{eq:hat_v_j}
\end{align}

We can also define sets analogous to $V_{\theta}$ and $V_{>\theta}$. Namely we let
\begin{align}
    \hat{V}_{>\theta}  &= \{j \colon \hat{v}_j > 5\theta/6\} \\
    \hat{V}_{\theta}  &= \{j \colon \hat{v}_j \in (4\theta/6, 5\theta/6]\} \\
    \hat{V} &= \hat{V}_{>\theta} \cup \hat{V}_\theta
\end{align}
Now we define a ``sampled violated constraint search oracle'' defined in terms of the sets $\hat{V}_{>\theta}$ and $\hat{V}_{\theta}$ based on the samples, as an analogue of the violated constraint oracle previously defined in terms of the sets $V_{>\theta}$ and $V_{\theta}$. 

\begin{definition}[Sampled violated constraint search oracle, $C_{\rm search}$,  cf.~\cref{def:violated_constraint_oracle}]\label{oracle:sampled_constraint_search}
    A sampled violated constraint search oracle is a (classical or quantum) subroutine that satisfies the following input-output criteria. It takes as input 
    \begin{itemize}
        \item An instance of the unit-trace SOCP feasibility problem of \cref{def:SOCP_feasibility}, specified by an error parameter $\theta$, matrices $A^\pp{0}, \ldots, A^\pp{r-1}$ and vector $\mathbf{b}$, where $m$ is the number of constraints, $r$ is the number of cones, and $n$ the total number of variables.
        %Alternatively, we are given the vectors $(-A^\pp{0}\mathbf{y} ,\ldots, -A^\pp{r-1}\mathbf{y} )$.
        \item A vector $\mathbf{y} \in \mathbb{R}^m$, where $s$ denotes the sparsity of $\mathbf{y}$, and $\beta = \nrm{\mathbf{y}}_1$.
        % \item The $\ell_1$-norm of $\mathbf{y}$, denoted by $\beta$  
        % \item The approximation margin $\theta$ for the solution
        \item An error parameter $\eta$
        \item A list of $T'$ samples of cone indices $\mathcal{T} = (k_0,\ldots, k_{T'-1}) \in [k]^{T'}$
    \end{itemize}
    The inputs implicitly define numbers $\hat{v}_j$, as in \cref{eq:hat_v_j}, and sets $\hat{V}_{>\theta}$ and $\hat{V}_{\theta}$. The output of the search is:
    \begin{enumerate}[label=(\roman*)]
        \item If $\hat{V}_{>\theta}$ is nonempty, then with probability at least $1-\eta$, output a value $j \in \hat{V} = \hat{V}_{>\theta} \cup \hat{V}_\theta$.
        \item If $\hat{V}$ is empty, then with probability at least $1-\eta$, output ``all constraints satisfied.''
        \item Otherwise, $\hat{V}_{>\theta}$ is empty but $\hat{V}_{\theta}$ is not empty. Then, with probability at least $1-\eta$ output either ``all constraints satisfied'' or output a value $j \in \hat{V}_{\theta}$.
        \end{enumerate}
\end{definition}

The reason to perform the two-step approach is that the cost of each step is additive, rather than multiplicative, as captured in the following theorem. 

\begin{theorem}\label{theorem:Csamp_Csearch_implements_alg1}
    Let $A^\pp{0},\ldots, A^\pp{r-1}, \mathbf{b},\theta$ denote a fixed instance of the feasibility SOCP (\cref{def:SOCP_feasibility}) with $r$ cones, $n=\sum^{r-1}_{k=0}n^\pp{k}$ total variables, and $m$ constraints, and let $\mathbf{y} \in \mathbb{R}^m$ be a fixed $s$-sparse non-negative vector with $\beta = \nrm{\mathbf{y}}_1$. Suppose one has a cone index sampler (\cref{oracle:cone_index_Gibbs_sampling}) that has complexity upper bounded by $C_{\rm samp}(m,r,n,s,\beta,\zeta,T')$ for drawing $T'$ samples with error parameter $\zeta$, and a sampled violated constraint search oracle (\cref{oracle:sampled_constraint_search}) that has complexity upper bounded by $C_{\rm search}(m,r,n,\beta,\theta,\eta, T')$ for error parameter $\eta$ and $T'$ total input samples. Then one can construct a violated constraint oracle (\cref{def:violated_constraint_oracle}) with failure probability $\xi$ with complexity 
    \begin{align}
        C_{\rm samp}(m,r,n,s,\beta,\frac{\xi}{4}, \frac{288 \log(\frac{8m}{\xi})}{\theta^2}) +   C_{\rm search}(m,r,n,\beta,\theta, \frac{\xi}{2},  \frac{288 \log(\frac{8m}{\xi})}{\theta^2})\,.
    \end{align}
\end{theorem}

\begin{proof}
    We aim to implement the violated constraint oracle of \cref{def:violated_constraint_oracle} with precision parameter $\theta$ and failure probability $\xi$. Here we specify a choice of parameters
    \begin{align}
        \eta &= \xi/2 \\
        \zeta &= \xi/4\\
        T' &= \frac{288\log(8m/\xi)}{\theta^2}
    \end{align}
    We query the Gibbs sampling oracle of \cref{oracle:cone_index_Gibbs_sampling} which produces samples $k_0,k_1,\ldots,k_{T'-1} \in [r]$ from some joint distribution that is $\zeta$-close in total variation distance to $T'$ independent samples from the ideal distribution $\mathcal{P}$ (for which $k$ is sampled with probability $\mathcal{Z}^\pp{k}/\mathcal{Z}$).  We view the quantities $\hat{v}_{j,h}$ defined in \cref{eq:hat_v_jh} as random variables that depend on $k_0,\ldots, k_{T'-1}$,  and we denote by $\mathbb{E}_{\mathcal{P}}$ the expectation value over samples from the ideal distribution. When sampling from the ideal distribution, the expectation value of $\hat{v}_{j,h}$ is $v_j$ for all $h$, verified by
    \begin{align}
        \mathbb{E}_{\mathcal{P}} [\hat{v}_{j,h}] &= \sum_{k=0}^{r-1}\frac{\mathcal{Z}^\pp{k}}{\mathcal{Z}}\left(A^\pp{k}_{j,:}\mathbf{p}^{\pp{k}} - b_j\right) \\ \nonumber
        &= \sum_{k=0}^{r-1} A_{j,:}^\pp{k} \mathbf{x}^\pp{k} - b_j \\ \nonumber 
        &= v_j
    \end{align}
    Since $\hat{v}_j$ is simply the average of $\hat{v}_{j,h}$, we thus have
    \begin{align}
        \mathbb{E}_{\mathcal{P}} [\hat{v}_{j}] = v_j
    \end{align}
    Furthermore, we may observe that $|\hat{v}_{j,h}| \leq 2$, since $|b_j|\leq 1$ and $\nrm{A_{j,:}^{\pp{k}}}_{\rm soc} \leq 1$ and $\Tr(\mathbf{p}^\pp{k}) =1$ implies $|A_{j,:}^{\pp{k}}\mathbf{p}^\pp{k}|\leq 1 $. Thus, we may use Hoeffding's inequality to assert that, when the samples are chosen according to $\mathcal{P}$,  the probability that $\hat{v}_{j}$ deviates from its mean $v_j$ by more than $\theta/6$ is upper bounded by $2\exp(-\theta^2 T'/288)$. Applied here, our choice of $T' = 288\log(8m/\xi)/\theta^2$ ensures that for each value of $j$
\begin{align}\label{eq:hat_v_j_close_to_v_j}
    \left\lvert \hat{v}_j - v_j \right\rvert \leq \frac{\theta}{6}
\end{align}
holds except with probability at most $\xi/(4m)$ over the random samples drawn from $\mathcal{P}$. By the union bound, the probability that this holds simultaneously for all $m$ values of $j$ is at least $1-\xi/4$. Since the actual distribution that produced $k_0,\ldots,k_{T'-1}$ is at most $\zeta$-far from this distribution, the probability that a certain event occurs can deviate by at most $\zeta = \xi/4$---thus, we conclude that $|\hat{v}_j - v_j|\leq \theta/6$ holds for all $j$ except with probability at most $\xi/2$. 
    
%     When sampling from the approximate distribution, there is some error in the expectation value, but this can be bounded. We may observe that $|\hat{v}_{j,h}| \leq 2$, since $|b_j|\leq 1$ and $\nrm{A_{j,:}^{\pp{k}}}_{\rm soc} \leq 1$ and $\Tr(\mathbf{p}^\pp{k}) =1$ implies $|A_{j,:}^{\pp{k}}\mathbf{p}^\pp{k}|\leq 1 $. Again since $\hat{v}_j$ is an average of $\hat{v}_{j,h}$, this also implies $|\hat{v}_j| \leq 2$. Thus, replacing $\mathcal{P}$ with $\widetilde{\mathcal{P}}$ can lead to at most $2\zeta$ error in the expected value. That is
%     \begin{align}
%         \left\lvert \mathbb{E}_{ \widetilde{\mathcal{P}}} [\hat{v}_{j}] - \mathbb{E}_{\mathcal{P}} [\hat{v}_{j}] \right \rvert \leq 2 \eta
%     \end{align}
% Now, Hoeffding's inequality ensures that the probability that $\hat{v}_{j}$ deviates from  $\mathbb{E}_{k_h \sim \widetilde{P}}[\hat{v}_{j,h}]$ by more than $\theta/12$ is upper bounded by $\exp(-\theta^2 T'/288)$. Applied here, our choice of $T' = 288\log(2m/\xi)/\theta^2$ ensures that for each value of $j$
% \begin{align}
%     \left\lvert \hat{v}_j - \mathbb{E}_{k_h  \sim \widetilde{\mathcal{P}}} [\hat{v}_{j,h}] \right\rvert \leq \frac{\theta}{12}
% \end{align}
% holds except with probability at most $\xi/(2m)$ over the random samples. By the triangle inequality, the union bound over the $m$ values of $j$, an the choice $\zeta = \theta/24$, we have that
% \begin{align}\label{eq:hat_v_j_close_to_v_j}
%     \left \lvert \hat{v}_j - v_j \right \rvert \leq \frac{\theta}{6}
% \end{align}
% holds simultaneously for all $j$, except with probability $\xi/2$. 

Due to the modified intervals in the definitions of $\hat{V}_{>\theta}$ and $\hat{V}_{\theta}$, compared to $V_{> \theta}$ and $V_{\theta}$, when the bound of \cref{eq:hat_v_j_close_to_v_j} holds for all $j$, we may assert that if $j \in V_{>\theta}$ then $j \in \hat{V}_{> \theta}$. Moreover, if $j \not\in V$ then $j \not\in \hat{V}$, or in  other words $\hat{V} \subset V$.

We run one call to the sampled violated constraint search oracle of \cref{oracle:sampled_constraint_search}, on inputs $k_0,\ldots, k_{T'-1}$. There is at most $\xi/2$ probability that \cref{eq:hat_v_j_close_to_v_j} does not hold, in which case the output may fail. Assuming it does hold, we must verify the three cases of the violated constraint oracle (\cref{def:violated_constraint_oracle}).   If $V_{>\theta}$ is not empty (case (i)), then by the above logic, $\hat{V}_{>\theta}$ is not empty, and the sampled violated constraint search oracle is guaranteed to produce a $j \in \hat{V} \subset V$, except with probability $\eta = \xi/2$. If $V$ is empty (case (ii)) then $\hat{V} \subset V$ is also empty,  and the sampled violated constraint search oracle is guaranteed to output ``all constraints satisfied'' except with probability $\eta = \xi/2$. Otherwise (case (iii)), $V_{> \theta}$ is empty but $V_{\theta}$ is not empty. Here we may be in any of the three cases of \cref{oracle:sampled_constraint_search}, but it suffices to observe that the sampled violated constraint search oracle outputs either ``all constraints satisfied'' or a value $j \in \hat{V} \subset V$, except with probability $\xi/2$. Considering the chance that \cref{eq:hat_v_j_close_to_v_j} fails, the overall failure probability is at most $\xi$, verifying the output is correct. 

The complexity of the algorithm is one call to the cone index Gibbs sampler oracle (producing $T'$ samples) and one call to the sampled violated constraint search oracle, with the appropriate parameters. 
\end{proof}

\section{Quantum implementation of two-step violated constraint oracle}\label{sec:quantum_implementation}

In \cref{sec:MWSOCP}, we illustrated how the feasibility SOCP can be solved with repeated calls to a violated constraint oracle of \cref{def:violated_constraint_oracle}, which roughly speaking finds a $j$ for which the $j$-th constraint is violated by at least $\theta$, if one exists. We then explained in \cref{sec:two_step_approach} how the violated constraint oracle can be solved with a two-step process; the first step is to sample cone indices with a cone index Gibbs sampler oracle (\cref{oracle:cone_index_Gibbs_sampling}), and the second step is to use those samples within a sampled violated constraint search oracle (\cref{oracle:sampled_constraint_search}). Here, we show how these two oracles can be implemented with a quantum algorithm and quantify their complexity. This enables us to prove the following theorems on the overall complexity of the violated constraint oracle, and by extension, the full SOCP, which we state here, referencing in the proof statements shown later in the section. 

\begin{theorem}\label{thm:quantum_complexity_violated_constraint_oracle}
    There is a quantum algorithm that implements the violated constraint oracle of \cref{def:violated_constraint_oracle} for a program with $r$ cones, $n$ total variables, $m$ constraints, and $\beta = \nrm{\mathbf{y}}_1$, with failure probability $\xi$ and precision parameter $\theta$ while incurring complexity
    \begin{align}
       \bigOt\left(\frac{\sqrt{r} \beta \log^2(1/\xi)}{\theta^2}  + \frac{\sqrt{m} \beta \log^2(1/\xi)}{\theta}\right) 
    \end{align}
    calls to the oracles $O_R$, $O_{\mathbf{y}}$,  $O_{\mathcal{T}}$, and their inverses, $O_{\mathbf{b}}$, plus the same number of additional single- and two-qubit gates up to a factor $\bigO(\log(r\bar{n}m))$. 
\end{theorem}
\begin{proof}
    This follows from the expression of \cref{theorem:Csamp_Csearch_implements_alg1} in terms of $C_{\rm samp}$ and $C_{\rm search}$, and the evaluation of these complexity expressions in \cref{cor:C_samp_quantum} and \cref{cor:C_search_quantum}, respectively. 
\end{proof}
\begin{corollary}\label{cor:quantum_complexity_feasibility_problem}
    There is a quantum algorithm for solving the unit-trace SOCP feasibility problem of \cref{def:SOCP_feasibility} for a program with $r$ cones, $n$ total variables,  $m$ constraints, and precision parameter $\theta$, while incurring complexity
    \begin{align}
    \bigOt\left(\frac{\sqrt{r}}{\theta^5}  + \frac{\sqrt{m}}{\theta^4}\right) 
    \end{align}
    calls to the oracles $O_R$, $O_{\mathbf{y}}$, $O_{\mathcal{T}}$ and their inverses, $O_{\mathbf{b}}$,  plus the same number of additional single- and two-qubit gates up to a factor $\bigO(\log(r\bar{n}m))$.
\end{corollary}
\begin{proof}
    This follows from \cref{thm:quantum_complexity_violated_constraint_oracle} and \cref{theorem:MW_algo}, noting that the maximum value of $\beta$ across all $T$ iterations of the algorithm is $\frac{T\theta}{6} = \bigO(\log(2r)/\theta)$. 
\end{proof}

\begin{corollary}[Complexity of Quantum implementation of SOCP MW] Let $\mathcal{P}$ be an instance of the general primal SOCP of \cref{Def:SOCP} with $r$ cones, $n$ total variables, and $m$ constraints. Assume $\mathcal{P}$ obeys the normalisation conditions and strong duality, and that it is $R$-trace and $\tilde{R}$-dual-trace constrained. Given error parameter $\epsilon$, there is a quantum algorithm that approximately solves $\mathcal{P}$ up to error $\epsilon$, in the sense of the statement of \cref{feasibility_reduction}, while incurring complexity
    \begin{align}
    \bigOt \left(\sqrt{r}\left(\frac{R\tilde{R}}{\epsilon}\right)^5 + \sqrt{m}\left(\frac{R\tilde{R}}{\epsilon}\right)^4\right) 
    \end{align}
    calls to the oracles  $O_R$, $O_{\mathbf{y}}$, $O_{\mathcal{T}}$ and their inverses, $O_{\mathbf{b}}$, plus the same number of additional single- and two-qubit gates up to a factor $\bigO(\log(r\bar{n}m))$. 
%operating under the following parameters: $m$ denotes the number of constraints, $r$ represents the number of cones, and $\theta$  defines the approximation margin for the solution. Additionally, $n=r\cdot \bar{n}$, where $\bar{n}$ is the maximum size of the cones. 
\end{corollary}
\begin{proof}
As described in \cref{feasibility_reduction}, it suffices to make $\bigOt(\log(R/\epsilon))$ calls to an oracle for the feasibility SOCP with error parameter $\theta = \epsilon/(4R\tilde{R})$. Thus, the complexity follows from \cref{cor:quantum_complexity_feasibility_problem}.
\end{proof}
Thus, the corollary establishes that, as stated in the introduction, the quantum algorithm has runtime $    \bigOt \left(\sqrt{r} \gamma^5 + \sqrt{m}\gamma^4\right) $, where $\gamma =\frac{R\tilde{R}}{\epsilon}$.

% Equivalently, we can state that there exists a end-to-end classical implementation based on the access model of \cref{section:access_model}, as described in \cref{section:classical_implementation_violated_constraint}, with query complexity
% \[
%  T \cdot \mathcal{O}(s n) + T \cdot \bigOt\left( \frac{m T' (R\tilde{R})^2}{\epsilon^2} \right) = \bigOt\left( \frac{(R\tilde{R})^2}{\epsilon^2} \left( n + \frac{m (R\tilde{R})^2}{\epsilon^2} \right) \right),
% \]
% where the sparsity \(s\) of \(\mathbf{y}\) satisfies \(s \leq T = \bigO(1/\theta^2)\).

% Thus, the quantum algorithm achieves a quadratic improvement in the \(m\)-term, and since typically \(n > r\), this can also be viewed as a quadratic speedup in \(r\).

\subsection{Quantum implementation of cone index Gibbs sampler oracle}

We aim to implement the oracle of \cref{oracle:cone_index_Gibbs_sampling}, where the inputs are the SOCP inputs $A^\pp{0},\ldots$, $A^\pp{r-1},$ $\mathbf{b},\theta$ , a vector $\mathbf{y} \in \mathbb{R}^m$, and an integer $T'$, and the desired output is $T'$ samples $k_0,\ldots,k_{T'-1}$ from the Gibbs distribution over cones. The entries of $\mathbf{y}$ are non-negative; let $\beta = \sum_{j=0}^{m-1} y_j$. The implementation utilizes a unitary $2\beta$-block-encoding of the arrowhead matrix for the vector
\begin{align}
    \mathbf{u} = (\mathbf{u}^\pp{0}; \mathbf{u}^\pp{1}; \ldots; \mathbf{u}^\pp{r-1}) = (A^{\pp{0}\top}\mathbf{y}; A^{\pp{1} \top}\mathbf{y}; \ldots; A^{(r-1)\top}\mathbf{y})
\end{align}
The terminology ``\((x, y)\)-block-encoding'' \cite{Gilyn2019} refers to a unitary block-encoding with subnormalisation \(x\) and approximation error at most \(y\). If only \(x\) is specified, it refers to the subnormalisation factor alone, and $y$ is presumed to be 0. 
That is, our construction queries a unitary $U_{\operatorname{Arw}(\mathbf{u})}$, where a certain block flagged by an ancilla register being in $\ket{\bar{0}}$ is equal to the block diagonal arrowhead matrix of \cref{arrw_multi_vector}. Subscripts on registers are included for convenience, consistent with \cref{section:access_model}. 
\begin{align}
    (I_{\rm cone} \otimes I_{\rm col} \otimes \bra{\bar{0}}_{\rm anc}) U_{\operatorname{Arw}(\mathbf{u})}(I_{\rm cone} \otimes I_{\rm col} \otimes \ket{\bar{0}}_{\rm anc}) = \frac{1}{2\beta}\sum_{k=0}^{r-1} \ket{k}\bra{k}_{\rm cone} \otimes \operatorname{Arw}\left( A^{\pp{k}\top} \mathbf{y} \right)_{\rm col}
\end{align}
The unitary $U_{\operatorname{Arw}(\mathbf{u})}$ can in turn be implemented using the data access oracles specified in \cref{section:access_model}---as per \cref{BE_Arrohead_matrix_theorem}, it can be implemented using 1 call to each of $O_R$, $O^{\dagger}_R$, $O_{\mathbf{y}}$, and $O^{\dagger}_{\mathbf{y}}$, plus $\bigO(\log(\bar{n}))$ other single- and two-qubit gates. 

The quantities $\mathcal{Z}^\pp{k}$, $\mathbf{p}^\pp{k}$ defined in \cref{eq:Z^\pp{k}} and \cref{eq:p^\pp{k}} can be expressed in terms of $\mathbf{u}^\pp{k}$ as
\begin{align}
    \mathcal{Z}^\pp{k} &= \Tr(e^{-\mathbf{u}^\pp{k}}) \\
    \mathbf{p}^\pp{k} &= \frac{e^{-\mathbf{u}^\pp{k}}}{\mathcal{Z}^\pp{k}}
\end{align}
Here, we recall that exponentiation  is understood in terms of the Jordan frame of the vector $\mathbf{u}^\pp{k}$
\begin{align}
    \mathbf{u}^\pp{k} &= \lambda_+^\pp{k} \mathbf{c}_+^\pp{k} + \lambda_-^\pp{k} \mathbf{c}_-^\pp{k} \\
    \mathbf{p}^\pp{k} &= \frac{e^{-\lambda_+^\pp{k}} \mathbf{c}_+^\pp{k} + e^{-\lambda_-^\pp{k}} \mathbf{c}_-^\pp{k}}{e^{-\lambda_+^\pp{k}}+e^{-\lambda_-^\pp{k}}} \,,
\end{align}
where $\lambda_{\pm}^\pp{k}$ are the eigenvalues and $\mathbf{c}_{\pm}^\pp{k}$ are the eigenvectors of $\mathbf{u}^\pp{k}$.
% for $k=0,\ldots, r-1$,  where 
% \begin{align}
%     \mathcal{Z}^\pp{k} = \Tr(e^{-A^{\pp{k}\top} \mathbf{y}})
% \end{align}
% is a normalisation constant that ensures that $\Tr(\mathbf{p}^\pp{k}) = 1$. 
% chosen so that the vector of traces $(\Tr(\mathbf{p}^0),\ldots, \Tr(\mathbf{p}^{r-1}))$ formed a probability distribution---the classical Gibbs distribution over cone indices. 
By construction we have that $\Tr(\mathbf{p}^\pp{k})=1$; however, it does not hold that $\nrm{\mathbf{p}^\pp{k}} = 1$ under the standard Euclidean vector norm, which is the relevant one when working with quantum states. 
To rectify this, it will be useful to define
\begin{align}
    \sqrt{\mathbf{p}^\pp{k}} = \frac{e^{-\mathbf{u}^\pp{k}/2}}{\sqrt{\mathcal{Z}^\pp{k}}} = \frac{e^{-\lambda_+^\pp{k}/2}\mathbf{c}_{+}^\pp{k}+ e^{-\lambda_-^\pp{k}/2}\mathbf{c}_{-}^\pp{k}}{\sqrt{e^{-\lambda_+^\pp{k}}+e^{-\lambda_-^\pp{k}}}}
\end{align}
We note the identity $\sqrt{\mathbf{p}^\pp{k}} \circ \sqrt{\mathbf{p}^\pp{k}} = \mathbf{p}^\pp{k}$. We also note that, due to the orthogonality of $\mathbf{c}_{\pm}^\pp{k} $ and the fact that $\nrm{\mathbf{c}_{\pm}^\pp{k}} = 1/\sqrt{2}$, we have
\begin{align}
    \nrm{\sqrt{\mathbf{p}^\pp{k}}} =  \frac{1}{\sqrt{2}}\,,
\end{align}
independent of $k$. 

Having observed this, we can define 
the normalized quantum state
\begin{align}\label{eq:ket_sqrt_2p^\pp{k}}
    \ket{\sqrt{2\smash{\mathbf{p}^\pp{k}}}}_{\rm col} = \sum_{i=0}^{\bar{n}-1}\sqrt{2}(\sqrt{\vphantom{2}\smash{\mathbf{p}^\pp{k}}})_{i} \ket{i}_{\rm col}
\end{align}
to have its amplitudes proportional to the entries of $\sqrt{\mathbf{p}^\pp{k}}$.

Now, consider the state which is a (weighted) superposition over these states for different values of $k$:
\begin{equation}
    \ket{\sqrt{2\mathbf{x}}}_{\rm cone, col} = \sum_{k=0}^{r-1} \sqrt{\frac{\mathcal{Z}^\pp{k}}{\mathcal{Z}}}\ket{k}_{\rm cone} \ket{\sqrt{2\smash{\mathbf{p}^\pp{k}}}}_{\rm col}\,.
\end{equation}
If this state can be prepared, then measuring the cone register yields sample $k$ with probability $\mathcal{Z}^\pp{k}/\mathcal{Z}$, the desired output of the cone index Gibbs sampler oracle. We will prepare this state by transforming the eigenvalues $\lambda^\pp{k}_{\pm}$ of the arrowhead matrix $\operatorname{Arw}(\mathbf{u})$ with QSVT, using a polynomial approximation of the exponential function. 

First, we need to have an approximation to the minimum eigenvalue $\min_k \lambda_{-}^\pp{k}$. 
We will accomplish this by using \cref{theorem:min_finding} (reproduced from \cite[Theorem 8]{lin2020near}), considering $\operatorname{Arw}(\mathbf{u})$ as our Hamiltonian and considering $$\frac{1}{\sqrt{r}}\sum^{r-1}_{k=0}\ket{k}_{\rm cone} \ket{\bar{0}}_{\rm col}$$ as the initial state, which can be easily prepared with $\bigO(\log(r))$ gate complexity. We also have a guarantee that this state has overlap $1/\sqrt{2r}$
%\alexnote{maybe better to phrase in terms of ``overlap''?}
to the eigenstate of $\operatorname{Arw}(\mathbf{u})$ with minimum eigenvalue---this is due to the fact that for any Jordan frame with orthogonal normalized (in Euclidean norm) eigenvectors $\ket{\sqrt{2}\mathbf{c}_{\pm}}$, we have $\ket{\bar{0}} = \frac{1}{\sqrt{2}}\left(\ket{\sqrt{2}\mathbf{c}_+} + \ket{\sqrt{2}\mathbf{c}_-}\right)$. Then, using \cref{theorem:min_finding}, the cost of finding the minimum eigenvalue of the Arrowhead matrix, with $\eta_\lambda$ precision and $1-\eta_{\rm fail}$ probability, is $\bigOt(\frac{\beta \sqrt{2r}\log(1/\eta_{\rm fail})}{\eta_{\lambda}})$ queries to the block-encoding of $\operatorname{Arw}(\mathbf{u})$ (with normalisation factor $2\beta$), and $\bigOt(\frac{\beta \sqrt{2r}\log(1/\eta_{\rm fail})\log(rm\bar{n})}{\eta_{\lambda}})$ other single- and two-qubit gates.\footnote{This $\log(r \bar n  m)$ factor is a loose upper bound; the true cost should be closer to $\bigO(\log(\bar n m))$, but we include the extra $\log r$ to capture any other possible hidden overhead.}

\begin{lemma}[Preparing \(\ket{\sqrt{2\mathbf{x}}}\) with QSVT ]\label{lem:preparing_ket_sqrt_2x}
  Let \(A^\pp{0},\dots,A^\pp{r-1},\theta\)
  define an instance of the SOCP feasibility problem, and
  let \(\mathbf{y}\in\mathbb{R}^{m}\) satisfy \(\|\mathbf{y}\|_{1}= \beta\).
  Suppose a value \(\tilde{\lambda}_{\min}\) satisfying $\lambda_{\rm min} - \eta_{\lambda}\leq \tilde{\lambda}_{\rm min} \leq \lambda_{\rm min} + \eta_{\lambda}\,,
$
    where $\lambda_{\rm min} = \min_{k \in [r]} \lambda_{-}^k $ 
  is known (with accuracy parameter \(\eta_{\lambda}=1/2\)).
  Given an error parameter \(\varepsilon_{\mathbf{x}}>0\), there exists a quantum circuit that outputs a state \(\rho_{\mathrm{out}}\) obeying
  \[
     D\!\bigl(\rho_{\mathrm{out}},
              \ket{\sqrt{2\mathbf{x}}}\!\bra{\sqrt{2\mathbf{x}}}
              \bigr)
     \;\le\;\varepsilon_{\mathbf{x}},
  \]
  where \(D(\cdot,\cdot)\) is the trace distance.
  The procedure uses
\begin{equation}
    \bigOt\left(\beta \sqrt{r} \log^2\left(\frac{1}{\varepsilon_{\mathbf{x}}} \right)\right)
\end{equation}
  calls to the block-encoding oracle
  \(U_{\operatorname{Arw}(\mathbf{u})}\),
  and an asymptotically equivalent number of additional one- and two-qubit gates, up to a factor $\bigO(\log(r\bar{n}m))$.
\end{lemma}
\begin{proof}
We use QSVT to exponentiate the eigenvalues of the arrowhead matrix, we apply the resulting matrix to an equal superposition of the corresponding eigenvectors, and finally we use amplitude amplification to boost the success probability. We begin by preparing an equal superposition over the identity vector of each of the $k$ cones, using $\bigO(\log(r))$ gates. We identify the notation $\ket{\mathbf{e}^\pp{k}}_{\rm cone,col}:= \ket{k}_{\rm cone}\ket{\bar{0}}_{\rm col}$, since the identity element has all its amplitude on the $0$ entry of the $k$-th cone. That is, we prepare the state
\begin{align}
    \frac{1}{\sqrt{r}}\sum^{r-1}_{k=0}\ket{\mathbf{e}^\pp{k}}_{\rm cone, col} =  \frac{1}{\sqrt{r}}\sum^{r-1}_{k=0}\ket{k}_{\rm cone} \ket{\bar{0}}_{\rm col} = \frac{1}{\sqrt{2r}}\sum^{r-1}_{k=0}\ket{k}_{\rm cone} \left(\ket{\sqrt{2}\mathbf{c}_+^{\pp{k}}} + \ket{\sqrt{2}\mathbf{c}_-^{\pp{k}}} \right)_{\rm col}\,,
\end{align}
where the final decomposition follows from the definition of the eigenvectors and the fact that they add to the identity element. 
Ideally, we would then implement the map:
\begin{align}
  &\frac{1}{\sqrt{2r}}\sum^{r-1}_{k=0}\ket{k}_{\rm cone} \left(\ket{\sqrt{2}\mathbf{c}_+^{\pp{k}}} + \ket{\sqrt{2}\mathbf{c}_-^{\pp{k}}} \right)_{\rm col} \longmapsto{}\nonumber\\
    &   \frac{1}{\sqrt{2r}}\sum^{r-1}_{k=0}\ket{k}_{\rm cone}\left(\sqrt{e^{-(\lambda^\pp{k}_+ -\lambda_{\min })}} \ket{\sqrt{2}\mathbf{c}^\pp{k}_+}+ \sqrt{e^{-(\lambda^\pp{k}_{-}-\lambda_{\min })}} \ket{\sqrt{2}\mathbf{c}^\pp{k}_-}\right)_{\rm col}\ket{0}_{\rm flag}
    \\
   & \qquad  +\ket{\text{garbage}^{\prime \prime}}_{\rm cone, col}\ket{1}_{\rm flag} \nonumber
\end{align}
 % where $\ket{\mathbf{e}^\pp{k}}_{\rm cone,col}:= \ket{k}_{\rm cone}\ket{\bar{0}}_{\rm col}$ only has support on the eigenvectors of interest,
%\alexnote{$\bar{n}$ is dimension of vector not the number of qubits} which correspond to the eigenvalues $\lambda^\pp{k}_+,\lambda^\pp{k}_-$. 

We approximate the map by replacing the exponential function with a polynomial approximation. In particular, we approximate the function $f(x) = e^{-2\beta x} / 4$, for $x\in [0,1]$,\footnote{The denominator 4 comes from requiring the domain of the exponential function to be between $[-1/4,0]$. In particular, this requirement is sourced from \cref{QSVT_poly_eig_trans_theorem}, which requires $|P(x)|\leq \frac{1}{2}$, and to prevent any possible deviation from our implementation, we make it satisfy $|P(x)|\leq \frac{1}{4}$.} with a $\delta_{exp}$ approximate polynomial $\tilde{P}(x)$ from \cref{Lemma_poly_approx}. The reason to subtract $\lambda_{\min}$ is to shift the domain of the exponentiation function from $[-1,1]$ to $[0,1]$. This is necessary in order to obtain a good polynomial approximation to the exponential function, while still obeying the QSVT constraints of \cref{Lemma_poly_approx}.

We use \cref{QSVT_poly_eig_trans_theorem} applied to the block-encoding $\text{Arw}(\textbf{u})-(\tilde{\lambda}_{\min}-\eta_{\lambda})I$, where the terms are combined with the linear combination of unitaries method,
% , where the subtraction is performed at constant overhead using coherent arithmetic, 
and thus the subnormalisation of the final block-encoding increases to $4\beta$, as $|\tilde{\lambda}_{\min}-\eta_{\lambda}|\leq 2\beta $. The inclusion of $\eta_\lambda$ ensures that the argument remains in the range $[0,1]$, even if the estimate of $\lambda_{\min}$ is inaccurate. 

In summary, we implement the map 
\begin{align}
\ket{\mathbf{e}^\pp{k}}\ket{0} &\mapsto \tilde{P} \left( \frac{1}{4\beta } (\text{Arw}(\textbf{u}^\pp{k})-(\tilde{\lambda}_{\min }-\eta_{\lambda})I) \right)  \ket{\mathbf{e}^\pp{k}}\ket{0} +\ket{\operatorname{garbage}^{\prime \prime \prime}}\ket{1}\\
&\approx \frac{1}{4\sqrt{2r}}\sum^{r-1}_{k=0}\left(\sqrt{e^{-(\lambda^\pp{k}_+ -(\tilde{\lambda}_{\min }-\eta_\lambda))}} \ket{\sqrt{2}\mathbf{c}^\pp{k}_+}+ \sqrt{e^{-(\lambda^\pp{k}_{-}-(\tilde{\lambda}_{\min }-\eta_\lambda))}} \ket{\sqrt{2}\mathbf{c}^\pp{k}_-}\right)\ket{0}\\
&\quad +\ket{\text{garbage}^{\prime \prime\prime}}\ket{1} \nonumber
\end{align}
where the equality is approximate to additive precision $\delta_{exp}$ in each amplitude,
using $\bigO(\operatorname{deg}(\tilde{P}))=\bigO(4\beta  \log (1 / \delta_{exp}))$ calls to ($4\beta$,0)-block-encoding of $\text{Arw}(\textbf{u}^\pp{k})-(\tilde{\lambda}_{\min}-\eta_{\lambda})I$ and its inverse. The number of additional single- and two-qubit gates for creating the LCU and performing the QSVT sequence is asymptotically similar to the number of queries, up to a factor of at most $\bigO(\log(\bar{n}m))$ (for reflection about the $\bigO(\log(\bar{n}m))$ block-encoding ancillas). 
% $\bigO(\log_2(\bar{n}))$ \isanote{Should this be r instead to make the superposition ?}other gates.

Finally, we use fixed-point amplitude amplification to boost the success probability, while controlling the approximation errors. Observe that the estimate of $\tilde{\lambda}_{\min}$ is only accurate to an additive error of $\xi := \tilde{\lambda}_{\min} - \lambda_{\min}$ with $|\xi| \leq \eta_\lambda = 0.5$. As a result, our attempt to block-encode $e^{-0.5(\lambda-(\lambda_{\min} - \eta_\lambda))}$ would actually block-encode $e^{-0.5(\lambda-(\lambda_{\min} + \xi - \eta_\lambda))} = e^{-0.5(\lambda-(\lambda_{\min} - \eta_\lambda))}e^{0.5 \xi}$. The unknown factor $e^{0.5 \xi}$ is the same for all of the eigenvectors, and can be absorbed into the subnormalisation factor without too much loss. In particular, we can lower bound the amplitude of obtaining $\ket{0}$ in the ancilla register by dropping all the terms in the sum except the one where $\lambda = \lambda_{\min}$---this gives a lower bound of $\frac{e^{-0.5(\eta_{\lambda} -\xi)}}{4\sqrt{2r}} \geq \frac{1}{8\sqrt{2r}}$ since $|\xi|\leq \eta_\lambda$ and $\eta_\lambda = 0.5$. We can boost the amplitude of the state flagged by $\ket{0}$ to at least $\sqrt{1-\omega_{AA}^2}$ using fixed-point amplitude amplification with $\bigO(8\sqrt{2r}\log(1/\omega_{AA}))$ calls to the procedure that prepares the state (plus an asymptotically equivalent number of other single- and two-qubit gates for the QSVT single-qubit rotations and reflections, up to a factor $\bigO(\log(\bar{n}m))$). Ultimately, there are two sources of error; the error in fixed-point amplitude amplification ($\omega_{AA}$), and the error in the approximation of the exponential function ($\delta_{exp}$). These can be accounted for using \cref{Lemma:ImperfectAABE}. The resulting trace-distance in the output state is bounded by $32\sqrt{2r\delta_{exp}}\log(1/\omega_{AA}) + \omega_{AA} + \sqrt{2\omega_{AA}} $. We set
\begin{align}
    \omega_{AA} &= \frac{\varepsilon_{\mathbf{x}}^2}{32} \\
    \delta_{exp} &= \frac{\varepsilon_{\mathbf{x}}^2}{8192r\log^2\left(\frac{32}{\varepsilon_{\mathbf{x}}^2} \right)}
\end{align}
to ensure the trace distance of the amplified state is bounded by $\varepsilon_{\mathbf{x}}$. The resulting circuit makes 
\begin{equation}
    \bigOt\left(\beta \sqrt{r} \log^2\left(\frac{1}{\varepsilon_{\mathbf{x}}} \right)\right)
\end{equation}
calls to the block-encoding oracle
  \(U_{\operatorname{Arw}(\mathbf{u})}\),
  and an asymptotically equivalent number of additional one- and two-qubit gates.

\end{proof}

\begin{corollary}\label{cor:C_samp_quantum}
    There is a quantum procedure that implements the cone index Gibbs sampler oracle of \cref{oracle:cone_index_Gibbs_sampling}---that is, it generates $T'$ samples from a joint distribution at most $\zeta$-far in total variation distance from $T'$ i.i.d.~samples from the ideal distribution---while incurring cost
    \begin{align}
        C_{\rm samp}(m,r,n,s,\beta,\zeta, T') = \bigOt(\beta \sqrt{r} T' \log^2(1/\zeta) )
    \end{align} 
    queries to oracles $O_R$,  $O^\dagger_R$, $O_{\mathbf{y}}$, $O^\dagger_{\mathbf{y}}$, and an asymptotically equivalent number of single- and two-qubit gates up to a factor $\bigO(\log(r\bar{n}m) )$. 
\end{corollary}
\begin{proof}
    We run the minimum-finding procedure to obtain an estimate $\tilde{\lambda}_{\rm min}$, which is correct up to $\bigO(1)$ additive error except with probability $\eta_{\rm fail}$.  We use this estimate to produce $T'$ copies of a state approximating $\ket{\sqrt{2\mathbf{x}}}$, with error parameter $\varepsilon_{\mathbf{x}}$. Conditioned on a fixed value of $\tilde{\lambda}$ that meets the additive error bound,  each sample is drawn independently from a distribution that is at most  $\varepsilon_{\mathbf{x}}$-far from the ideal distribution in total variation distance, and thus the joint distribution over the $T'$ samples is at most $T'\varepsilon_{\mathbf{x}}$-far. The actual distribution over the samples is a mixture of the distribution obtained from each possible value of $\tilde{\lambda}_{\rm min}$, but in the case the additive error bound fails, the total variation distance is still upper bounded by 1.
    
    Thus the overall total variation distance, choosing $\varepsilon_{\mathbf{x}} = \zeta/(2T')$, $\eta_{\rm fail}=\zeta/2$, is upper bounded by $\eta_{\rm fail} +T'\varepsilon_{\mathbf{x}}=\zeta$. The 
    number of queries to $\operatorname{Arw}(\mathbf{u})$ is given by the query complexity to find $\tilde{\lambda}_{\min}$ plus the complexity to prepare $\ket{\sqrt{2\mathbf{x}}}$ (\cref{lem:preparing_ket_sqrt_2x}), which reduces to preparing $U_{\operatorname{Arw}(\mathbf{u})}$ (\cref{BE_Arrohead_matrix_theorem}), multiplied by $T'$. The latter contribution dominates.
\end{proof}

\subsection{Quantum implementation of sampled violated constraint search oracle}

Now, we explain how, given $T'$ samples $k_0,k_1,\ldots,k_{T'-1}$, the quantum algorithm can implement the search for a violated constraint as in \cref{oracle:sampled_constraint_search}. This subroutine again uses $U_{\operatorname{Arw}(\mathbf{u})}$. It will also query a unitary $U_{\hat{V}}$ which is an approximate $(3,0)$-block-encoding of a matrix $\hat{V}$ containing the violation amounts $\hat{v}_j$ from \cref{eq:hat_v_j} on its diagonal. 
\begin{align}
    \left\lVert 3 (I_{\rm row} \otimes \bra{\bar{0}}) U_{\hat{V}} (I_{\rm row} \otimes \ket{\bar{0}}) - \hat{V} \right\rVert  \leq \nu \qquad \text{where} \qquad 
    \hat{V} = \sum_{j=0}^{m-1} \hat{v}_j \ket{j}\bra{j}
\end{align}

We explain how $U_{\hat{V}}$ can be constructed out of the data access oracles. The key is observing how the quantities $\hat{v}_j$ can be expressed in terms of overlaps between quantum states we can prepare. Recall that $\hat{v}_j = \frac{1}{T'}\sum_{h = 0}^{T'-1} \hat{v}_{j,h}$, where $\hat{v}_{j,h} = A^\pp{k_h}_{j,:} \mathbf{p}^\pp{k_h} - b_j$. We may rewrite this using the following lemma. 
\begin{lemma}\label{Lemma:ComputeViolation}
     For any $k \in [r]$, we have 
    \begin{align}
       A^\pp{k}_{j,:} \mathbf{p}^{\pp{k}} = \frac{1}{2}\bra{\sqrt{2\mathbf{p}^\pp{k}}} \operatorname{Arw}(A_{j,:}^{\pp{k}\top})\ket{\sqrt{2\smash{\mathbf{p}^\pp{k}}}}\,,
    \end{align}
    where $\ket{\sqrt{2 \mathbf{p}^\pp{k}}}$ is given in \cref{eq:ket_sqrt_2p^\pp{k}}. 
\end{lemma}
\begin{proof}
    Generally, let $\mathbf{a}$, $\mathbf{b}$, be two vectors in the second-order cone Euclidean Jordan algebra, with $\mathbf{b} \succeq 0$ (so that the state $\sqrt{2\mathbf{b}}$ is well defined). We may write
\begin{align}
    \mathbf{a}^{\top} \mathbf{b} &= \frac{1}{2} \mathbf{a}^{\top} (2\mathbf{b}) = \frac{1}{2}\mathbf{a}^{\top} (\sqrt{2\mathbf{b}} \circ  \sqrt{2\mathbf{b}}) \tag{$\sqrt{\mathbf{u}}\circ \sqrt{\mathbf{u}} = \mathbf{u}$ if $\mathbf{u} \succeq 0$}\\
    &= \frac{1}{2} \mathbf{a}^\top  \operatorname{Arw}(\sqrt{2\mathbf{b}}) \sqrt{2\mathbf{b}} \tag{$\operatorname{Arw}(\mathbf{u})\mathbf{v}=\mathbf{u} \circ \mathbf{v}$}\\
    &= \frac{1}{2} (\operatorname{Arw}(\sqrt{2\mathbf{b}})\mathbf{a})^\top   \sqrt{2\mathbf{b}} \tag{$\operatorname{Arw}(\mathbf{u})^\top = \operatorname{Arw}(\mathbf{u})$} \\
    &= \frac{1}{2} (\operatorname{Arw}(\mathbf{a})\sqrt{2\mathbf{b}} )^\top   \sqrt{2\mathbf{b}} \tag{$ \operatorname{Arw}(\mathbf{u})\mathbf{v}=\operatorname{Arw}(\mathbf{v})\mathbf{u}$} \\
    &= \frac{1}{2} \sqrt{2\mathbf{b}}^\top   \operatorname{Arw}(\mathbf{a})\sqrt{2\mathbf{b}}  \tag{$\operatorname{Arw}(\mathbf{u})=\operatorname{Arw}(\mathbf{u})^\top$} 
\end{align}
Now, it suffices to take $\mathbf{a} = A^{\pp{k}\top}_{j,:}$, and $\mathbf{b} = \mathbf{p}^\pp{k}$, and note that the state $\ket{\sqrt{2\smash{\mathbf{p}^\pp{k}}}}$ is the normalized state whose entries are exactly equal to those of the vector $\sqrt{2\mathbf{p}^\pp{k}}$. 
\end{proof}

Let $U_{\sqrt{2\mathbf{p}^\pp{k}}}$ be a unitary that maps $\ket{k}_{\rm cone} \ket{\bar{0}}_{\rm col} \mapsto \ket{k}_{\rm cone} \ket{\sqrt{2\smash{\mathbf{p}^\pp{k}}}}_{\rm col}$. This can be implemented with the following complexity. 
\begin{lemma}[Implementing $U_{\sqrt{2 \mathbf{p}^\pp{k}}}$]

The unitary $U_{\sqrt{2 \mathbf{p}^\pp{k}}}$ can be implemented up to error $\epsilon_{\mathbf{p}}$ (in spectral norm) using $\bigO(\beta' \log(\beta'/\epsilon_{\mathbf{p}})\log(1/\epsilon_{\mathbf{p}}))$ calls to  $O_R$,  $O^{\dagger}_R$, $O_{\mathbf{y}}$, and $O^{\dagger}_{\mathbf{y}}$, where $\beta' = \max(1,\beta)$, and an asymptotically equivalent number of single- and two-qubit gates, up to a factor of $\bigO(\log(\bar{n}))$. 
\end{lemma}
\begin{proof}
    The unitary $U_{\sqrt{2 \mathbf{p}^\pp{k}}}$ is constructed via a sequence of steps: (i) we note the ability to apply a unitary $U_{\rm od}$ which is a $(\beta,0)$-block-encoding the offdiagonal portion of $\sum_k \ket{k}\bra{k} \otimes \operatorname{Arw}(\mathbf{u}^\pp{k})$ (\cref{lemma_off_diag}), (ii) applying QSVT to $U_{\rm od}$ for a well-chosen polynomial, we construct a unitary $U_{\sigma}$ which has the property that $U_{\sigma}\ket{k}\ket{\bar{0}}\ket{0} = C^\pp{k} \ket{k} \ket{\sqrt{2\smash{\mathbf{p}^\pp{k}}}}\ket{0} + \ket{\rm garbage}\ket{1}$, where the constant $C^\pp{k} \geq 1/2$ for all $k$, (iii) we form $U_{\sqrt{2\mathbf{p}^\pp{k}}}$ by wrapping fixed-point amplitude amplification around $U_{\sigma}$ to boost the amplitude of each  $\ket{k}\ket{\sqrt{2\smash{\mathbf{p}^\pp{k}}}}\ket{0}$  to 1. In the remainder of the proof, we explain each of these steps in more detail and justify its correctness. 
    
    We recall some notation. We let $\mathbf{c}_\pm^\pp{k} = \frac{1}{2}(1; \pm\frac{\vec{u}^\pp{k}}{\nrm{\vec{u}^\pp{k}}})$ be the eigenvectors (in the second-order cone Euclidean Jordan algebra sense) of the vector $\mathbf{u}^{\pp{k}} = A^{\pp{k} \top} \mathbf{y}$, and we let $\ket{\sqrt{2}\mathbf{c}_\pm^\pp{k}}_{\rm col}$ denote normalized quantum states with amplitudes  proportional to the entries of $\sqrt{2}\mathbf{c}_{\pm}^\pp{k}$. We recall that for any $k$ we may decompose the state $\ket{k} \ket{\bar{0}} \equiv \ket{\mathbf{e}^\pp{k}}$ into an equal superposition of the eigenvectors, that is
    \begin{align}
        \ket{\bar{0}}_{\rm col} = \frac{1}{\sqrt{2}}\left(\ket{\sqrt{2}\mathbf{c}_+^\pp{k}} + \ket{\sqrt{2}\mathbf{c}_-^\pp{k}}\right)_{\rm col}\,.
    \end{align}
    Now we continue with the construction. 
\begin{enumerate}[(i)]
        \item Let $U_{\rm od}$ denote the $(\beta,0)$-block-encoding of the offdiagonal portion of the arrowhead matrix of $\mathbf{u}^\pp{k}$; specifically, $U_{\rm od}$ satisfies
    \begin{align}
        &(I_{\rm cone} \otimes I_{\rm col} \otimes \bra{\bar{0}}_{\rm anc})U_{\rm od}(I_{\rm cone} \otimes I_{\rm col} \otimes \ket{\bar{0}}_{\rm anc}) \nonumber \\
        ={}& \frac{1}{\beta} \sum_{k=0}^{r-1} \ket{k}\bra{k}_{\rm cone} \otimes \left(\ket{\bar{0}}\bra{\mathbf{u}^\pp{k}}+\ket{\mathbf{u}^\pp{k}}\bra{\bar{0}} - 2u_0^\pp{k}\ket{\bar{0}}\bra{\bar{0}}\right)_{\rm col} =: \frac{A_{\rm od}}{\beta}
    \end{align}
    where the right-hand side defines the matrix $A_{\rm od}$. 
    We may note that $\ket{k} \otimes \ket{\sqrt{2} \mathbf{c}_{\pm}^\pp{k}}$ are normalized eigenvectors of $A_{\rm od}/\beta$, corresponding to eigenvalues $\pm \nrm {\vec{u}^\pp{k}}/\beta$. Thus, by applying QSVT to $U_{\rm od}$, we can apply polynomial transformations to these eigenvalues. We can construct $U_{\rm od}$ using \cref{lemma_off_diag} with one query to each of $O_R$, $O^{\dagger}_R$,  $O_\mathbf{y}$ and $O^{\dagger}_\mathbf{y}$, plus $\bigO(\log(\bar{n}))$ additional single- and two-qubit gates.

    \item For the QSVT, we will use the sigmoid (logistic) function. Let $\sigma(x) = 1/(1 + e^{\beta x})$ be the (stretched and inverted) sigmoid function, and note that $\sigma(x) = \frac{1}{2}(1-\tanh(\frac{\beta x}{2}))$. Let $p(x)$ be a polynomial approximation for $\tanh(\frac{\beta x}{2})/2$, which by \cref{lem:polynomial_for_tanh} can achieve error $\varepsilon_{\rm tanh} < 1/2$ with degree $d=\bigO(\beta'\log(\beta'/\varepsilon_{\rm tanh}))$. We have that $|p(x)| \leq  1$ when $|x|\leq 1$, so by QSVT, we can construct a $(1,0)$-block-encoding of the matrix $p(A_{\rm od}/\beta)$ using $d$ calls to $U_{\rm od}$ and its inverse. Then, by  linear combination of unitaries with a $(1,0)$-block-encoding of the identity $I$, we can construct a unitary $U_{\sigma}$ which is a $(3/2,0)$-block-encoding of the matrix $\frac{I}{2} - p(A_{\rm od}/\beta)$. The unitary $U_{\sigma}$ is thus a $(3/2,\epsilon_{\rm tanh})$-block-encoding of $\sigma(A_{\rm od}/\beta)$.
    
    Consider the initial state $\ket{\psi_k} = \ket{k}_{\rm cone} \otimes \ket{\bar{0}}_{\rm col}$. The matrix $A_{\rm od}$ has the property that 
    \begin{align}
        \sigma(A_{\rm od}/\beta) \ket{\psi_k} 
        &= \sigma(A_{\rm od}/\beta)\frac{1}{\sqrt{2}}\ket{k} \otimes \left(\ket{\sqrt{2}\mathbf{c}_+^\pp{k}}+\ket{\sqrt{2}\mathbf{c}_-^\pp{k}} \right) \\
        &= \frac{1}{\sqrt{2}}\ket{k} \otimes \left(\frac{1}{1+e^{\nrm{\vec{u}^\pp{k}}}}\ket{\sqrt{2}\mathbf{c}_+^\pp{k}}+\frac{1}{1+e^{-\nrm{\vec{u}^\pp{k}}}}\ket{\sqrt{2}\mathbf{c}_-^\pp{k}} \right) \\
        &=\frac{e^{\nrm{\vec{u}^\pp{k}}/2}}{\sqrt{2}(1+e^{\nrm{\vec{u}^\pp{k}}})}\ket{k} \otimes \left(e^{-\nrm{\vec{u}^\pp{k}}/2}\ket{\sqrt{2}\mathbf{c}_+^\pp{k}}+e^{\nrm{\vec{u}^\pp{k}}/2}\ket{\sqrt{2}\mathbf{c}_-^\pp{k}} \right) \\
        &= C^\pp{k} \ket{k} \otimes \ket{\sqrt{2 \mathbf{p}^\pp{k}}}
    \end{align}
    where
    \begin{align}
        C^\pp{k} = \frac{ \sqrt{1
        +e^{\nrm{2 \vec{u}^\pp{k}}}}}{\sqrt{2}(1+e^{\nrm{\vec{u}^\pp{k}}})} \geq \frac{1}{2}\,,
    \end{align}
    for all $k$.
    \item Next, we perform fixed-point amplitude amplification. For any fixed value of $k$, the unitary $U_{\sigma}$, when acting on the state $\ket{k}\ket{\bar{0}}$, prepares a $2\varepsilon_{\rm tanh}/3$ approximation to the subnormalized state $(2C^\pp{k}/3)\ket{k}\ket{\sqrt{2\smash{\mathbf{p}^\pp{k}}}}$ when one projects on the block-encoding ancillas being $\ket{\bar{0}}$. Since $C^\pp{k} \geq 1/2$ and $\varepsilon_{\rm tanh} < 1/2$, the amplitude of this state is bounded by a $k$-independent lower bound $\Lambda = \Omega(1)$. Fixed-point amplitude amplification prepares the state $\ket{k}\ket{\sqrt{2\smash{\mathbf{p}^\pp{k}}}}$ up to error $\delta + \varepsilon_{\rm tanh}/C^\pp{k} \leq \delta + 2\varepsilon_{\rm tanh}$ using $\bigO(\log(1/\delta))$ calls to $U_{\sigma}$, and $\bigO(\log(1/\delta))$ reflections about the initial state $\ket{\psi_k}$. 
    
    However, we wish for $k$ to vary, and for the unitary $U_{\sqrt{2 \mathbf{p}^\pp{k}}}$ to act on superpositions of different values of $k$. We note that $A_{\rm od}$ has a block-diagonal structure with respect to the cone register. Thus, if we begin in the state $\ket{\psi_k} = \ket{k}\ket{\bar{0}}$ and apply sequences of $U_{\rm od}$, we will not leave the subspace where the first register is $\ket{k}$. Consequently, the outcome is unmodified if in fixed-point amplitude amplification, we replace the reflection about $\ket{\psi_k}$ with a reflection about the image of the projector $I_{\rm cone} \otimes \ket{\bar{0}}\bra{\bar{0}}_{\rm col}$. Since this projector is independent of $k$, we conclude that the resulting unitary $U_{\sqrt{2 \mathbf{p}^\pp{k}}}$ has the property that for any initial state $\ket{k}\ket{\bar{0}}$, it prepares $\ket{k}\ket{\sqrt{2\smash{\mathbf{p}^\pp{k}}}}$ up to error $2\epsilon_{\rm tanh} + \delta$. 
    \end{enumerate}
    By taking $\epsilon_{\rm tanh} = \epsilon_{\mathbf{p}}/4$ and $\delta = \epsilon_{\mathbf{p}}/2$, we have the claimed error bound. The total complexity is $\bigO(\log(1/\delta)\beta\log(\beta/\epsilon_{\mathrm{tanh}}))$ (if  $\beta > 1$) or $\bigO(\log(1/\delta)\log(1/\epsilon_{\mathrm{tanh}}))$ (if $\beta < 1$) calls to the oracles $O_R$, $ O^{\dagger}_R$, $O_{\mathbf{y}}$, and $O^{\dagger}_{\mathbf{y}}$ and requires $\bigO(\log(\bar{n}))$ number of additional single- and two-qubit gates per oracle call. 
    %and a similar number of other gates. 
\end{proof}

 Let $U_{A}$ be a $(2,0)$-block-encoding of the operator
 \begin{align}
    \sum_{j=0}^{m-1}\sum_{k=0}^{r-1} \ket{j}\bra{j}_{\rm row} \otimes \ket{k}\bra{k}_{\rm cone} \otimes  \operatorname{Arw}(A^{\pp{k}\top}_{j,:})_{\rm col}
 \end{align} 
 \begin{lemma}
     The unitary $U_A$ can be implemented with 1 query to each of  $O_R$ and  $O^{\dagger}_R$, plus \newline$\bigO(\log(r\bar{n}m))$ other one- and two-qubit gates\footnote{Again, $\bigO(\log(r\bar{n}m))$ is a loose upper bound, we expect the actual cost to be $\bigO(\log(\bar{n}m))$}
 \end{lemma}
 \begin{proof}
     Note that $A^{\pp{k}\top}_{j,:} = A^{\pp{k}\top}\mathbf{y}$ with $\mathbf{y}$ the vector with a 1 in the $j$-th position and 0s elsewhere, that is, $\ket{\mathbf{y}} = \ket{j}$. For this choice of $\mathbf{y}$, the unitary that enacts $\ket{\bar{0}} \mapsto\ket{\mathbf{y}}$ is simply a particular pattern of $X$ gates, corresponding to the binary representation of $j$. Thus, using \cref{BE_Arrohead_matrix_theorem} with $\ket{\mathbf{y}}=\ket{j}$, we can implement a $(2,0)$-block-encoding of $\sum_k \ket{k}\bra{k} \otimes \operatorname{Arw}(A^{\pp{k}\top}_{j,:})$ using 1 call to $O_R, O^{\dagger}_R$ and at most $\bigO(\log(\bar{n}))$ other gates, by replacing the oracle $O_{\mathbf{y}}$ with the corresponding fixed pattern of $X$ gates. Now, we wish to control also on a value $j$ in the row register. The pattern of $X$ gates depends on $j$, but can be implemented coherently as a simple set of $\lceil \log_2(m) \rceil$ controlled-not gates that map $\ket{j}\ket{0} \mapsto \ket{j}\ket{j}$. This yields the unitary $U_A$. 
 \end{proof}
 Recall that the oracle $O_{\mathcal{T}}$ for the dataset $(k_0,\ldots,k_{T'-1})$ performs the transformation
\begin{align}
    \ket{\bar{0}}_{\mathrm{cone}} \ket{\bar{0}}_{\mathrm{sampindex}} \mapsto \frac{1}{\sqrt{T'}}\sum_{h=0}^{T'-1}\ket{k_h}_{\mathrm{cone}} \ket{h}_{\mathrm{sampindex}}
\end{align}
Then from \cref{Lemma:ComputeViolation}, we may assert that
    \begin{align}\label{eq:block_encoding_of_inner_product}
        &(I_{\rm row} \otimes \bra{\bar{0}}_{\rm cone,col,sampindex})O_{\mathcal{T}}^\dag U_{\sqrt{2 \mathbf{p}^\pp{k}}}^\dag U_A U_{\sqrt{2 \mathbf{p}^\pp{k}}}O_{\mathcal{T}} (I_{\rm row} \otimes \ket{\bar{0}}_{\rm cone,col,sampindex}) \nonumber\\
        ={}& \sum_{j=0}^{m-1}\left(\frac{1}{2T'}\sum_{h=0}^{T'-1} A_{j,:}^\pp{k_h} \mathbf{p}^\pp{k_h} \right) \ket{j}\bra{j}_{\rm row}
    \end{align}

This identity gives a method for implementing $U_{\hat{V}}$ by linear combination of unitaries.  
\begin{lemma}
    The unitary $U_{\hat{V}}$ with error parameter $\gamma$ can be implemented using one call to $O_{\mathbf{b}}$, $\bigO(\beta \log^2(1/\gamma))$ calls to $O_R$, $O^\dagger_R$, $O_{\mathbf{y}}$, $O^\dagger_{\mathbf{y}}$ and $2$ calls to the quantum lookup table $O_{\mathcal{T}}, O^\dagger_{\mathcal{T}}$ storing $(k_0,\ldots,k_{T-1})$, and an asymptotically equivalent number of other single- and two-qubit gates, up to a factor of $\bigO(\log(r\bar{n}m))$.
\end{lemma}
\begin{proof}
    Recall the action of $O_\mathbf{b}$, which sends $\ket{j}_{\rm row}\ket{0}_{\rm flag} \mapsto b_j \ket{j}_{\rm row} \ket{0}_{\rm flag} + \ket{\rm garbage}_{\rm row} \ket{1}_{\rm flag}$. This may be rewritten as 
    \begin{align}
        (I_{\rm row} \otimes \bra{0}_{\rm flag}) O_{\mathbf{b}} (I_{\rm row} \otimes \ket{0}_{\rm flag}) = \sum_{j=0}^{m-1} b_j \ket{j}\bra{j}_{\rm row}
    \end{align}
    showing that it is a $(1,0)$-block-encoding of the operator $\sum_j b_j \ket{j}\bra{j}$. The result of \cref{eq:block_encoding_of_inner_product} showed that $O^{\dag}_{\mathcal{T}} U^{\dag}_{\sqrt{2 \mathbf{p}^\pp{k}}} U_A U_{\sqrt{2 \mathbf{p}^\pp{k}}}O_{\mathcal{T}} $ is a $(2,0)$-block-encoding of $\sum_j \left(\frac{1}{T} \sum_{h} A_{j,:}^\pp{k_h} \mathbf{p}^\pp{k_h}\right)\ket{j}\bra{j}$. We can thus combine these via linear combination of unitaries (LCU) into a $(3,0)$-block-encoding of the difference between the operators, $U_{\hat{V}}$. Naively, the cost for the LCU is one controlled query to each of $O_{\mathbf{b}}$, $U_A$, $O_{\mathcal{T}}$, $O_{\mathcal{T}}^\dag$, $U_{\sqrt{2 \mathbf{p}^\pp{k}}}$ and $U_{\sqrt{2 \mathbf{p}^\pp{k}}}^\dag$ and $\bigO(1)$ other gates---however, we may observe that we need only control $U_A$ and $O_{\mathbf{b}}$, as the other gates will cancel with their adjoints when the control is off.\footnote{Furthermore, controlled $U_A$ can be performed by controlling all the calls to $U_{\operatorname{Arw}(\mathbf{u})}$, and one may note that in that implementation, $O_{\mathbf{y}}$ need not be controlled.} Each of these gates is exact except $U_{\sqrt{2 \mathbf{p}^\pp{k}}}$ and $U_{\sqrt{2 \mathbf{p}^\pp{k}}}^\dag$. To achieve error $\gamma$ on $U_{\hat{V}}$, it suffices to choose the error on $U_{\sqrt{2 \mathbf{p}^\pp{k}}}^\dag$ to be $\epsilon_{\mathbf{p}} = \gamma/6$. Thus, the overall cost is $\bigO(\log(1/\gamma)^2)$ calls to $O_R, O^\dagger_R, O_{\mathbf{y}}$ and $O^\dagger_{\mathbf{y}}$, one call to $O_{\mathbf{b}}$ and $\bigO(\log(1/\gamma)^2\log_2(\bar{n}))$ other single- and two-qubit gates. 
\end{proof}

We use this construction of $U_{\hat{V}}$, which encodes the degree of violation or satisfaction of each constraint. The next lemma shows how, by repeatedly using \( U_{\hat{V}} \), we can construct a new block-encoding \( U_{\Theta} \) that approximates the Heaviside function applied to \( \hat{V} \). When \( U_{\Theta} \) is applied to an equal superposition over all constraints, it maps non-violated constraints to approximately \( 0 \) and violated constraints to approximately \( 1 \). After applying amplitude amplification, measuring the resulting state produces one of the three possible outcomes described in \cref{def:violated_constraint_oracle}.

\begin{lemma}
    Let $\gamma = \theta/24$ and let $U_{\hat{V}}$ be a $(3,\gamma)$-block-encoding of $\hat{V}$. Then, there is a quantum procedure that implements the sampled violated constraint search oracle (\cref{oracle:sampled_constraint_search}) with failure probability at most $\eta$ using $\bigOt(\frac{\sqrt{m}}{\theta}\log^2(1/\eta))$ calls to $U_{\hat{V}}$, plus an asymptotically equivalent number of additional single- and two-qubit gates, up to a factor of $\bigO(\log(r\bar{n}m))$.
\end{lemma}
\begin{proof}
Roughly speaking, we can use QSVT to implement a unitary $U_\Theta$ which block-encodes an approximation of the Heaviside function applied to $\hat{V}$. The Heaviside function transforms the diagonal entries such that unviolated constraints ($\hat{v}_j \leq \frac{4\theta}{6}$) are mapped to $0$ and violated constraints ($\hat{v}_j > \frac{5\theta}{6}$) to $1$. Applying $U_\Theta$ to the initial state $\ket{0}\frac{1}{\sqrt{m}} \sum_j \ket{j}$ and post-selecting on the ancilla in state $\ket{0}$ samples a violated constraint, if one exists, or returns $\ket{1}$ if all constraints are satisfied. The complexity of the search can be improved quadratically using fixed point amplitude amplification.

Observe that the block-encoding is subnormalized by a factor of 3, that is, 
\begin{equation}
    \lVert \hat{V}/3 - \left(\bra{0}^{\otimes m} \otimes I\right) U_{\hat{V}} \left(\ket{0}^{\otimes m} \otimes I\right) \rVert \leq \gamma/3 .
\end{equation}
We shift the block-encoding by $-\theta/4$ (i.e., we instead examine a block-encoding of $\frac{\hat{V}}{3}-\frac{\theta}{4}I$) to move the interval $ [\frac{4\theta}{18},\frac{5\theta}{18})$ corresponding to $v_j \in \hat{V}_{\theta}$ (\cref{oracle:sampled_constraint_search}) so that it is symmetric around 0. 
This can be done with LCU incurring at most $\bigO(1)$ factor increase in the normalisation factor of the block-encoding, which can be absorbed into the big-$\bigO$ notation.
That is, the interval $[\frac{4\theta}{6},\frac{5\theta}{6}]$ after being scaled down by 3 and shifted, becomes $[-\frac{\theta}{36}, \frac{\theta}{36}]$. This facilitates the use of symmetric QSVT functions in subsequent steps. 
We narrow the transition interval by a buffer of width $\gamma/3$; that is, our approximation to the Heaviside function will transition from close to 0 to close to 1 on the interval $[-\frac{\theta}{36} + \frac{\gamma}{3},\frac{\theta}{36}-\frac{\gamma}{3}]$. This ensures the error after applying the approximate Heaviside function to $\hat{v}_j \notin \hat{V}_\theta$ is bounded by $\delta$, whenever the error on $\hat{v}_j$ is bounded by $\gamma$. We choose $\gamma=\theta/24$, so that the transition region is now between $[-\frac{\theta}{72}, \frac{\theta}{72}]$. The QSVT circuit for $U_\Theta$ makes $\mathcal{O}\left(\frac{1}{\theta} \log\left(\frac{1}{\delta}\right)\right)$ calls to $U_{\hat{V}}$, corresponding the degree of the polynomial required to approximate the Heaviside function to error $\delta$, except on the transition window.

Fixed-point amplitude amplification is used to find violated constraints. In the worst case, only one constraint is violated, and the success probability is at least $\frac{(1-\delta)^2}{m}$, where the numerator accounts for the approximation in the Heaviside function in the first step and can always be bounded by a constant, so is neglected in the following discussion. Hence we use a degree $\mathcal{O}\left(\sqrt{m} \log\left(\frac{1}{\omega} \right)  \right)$ polynomial implemented via QSVT, which makes the same number of calls to the Heaviside block-encoding $U_\Theta$. Here $\omega$ bounds the deviation of the amplified values from $1$.

Overall, the algorithm makes $\mathcal{O}\left(\frac{\sqrt{m}}{\theta} \log\left(\frac{1}{\omega} \right) \log\left(\frac{1}{\delta} \right)  \right)$ calls to the block-encoding of $\hat{V}$ and a similar number of other single- and two-qubit gates, up to factors of $\log_2(\bar{n}m)$ due to the QSVT circuits. To account for the failure modes of the oracle, we examine each case separately. In case (i) ($\hat{V}_{>\theta}$ is nonempty) failure occurs if the circuit outputs $\ket{0^\perp}$ on the ancilla register (indicating all constraints are satisfied due to a failure of amplitude amplification) or the algorithm returns $j \notin \hat{V}$ (due to amplification of satisfied constraints with small but non-zero amplitude after the imperfect Heaviside function). The probability of either of these occurrences can be bounded by the trace distance to the ideal output state, which by \cref{Lemma:ImperfectAABE} is bounded by $4 \sqrt{m \delta}\log\left(\frac{1}{\omega}\right) + \omega +\sqrt{2\omega}$. In case (ii) ($\hat{V}$ is empty) failure occurs if we sample $\ket{0}$. This can result from amplification of satisfied constraints with small but non-zero amplitude after the imperfect Heaviside function. The probability is bounded
\begin{equation}
\lVert Q(\tilde{A})\ket{\psi_0}\rVert^2 \leq \lVert Q(A)-Q(\tilde{A})\rVert^2 \leq 16m\delta \log^2\left(\frac{1}{\omega} \right)    
\end{equation}
where $Q(A)$ is the fixed point amplitude amplification polynomial applied to the projected unitary encoding $A = \ket{0}\bra{0}U_\Theta \ket{\psi_0}\bra{\psi_0}$, and we have used that $A=0$ and applied \cite[Lemma 22]{Gilyn2019}. Finally, the failure mode of case (iii) ($\hat{V}_{>\theta}$ is empty but $\hat{V}_{\theta}$ is not) is the same as case (ii). Each of these cases are mutually exclusive. All of the contributions to the failure probability can be exponentially suppressed. We set $\omega=\eta^2/32$, and
\begin{align}
    \delta &= \frac{\eta^2}{64m\log^2\left(\frac{32}{\eta^2}\right)} ,
\end{align}
which are sufficient to ensure the maximum probability of failure is at most $\eta$.
\end{proof}

\begin{corollary}\label{cor:C_search_quantum}
    There is a quantum procedure that implements the sampled violated constraint search oracle of \cref{oracle:sampled_constraint_search} while incurring cost
    \begin{align}
    &C_{\rm search}(m,r,n,\beta,\theta,\eta, T')  ={}
        \bigOt(\frac{\sqrt{m}\beta}{\theta} \log^2(1/\eta)) 
    \end{align}
    calls to $O_R$, $O_R^\dagger$,  $O_{\mathbf{y}}$, $O_{\mathbf{y}}^
    \dagger$, $O_{\mathcal{T}}$, $O^\dagger_{\mathcal{T}}$, $O_{\mathbf{b}}$ plus up to a factor $\bigO(\log(r\bar{n}m)) $ other single- and two-qubit gates. 
\end{corollary}

\section{Classical implementation of two-step violated constraint oracle}\label{section:classical_implementation_violated_constraint}

We now give a classical algorithm in the sample-and-query access model that implements the violated constraint oracle, and by extension the full SOCP, as described in the following statements, which rely on claims shown later in the section. 

\begin{theorem}\label{thm:classical_complexity_violated_constraint_oracle}
    There is a classical algorithm that implements the violated constraint oracle of \cref{def:violated_constraint_oracle} for a program with $r$ cones, $n$ total variables, $m$ constraints, and $\beta = \nrm{\mathbf{y}}_1$, with failure probability $\xi$ and precision parameter $\theta$ while incurring complexity
    \begin{align}
       \bigOt\left(sn + \frac{m\log(1/\xi)}{\theta^4}\right) 
    \end{align}
    samples and queries to the instance data. 
\end{theorem}
\begin{proof}
    This follows from the expression of \cref{theorem:Csamp_Csearch_implements_alg1} in terms of $C_{\rm samp}$ and $C_{\rm search}$, and the evaluation of these complexity expressions in \cref{lem:C_samp_classical} and \cref{lem:C_search_classical}, respectively. 
\end{proof}
\begin{corollary}\label{cor:classical_complexity_feasibility_problem}
    There is a classical algorithm for solving the unit-trace SOCP feasibility problem of \cref{def:SOCP_feasibility} for a program with $r$ cones, $n$ total variables,  $m$ constraints, and precision parameter $\theta$, while incurring complexity
    \begin{align}
    \bigOt\left(\frac{n}{\theta^4}  + \frac{m}{\theta^6}\right) 
    \end{align}
    samples and queries to the instance data. 
\end{corollary}
\begin{proof}
    This follows from \cref{thm:classical_complexity_violated_constraint_oracle} and \cref{theorem:MW_algo}, noting that the maximum value of $s$ across all $T$ iterations of the algorithm is $T = \bigO(\log(2r)/\theta^2)$. 
\end{proof}

\begin{corollary}[Complexity of classical implementation of SOCP MW] Let $\mathcal{P}$ be an instance of the general primal SOCP of \cref{Def:SOCP} with $r$ cones, $n$ total variables, and $m$ constraints. Assume $\mathcal{P}$ obeys the normalisation conditions and strong duality, and that it is $R$-trace and $\tilde{R}$-dual-trace constrained. Given error parameter $\epsilon$, there is a classical algorithm that approximately solves $\mathcal{P}$ up to error $\epsilon$, in the sense of the statement of \cref{feasibility_reduction}, while incurring complexity
    \begin{align}
    \bigOt \left(n\left(\frac{R\tilde{R}}{\epsilon}\right)^4 + m\left(\frac{R\tilde{R}}{\epsilon}\right)^6\right) 
    \end{align}
    samples and queries to the instance data. 
%operating under the following parameters: $m$ denotes the number of constraints, $r$ represents the number of cones, and $\theta$  defines the approximation margin for the solution. Additionally, $n=r\cdot \bar{n}$, where $\bar{n}$ is the maximum size of the cones. 
\end{corollary}
\begin{proof}
As described in \cref{feasibility_reduction}, it suffices to make $\bigOt(\log(R/\epsilon))$ calls to an oracle for the feasibility SOCP with error parameter $\theta = \epsilon/(4R\tilde{R})$. Thus, the complexity follows from \cref{cor:classical_complexity_feasibility_problem}.
\end{proof}

We leave as an open problem whether the classical complexity can be reduced from $\bigOt(n+m)$ to $\bigOt(r+m)$ (when $R\tilde{R}/\epsilon = \bigO(1)$), which would closer match the quantum algorithm.  

\subsection{Classical implementation of cone index Gibbs sampler oracle}

\begin{lemma}\label{lem:C_samp_classical}
There is a classical procedure that implements the cone index Gibbs sampler oracle of \cref{oracle:cone_index_Gibbs_sampling}---that is, it generates $T'$ samples from a joint distribution at most $\zeta$-far in total variation distance from $T'$ i.i.d.~samples from the ideal distribution---while incurring cost
    \begin{align}
        C_{\rm samp}(m,r,n,s,\beta,\zeta, T') = \bigO(sn + T')
    \end{align} 
    samples and queries to the input data and other arithmetic operations.
\end{lemma}
\begin{proof}
    We use query access to the data in $A^\pp{k}$ for $k=0,\ldots, r-1$  to explicitly compute $\mathbf{u}^\pp{k} = A^{\pp{k} \top}\mathbf{y}$ for each $k$ via matrix multiplication. Since $\mathbf{y}$ is $s$-sparse, this matrix multiplication requires $\bigO(n^\pp{k} s)$ query complexity and other arithmetic operations for each $k$, for a total of $\bigO(ns)$ complexity to iterate over all values of $k$. Once $\mathbf{u}^\pp{k}$ is computed exactly, one can compute $\mathcal{Z}^\pp{k}$ exactly for each $k$ with $\bigO(n)$ arithmetic operations. Using classical random number generator, one can then draw $T'$ samples from the correct distribution with exact precision, at total cost $\bigO(sn)$ queries and $\bigO(sn + T')$ arithmetic operations. 
\end{proof}

It could be possible to implement the cone index Gibbs sampler oracle in complexity scaling linearly with $r$ rather than $n$, although an immediate approach fails to achieve this, as we now explain. The idea would be to estimate each of the $r$ values $\mathcal{Z}^{\pp{k}}$ without writing down the whole length-$n^{\pp{k}}$ vector $\mathbf{u}^{\pp{k}}$, by leveraging the ability to sample from rows of $A^{\pp{k}}$. One can make progress toward this by noting that the sample-and-query model enables one to estimate norms of matrix-vector products, such as $\nrm{\mathbf{u}^{\pp{k}}} = \nrm{A^{\pp{k}\top} \mathbf{y}}$ and to compute $u_0^{\pp{k}}$, which could then be used to estimate $\mathcal{Z}^{\pp{k}}$. The roadblock is that achieving additive precision $\Delta$ on the estimate of $\nrm{\mathbf{u}^{\pp{k}}}$ requires $\mathrm{poly}(1/\Delta)$ complexity (in contrast to the quantum approach, which achieves high precision directly QSVT polynomials with degree scaling as $\bigO(\log(1/\Delta)$). Since there are $r$ values of $k$ to estimate, it may be necessary to take $\Delta = 1/r$, which ruins the linear-in-$r$ complexity. 

\subsection{Classical implementation of sampled violated constraint search oracle}

\begin{lemma}\label{lem:C_search_classical}
    There is a classical procedure that implements the sampled violated constraint search oracle of \cref{oracle:sampled_constraint_search} while incurring cost
    \begin{align}
    C_{\rm search}(m,r,n,\beta,\theta,\eta,  T') = \bigO(sn) + \bigO\left(\frac{mT'\log(m/\eta)}{\theta^2}\right)
    \end{align}
    samples and queries to the input data, and other arithmetic operations. 
\end{lemma}
\begin{proof}
    We use query access to the data in $A^\pp{k_h}$ for $h=0,\ldots, T'-1$  to explicitly compute $\mathbf{u}^\pp{k_h} = A^{\pp{k_h}  \top}\mathbf{y}$ for each $k_h$ via matrix multiplication. Since $\mathbf{y}$ is $s$-sparse, this matrix multiplication requires $\bigO(n^\pp{k_h} s)$ query complexity and other arithmetic operations for each $k_h$, for a total of no more than $\bigO(ns)$ complexity to iterate over all values of $h$. Once $\mathbf{u}^\pp{k_h}$ is computed exactly, one can compute $\mathbf{p}^\pp{k_h}$ exactly by exponentiation and normalisation, again accomplished for all values of $h$ in no more than $\bigO(n)$ arithmetic operations. Recall that in our classical data access model, we assume sample-and-query access to the input data in each row $A_{j,:}^\pp{k}$. Using this access along with query access to $\mathbf{p}$, the result of \cite[Prop.~4.2]{tang2019quantumInspired} shows that for each $h$ and for each $j$, we can compute an estimate for the inner product $A^\pp{k}_{j,:} \mathbf{p}$. This estimate for the inner product can be shifted by $b_j$ (which can be queried exactly) to compute an estimate $\tilde{v}_{j,h}$ for $\hat{v}_{j,h} = A_{j,:}^\pp{k_h} \mathbf{p}^\pp{k_h} - b_j$, which satisfies $|\hat{v}_{j,h}-\tilde{v}_{j,h} | \leq \mu\nrm{A^{\pp{k-h}}_{j,:}}\nrm{\mathbf{p}^\pp{k_h}}$ with probability at least $1-\delta$. The complexity is $\bigO(\log(1/\delta)/\mu^2)$ samples, queries, and other arithmetic operations. We note that since $\mathbf{p}^\pp{k_h} \succeq 0$, $\nrm{\mathbf{p}^\pp{k_h}} \leq \sqrt{2}p^\pp{k_h}_0 \leq 1/\sqrt{2}$, and we have assumed by convention that $\nrm{A^\pp{k_h}_{j,:}} \leq 1$. We may choose $\delta = \eta/m$, $\mu = \theta/12$. By averaging the $\tilde{v}_{j,h}$ over $h$, we obtain an estimate $\tilde{v}_j$ satisfying $|\tilde{v}_j - \hat{v}_j|\leq \theta/12$, which by the union bound holds simultaneously for all $j$ with probability at least $1-\eta$. If there exists an estimate $\tilde{v}_j > 3\theta/4$, we output one such value of $j$; otherwise, we output ``all constraints satisfied.'' Due to the bounded deviation between $\tilde{v}_j$ and $\hat{v}_j$, we may confirm output conditions (i), (ii), and (iii) for the oracle. The total complexity is $\bigO(sn)$ queries to construct the $\mathbf{p}^\pp{k_h}$ vectors, and $\bigO(mT'\log(m/\eta)/\theta^2)$ samples and queries to estimate all $mT'$ inner products to the desired precision, with a similar number of arithmetic operations. 
\end{proof}

\section*{Acknowledgements}
We thank Fernando Brand\~ao, Andr\'as Gily\'en, Urmila Mahadev and John Preskill for helpful discussions. We
also thank Simone Severini,  James Hamilton, Nafea Bshara, Peter DeSantis, and Andy Jassy for their involvement and support of the
research activities at the AWS Center for Quantum
Computing. Author MIFG was supported by the National Science Foundation under grant NSF CAREER award CCF-2048204. Part of this work was carried out while author MIFG was visiting the Simons Institute for the Theory of Computing.

\bibliographystyle{alpha}
\bibliography{Bib}

\appendix

\section{Block-encoding of Arrowhead matrix}

We provide an instantiation of the block-encoding of the arrowhead matrix based on the access model detailed in \cref{section:access_model}. First, we define the following vector (as in the main text) to ease the notation: 
\begin{equation}
    \mathbf{u}^\pp{k}:= A^{\pp{k} \top} \mathbf{y}
\end{equation}

\begin{equation}
    \mathbf{u}^\pp{k} = \begin{bmatrix}
        u^{\pp{k}}_0 \\
        \vec{u}^\pp{k}
    \end{bmatrix}  = \begin{bmatrix}
        (A^{\pp{k} \top})_{0,:}\mathbf{y}  \\
        (A^{\pp{k} \top})_{1:,:}\mathbf{y}
    \end{bmatrix}
\end{equation}

For the one cone case:
\begin{equation}
    \mathrm{diag}(u_0^\pp{k}/\beta):=
        \frac{u^\pp{k}_0}{\beta} I   
\end{equation}

\begin{theorem} [$U_{\operatorname{Arw}(\mathbf{u})}$, $(2\beta,0)$-block-encoding of Arrowhead matrix $\operatorname{Arw}(\mathbf{u})$] \label{BE_Arrohead_matrix_theorem}
    Given oracle access to the matrices \( A^{\pp{0}},\ldots, A^{\pp{r-1}} \in \mathbb{R}^{m \times \bar n} \) through \( O_R, O^{\dagger}_R \) and oracle access to a vector \( \mathbf{y} \in \mathbb{R}^m \) through \( O_\mathbf{y}, O^{\dagger}_\mathbf{y} \) along with its associated constant \( \beta = \nrm{\mathbf{y}}_1 \), we can construct the block-encoding $U$, with a subnormalisation constant \( 2\beta  \), for the matrix:
    \[
    \sum_{k=0}^{r-1} \ket{k}\bra{k} \otimes \operatorname{Arw}(A^{\pp{k}\top} \mathbf{y})
    \]
     This construction requires one query to each of $O_R$, $O_R^\dagger, O_{\mathbf{y}}$, and $O_{\mathbf{y}}^\dagger$, and $\bigO(\log(\bar{n}))$ additional single- and two-qubit gates.
\end{theorem}
\begin{proof}
We present as proof the schematic circuit in \cref{LCU_schematic} and its fully detailed implementation in \cref{LCU_detailed}. We can divide the block-encoding of the arrowhead matrix in two parts, the diagonal matrix $U_d$, and the off-diagonal $U_{od}$, as it can be easily viewed in \cref{def:arrowhead_def}. 

As demonstrated in \cref{Lemma_diag_be}, a single invocation of each oracle $O_R,\;O_{\mathbf{y}},\;O_{\mathbf{y}}^{\dagger}$ combined with $\bigO\!\bigl(\log\bar{n}\bigr)$ additional one- and two-qubit gates, suffices to implement the block-encoding $U_d$ that contains the desired diagonal elements.

Similarly, we can construct the block-encoding with off diagonal matrix $U_{od}$ with one query to $O_R$,$O^{\dagger}_R$, $O_{\mathbf{y}}$, $O^{\dagger}_{\mathbf{y}}$ and $\bigO(\log(\bar{n}))$ queries to one qubit and two qubit gates, as we prove in \cref{lemma_off_diag}.

Once we have access to these two block-encodings, we can sum them using standard techniques, in particular, we perform Linear Combination of Unitaries (LCU) between both block-encodings:
$$
U_{\mathrm{SEL}}:=|0\rangle\langle 0| \otimes U_{d}+|1\rangle\langle 1| \otimes U_{od},
$$
\[
V^{\dagger}_{\mathrm{PREP}}=
V_{\mathrm{PREP}} = H\,,
\]
where $H$ is the Hadamard gate.
Let $I_s$ denote the identity matrix of dimension $2^s$, where $s= \text{(number of qubits
 \cref{LCU_schematic})}-1$ is the number of qubits on which $U_d$ and $U_{od}$ act. Then, define the block-encoding unitary $U$ by
$$
U=\left(V^{\dagger}_{\mathrm{PREP}} \otimes I_s\right) U_{\mathrm{SEL}}\left(V_{\mathrm{PREP}} \otimes I_s\right)
$$
where $U$ is a block-encoding of $\operatorname{Arw}(\mathbf{u})$ with subnormalisation $2\beta$. An equivalent schematic representation of $U$ is given below.\\

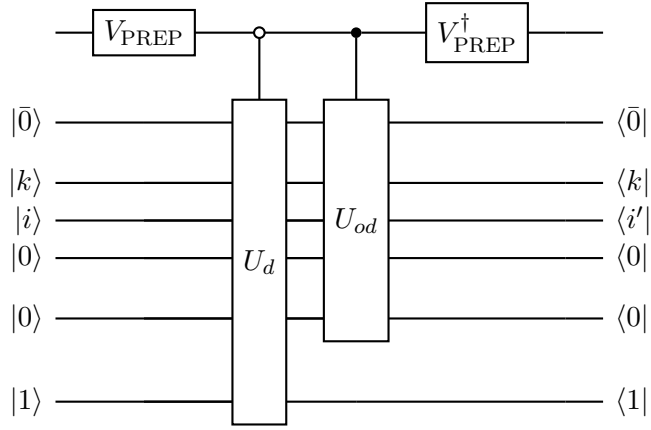
\begin{figure}[H]
\centering
\begin{quantikz}
 & \gate{V_{\mathrm{PREP}}} & \octrl{1} & \ctrl{1}& \gate{V^{\dagger}_{\mathrm{PREP}}} & \qw & \qw \\
 \lstick{$\ket{\bar{0}}$} & \qw & \gate[6]{U_d} & \gate[5]{U_{od}}& \qw & \qw & \rstick{$\bra{\bar{0}}$}\qw \\
\lstick{$\ket{k}$} & \qw & \qw & \qw & \qw & \qw & \rstick{$\bra{k}$}  \qw \\
\lstick{$\ket{i}$} & \qw & \qw & \qw & \qw & \qw &\rstick{$\bra{i'}$} \qw \\
\lstick{$\ket{0}$} & \qw & \qw & \qw & \qw & \qw & \rstick{$\bra{0}$} \qw \\
\lstick{$\ket{0}$} & \qw & \qw & \qw & \qw & \qw & \rstick{$\bra{0}$}\qw \\
\lstick{$\ket{1}$} & \qw & \qw & \qw & \qw & \qw  &\rstick{$\bra{1}$}\qw 
\end{quantikz}
\caption{Schematic LCU circuit.}\label{LCU_schematic}
\end{figure}
Since $U_{od}$ and $U_d$ have a common structure, they need not be performed sequentially and independently. A diagram of the circuit is shown in \cref{LCU_schematic}.

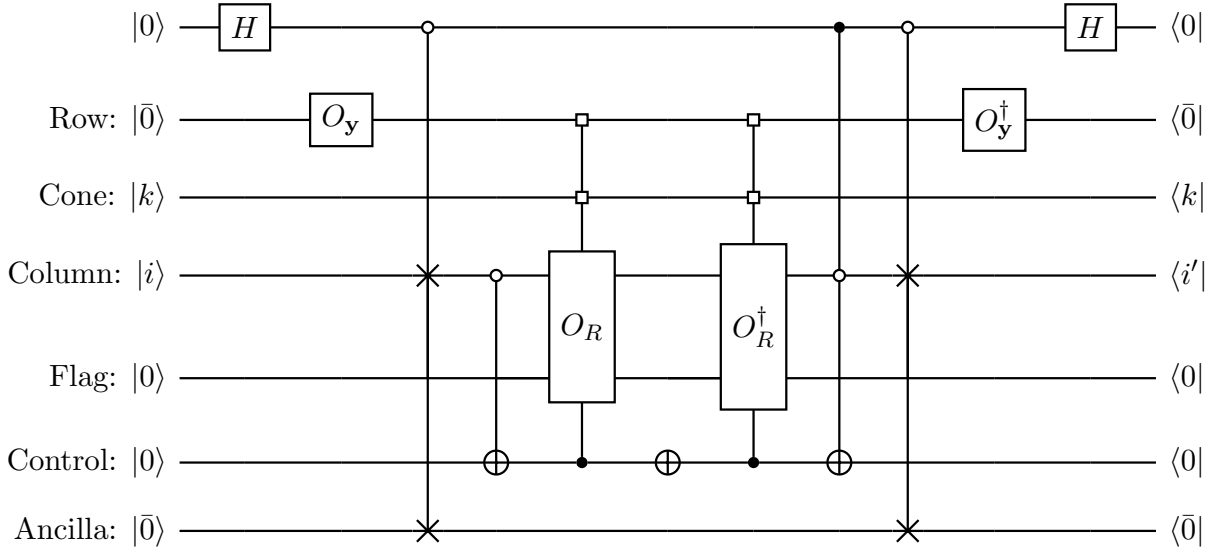
\begin{figure}[H]

\centering
\resizebox{\textwidth}{!}{
\begin{quantikz}
\lstick{$\ket{0}$} & \gate{H} & \qw &  \octrl{4}  &  \qw & \qw    & \qw & \qw  & \ctrl{3} & \octrl{4} & \qw  & \gate{H}& \qw \rstick{$\bra{0}$} \\
\lstick{Row: $\ket{\bar{0}}$} & \qw  & \gate{O_{\mathbf{y}}}      & \qw  & \qw & \mltplex{1}  & \qw  & \mltplex{1} & \qw & \qw &\gate{O^{\dagger}_{\mathbf{y}}}& \qw &  \rstick{$\bra{\bar 0}$}\qw   \\
\lstick{Cone: $\ket{k}$} & \qw &\qw &\qw   & \qw  & \mltplex{2} & \qw & \mltplex{2}  & \qw & \qw & \qw & \qw  & \rstick{$\bra{k}$}\qw \\
  \lstick{Column: $\ket{i}$} &\qw & \qw &\swap{3}  & \octrl{2}& \gate[2]{O_R} &\qw  & \gate[2]{O^{\dagger}_R}  & \octrl{2} & \swap{3} & \qw & \qw & \rstick{$\bra{i'}$}\qw\\
 \lstick{Flag: $\ket{ 0}$} & \qw & \qw & \qw  & \qw  & \qw & \qw & \qw &\qw & \qw & \qw  &\qw  &   \rstick{$\bra{ 0}$} \qw\\
  \lstick{Control: $\ket{0}$} & \qw &\qw   & \qw  & \targ{} & \ctrl{-1}  & \targ{}  & \ctrl{-1} & \targ{} & \qw & \qw & \qw & \rstick{$\bra{0}$}\qw\\
 \lstick{Ancilla: $\ket{\bar{0}}$}  & \qw & \qw   & \swap{} & \qw   & \qw & \qw & \qw &\qw & \swap{} & \qw & \qw &\rstick{$\bra{\bar{0}}$}\qw\\
\end{quantikz} 
}
\caption{Detailed LCU circuit.}\label{LCU_detailed}
\end{figure}

We provide an overview of the operations occurring between $V_{\mathrm{PREP}}$ and $V^{\dagger}_{\mathrm{PREP}}$. Let's start by considering that starting with the first register on $\ket{0}$ and post-measuring on $\bra{0}$ we obtain the correct result, that being, we are controlling on the diagonal block-encoding. Up to the gate $O^{\dagger}_R$:
\begin{align}
    \ket{0,\bar{0},k,i,0,0,\bar{0}} \mapsto \ket{0}
    \Bigg( &  \sum_{j,\tilde{i}} A^{\pp{k} }_{j\tilde{i}}\sqrt{\frac{ y_j}{2\beta}} \ket{j,k,\tilde{i}}\ket{0,0} +  \\
     &\sum_{j,\tilde{i}}\sqrt{1-|A^{\pp{k}}_{j\tilde{i}}|^2}\sqrt{\frac{ y_j}{2\beta}}\ket{j,k,\tilde{i}}\ket{1,0} \Bigg)\ket{i} \notag
\end{align}
Post-selecting on $\bra{0}$ on the first register, up to $O^{\dagger}_R$ from the left:
\begin{equation}
    \bra{0,\bar{0},k,i',0,0,\bar{0}} \mapsto \bra{0}\left( \sum^m_{j=1}\sqrt{\frac{y_j}{2\beta}}\bra{j,k,\bar{0}}\bra{0,0}\right)\bra{i'}
\end{equation}
Therefore: 
\begin{equation}
    z:=\bra{0,\bar{0},k,i',0,0,\bar{0}} U_{\mathrm{SEL}} \ket{0,\bar{0},k,i,0,0,\bar{0}} 
\end{equation}
\begin{itemize}
    \item If $ i'=i   \mapsto z= u^\pp{k}_0/2\beta $ 
    \item If $ i'\neq i \mapsto z=0 $ 
\end{itemize}
We do the same starting with $\ket{1}$ on the first register and post-selecting on $\bra{1}$, that being, we are controlling on the off-diagonal block-encoding. Up to the gate $O^{\dagger}_R$, if $i=\bar{0}$
\begin{align}
    \ket{1,\bar{0},k,i=\bar{0},0,0,\bar{0}} \mapsto \ket{1}
    \Bigg(&\sum_{j,\tilde{i}} A^{\pp{k} }_{j\tilde{i}}\sqrt{\frac{ y_j}{2\beta}} \ket{j,k,\tilde{i}}\ket{0,0} + \\
    &\sum_{j,\tilde{i}}\sqrt{1-|A^{\pp{k}}_{j\tilde{i}}|^2}\sqrt{\frac{ y_j}{2\beta}}\ket{j,k,\tilde{i}}\ket{1,0}\Bigg)\ket{\bar{0}} \notag 
\end{align}

\noindent If $i\neq\bar{0}$:
\begin{equation}
    \ket{1,\bar{0},k,i\neq\bar{0},0,0,\bar{0}} \mapsto \ket{1}
\sum_{j} \sqrt{\frac{ y_j}{2\beta}} \ket{j,k,i}\ket{0,1,\bar{0}} 
\end{equation}
Post-selecting on $\bra{1}$ on the first register, up to $O^{\dagger}_R$, for $i'=\bar{0}$
\begin{align}
    \bra{1,\bar{0},k,i'=\bar{0},0,0,\bar{0}} \mapsto \bra{1}\Bigg(&  \sum_{j,\tilde{i}}A^{\pp{k} }_{j\tilde{i}}\sqrt{\frac{y_j}{2\beta}}\bra{j,k,\tilde{i},0}\bra{0,1}+ \\
    & \sum_{j,\tilde{i}}\sqrt{1-|A^{\pp{k} }_{j\tilde{i}}|^2}\sqrt{\frac{y_j}{2\beta}}\bra{j,k,\tilde{i}}\bra{1,1}\Bigg)\bra{\bar{0}}  \notag
\end{align}
If $i\neq\bar{0}$:
\begin{equation}
    \bra{1,\bar{0},k,i'\neq\bar{0},0,0,\bar{0}} \mapsto \bra{1} \sum_{j}\sqrt{\frac{y_j}{2\beta}}\bra{j,k,i'}\bra{0,0,\bar{0}}
\end{equation}
Therefore: 
\begin{equation}
    z:=\bra{1,\bar{0},k,i',0,0,\bar{0}} U_{\mathrm{SEL}} \ket{1,\bar{0},k,i,0,0,\bar{0}} 
\end{equation}
\begin{itemize}
    \item If $ i'=\bar{0}, i=\bar{0}  \mapsto z=0 $ 
    \item If $ i'\neq \bar{0}, i=\bar{0} \mapsto z=\vec u_{i'}/2\beta $ 
    \item If $ i'= \bar{0}, i\neq\bar{0} \mapsto z=\vec u_{i}/2\beta$     
    \item If $ i'\neq \bar{0}, i\neq \bar{0} \mapsto z=0 $ 
\end{itemize}
Hence, given the reuse of queries to build the $U_{od}$ and $U_d$, the final query complexity is: 
one query to $O_R$, $O_R^\dagger$, $O_{\mathbf{y}}$ and $O_{\mathbf{y}}^\dagger$, and $\bigO(log_2(\bar{n}))$ additional single- and two-qubit gates. 
\end{proof}

Next, we give the constructions for the block-encodings of the diagonals and off-diagonal elements.

\begin{lemma}[$U_d$, $(\beta,0)$-block-encoding of diagonal matrix]\label{Lemma_diag_be}
Given oracle access to the matrices \( A^{\pp{0}},\ldots, A^{\pp{r-1}} \in \mathbb{R}^{m \times \bar n} \) through \( O_R, O^{\dagger}_R \) and oracle access to a vector \( \mathbf{y} \in \mathbb{R}^m \) through \( O_\mathbf{y}, O^{\dagger}_\mathbf{y} \) along with its associated constant \( \beta \) (as specified in \cref{O_y}), we can construct the block-encoding $U_d$, with a subnormalisation constant \( \beta  \), for the matrix:
\[
\sum_{k=0}^{r-1} \ket{k}\bra{k} \otimes \mathrm{diag}(u^\pp{k}_0/\beta)
\]
 This construction requires one query to each of $ O_R, O_{\mathbf{y}}$ and $O_{\mathbf{y}}^\dagger$, and $\bigO(\log(\bar{n}))$ additional single- and two-qubit gates.
\end{lemma}
\begin{proof}

% \begin{figure}[h]\label{figure:diagonal}
% \resizebox{0.8\textwidth}{!}{
% \centering
% \begin{quantikz}
% \lstick{Column: $\ket{\bar{0}}$} & \qw  & \gate{O_y}  & \mltplex{1}     &\qw & \qw & \mltplex{1}   & \gate{O^{\dagger}_y}  &  \qw & \rstick{$\bra{\bar 0}$}\qw \\
% \lstick{Cone: $\ket{k}$ }  & \qw  & \qw & \mltplex{1}  &\qw  & \qw & \mltplex{1}  & \qw  & \qw & \qw\\
% \lstick{Row: $\ket{i}$} & \swap{4} & \qw  &  \gate[2]{O_R}   &\swap{2} & \qw & \gate[2]{O_R^\dag}& \qw   & \swap{4}  & \qw\\
%  \lstick{Flag: $\ket{0}$} & \qw  & \qw &   & \qw & \swap{2} &  & \qw    & \qw &  \rstick{$\bra{ 0}$} \qw\\
%   \lstick{Anc: $\ket{\bar{0}}$} & \qw  & \qw & \qw  & \swap{}  & \qw  & \qw & \qw    & \qw &  \rstick{$\bra{\bar{0}}$} \qw\\
%  \lstick{Anc: $\ket{ 0}$} & \qw & \qw & \qw   & \qw  & \swap{}  & \qw & \qw & \qw  & \qw \rstick{$\bra{ 0}$} \qw\\
% \lstick{Anc: $\ket{\bar{0}}$} & \swap{}   & \qw& \qw & \qw  &  \qw & \qw   & \qw   & \swap{} & \qw \rstick{$\bra{\bar{0}}$}\\
% \end{quantikz} 
% }\caption{Diagonal block-encoding. $\ket{\bar{0}}$ here represents $\ket{0^{log_2(m)}}$. $\bullet$ also refers to $\bullet^{\otimes log_2(r)}$.  Same applies for $\text{\scalebox{0.5}{$\square$}}$, which is a simplification for $\text{\scalebox{0.5}{$\square$}}^{\otimes log_2(m)}$}
% \end{figure}

\begin{figure}[H]\label{figure:diagonal}
\centering
\resizebox{0.7\textwidth}{!}{
\begin{quantikz}
\lstick{Row: $\ket{\bar{0}}$} & \gate{O_{\mathbf{y}}}  &  \qw       & \qw & \mltplex{1}   & \qw  &\qw & \gate{O^{\dagger}_y}  & \rstick{$\bra{\bar 0}$}\qw \\
\lstick{Cone: $\ket{k}$ } & \qw  & \qw &  \qw & \mltplex{1}   & \qw   & \qw & \qw   & \qw\\
\lstick{Column: $\ket{i}$}& \qw & \swap{3} &  \qw &  \gate[2]{O_R}    & \qw & \swap{3} & \qw  & \qw \rstick{$\bra{ i'}$} \qw\\
 \lstick{Flag: $\ket{0}$}& \qw   & \qw &  \qw &  \qw  & \qw  & \qw   & \qw   &  \rstick{$\bra{ 0}$} \qw\\
 \lstick{Control: $\ket{ 0}$} & \qw & \qw & \targ{} & \ctrl{-1}   & \targ{}    & \qw & \qw   & \qw \rstick{$\bra{ 0}$} \qw\\
   \lstick{Ancilla: $\ket{\bar{0}}$}  & \qw  & \swap{} & \qw   &  \qw  & \qw & \swap{}    & \qw  &  \rstick{$\bra{\bar{0}}$} \qw\\
\end{quantikz} 
}\caption{Diagonal block-encoding. $\ket{\bar{0}}$ here represents top to bottom $\ket{0^{log_2(\bar{m})}}$ and  $\ket{0^{log_2(\bar{n})}}$. $\text{\scalebox{0.5}{$\square$}}$ is here a simplification for $\text{\scalebox{0.5}{$\square$}}^{\otimes log_2(m)}$, $\text{\scalebox{0.5}{$\square$}}^{\otimes log_2(r)}$}
\end{figure}

Up to the $O_R$ gate, we see the circuit behaves as the following map
\begin{align}
  \ket{\bar{0},k,i,0,0,\bar{0}} & \mapsto 
\left(\sum_{j,\tilde{i}} A^{\pp{k}}_{j\tilde{i}} \sqrt{\frac{y_j}{\beta} }\ket{j,k,\tilde{i}}\ket{0} + \sum_{j,\tilde{i}} \sqrt{1-|A^{\pp{k} }_{j\tilde{i}}|^2} \sqrt{\frac{y_j}{\beta} }\ket{j,k,\tilde{i}}\ket{1}\right)\ket{1,i}
 \end{align}
For an initialisation $k,i'$ from left to right:
\begin{align}
  \ket{\bar{0},k,i',0,0,\bar{0}} & \mapsto 
 \sum_{j}\sqrt{\frac{y_j}{\beta}} \ket{j,k,\bar{0},0,1,i'}
 \end{align}
Then for $i=i'$ we are selecting on the correct value $ u^\pp{k}_0/\beta$. Then the circuit description with the appropriate initialization gives the diagonal block-encoding:
\begin{equation}
    \bra{\bar{0},k,I_{\bar{n}},0,0,\bar{0}}  
    U_d\ket{\bar{0} ,k,I_{\bar{n}},0,0,\bar{0}} = \mathrm{diag}(u^\pp{k}_0/\beta)  
\end{equation}

Therefore, up to a swap between $\ket{column,cone}$ registers, we confirm the above circuit finalises the proof of construction of $\sum_{k=0}^{r-1} \ket{k}\bra{k} \otimes \mathrm{diag}(u^\pp{k}_0/\beta)$. 

\end{proof}
As with the diagonal block-encoding, we now define the off-diagonal matrix we wish to instantiate:
\begin{equation}
    \text{Off-diag}(\vec u^\pp{k}/ \beta):= \begin{pmatrix}
        0 & \vec u^{\pp{k} \top}/ \beta  \\
        \vec u^\pp{k}/ \beta & 0         
    \end{pmatrix}
\end{equation}

 \begin{lemma}
[$U_{od}$, $(\beta,0)$-block-encoding of off-diagonal matrix] \label{lemma_off_diag}
Given oracle access to the matrices \( A^{\pp{0}},\ldots, A^{\pp{r-1}} \in \mathbb{R}^{m \times \bar n} \) through \( O_R, O^{\dagger}_R \) and oracle access to a vector \( \mathbf{y} \in \mathbb{R}^m \) through \( O_\mathbf{y}, O^{\dagger}_\mathbf{y} \) along with its associated constant \( \beta \) (as specified in \cref{O_y}), we can construct the block-encoding $U_{od}$, with a subnormalisation constant \( \beta  \), for the matrix:
\[
\sum_{k=0}^{r-1} \ket{k}\bra{k} \otimes \text{Off-diag}(\vec u^\pp{k}/ \beta)
\]
 This construction requires one query to each of $ O_R$, $O^{\dagger}_R$, $ O_{\mathbf{y}}$,  $O_{\mathbf{y}}^\dagger$, and \newline $\bigO(\log(\bar{n}))$ additional single- and two-qubit gates.
\end{lemma}

\begin{proof}

We provide as proof the below circuit \cref{circuit_off_diag} that builds  $U_{od}$, the block-encoding of the matrix $
\sum_{k=0}^{r-1} \ket{k}\bra{k} \otimes \text{Off-diag}(\vec u^\pp{k}/ \beta)
$, up to an extra swap between $\ket{row, cone}$\\

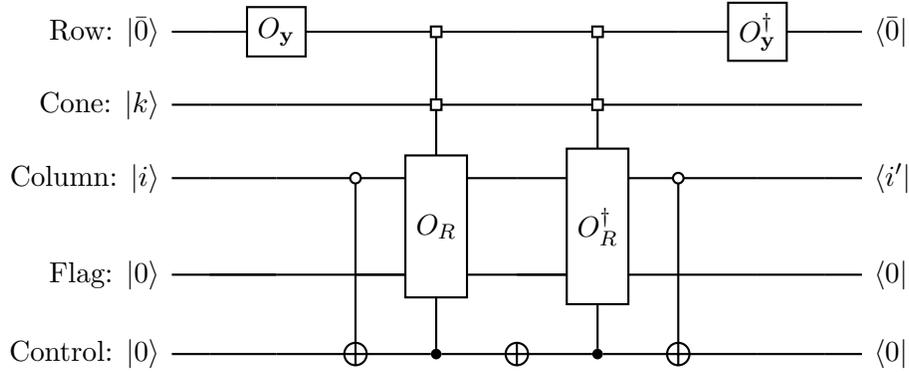
\begin{figure}[H]
\centering
\begin{quantikz}
\lstick{\text{Row: }$\ket{\bar 0}$}  & \qw  & \gate{O_{\mathbf{y}}} & \qw & \mltplex{1} & \qw & \mltplex{1} & \qw & \gate{O^{\dagger}_{\mathbf{y}}} & \qw   & \rstick{$\bra{\bar 0}$} \qw\\
 \lstick{\text{Cone: }$\ket{k}$}   & \qw  & \qw &\qw & \mltplex{1} & \qw & \mltplex{1} & \qw & \qw & \qw & \qw  & \\
 \lstick{\text{Column: }$\ket{i}$}  & \qw & \qw & \octrl{2} & \gate[2]{O_R} & \qw  & \gate[2]{O^{\dagger}_R} & \octrl{2} & \qw & \qw & \qw  \rstick{$\bra{i'}$}  \\
\lstick{\text{Flag: }$\ket{0}$}   & \qw & \qw \qw & \qw & \qw & \qw & \qw & \qw & \qw & \qw   & \rstick{$\bra{0}$} \qw\\
\lstick{\text{Control: }$\ket{0}$} & \qw   & \qw & \targ{} & \ctrl{-1} & \targ{} & \ctrl{-1} & \targ{} & \qw  & \qw   & \rstick{$\bra{0}$} \qw\\
\end{quantikz}
\caption{ Off-diagonal elements block-encoding. Top to bottom, the first $\ket{\bar{0}}=\ket{0^{log_2(m)}}$. We also clarify $\text{\scalebox{0.5}{$\square$}}$ is a simplification for $\text{\scalebox{0.5}{$\square$}}^{\otimes log_2(m)}$ and $\text{\scalebox{0.5}{$\square$}}^{\otimes log_2(r)}$ .
%\toremove{that only gets activated if all the registers are $\ket{O^{\otimes log(\bar{n})}}$: do we need to specify anything on this??}
Similarly $\circ$ here refers to open control and its application on the column register is $\circ^{\otimes log_2(\bar{n})}$.The open control acts when all the registers it controls are on state 0, instead of the conventional one that acts when the register is on state 1.  \label{circuit_off_diag}}
\end{figure}

To understand the circuit, we give the following pointers: it is symmetrical around the X gate in the middle of the circuit. We see that the left part of the circuit performs the following map:
\begin{equation}
    \ket{\bar{0},k,i=\bar{0},0,0} \mapsto \left( \sum_{j,\tilde{i}}A^{\pp{k}}_{j\tilde{i}}\sqrt{\frac{y_j}{\beta}}\ket{j,k,\tilde{i}}\ket{0}+\sum_{j,\tilde{i}}\sqrt{1-|A^{\pp{k}}_{j\tilde{i}}|^2}\sqrt{\frac{y_j}{\beta}}\ket{j,k,\tilde{i}}\ket{1} \right)\ket{1}
\end{equation}

Next, we see what happens in the right part of the circuit, for the case where $i'\neq\bar{0}$:
\begin{equation}
    \bra{\bar{0},k,i'\neq\bar{0},0,0} \mapsto  \sum^{m-1}_{j=0}\sqrt{\frac{y_j}{\beta }} \bra{j,k,i',0,0} 
\end{equation}

By considering the bit-flip gate at the center of the circuit, we conclude the proof that the construction realizes the correct block-encoding. Specifically, for $ z := \bra{\bar{0},k,i',0,0} \, U_{od} \, \ket{\bar{0},k,i,0,0}$, we have

\begin{itemize}
    \item If $ i'=\bar{0}, i=\bar{0}  \mapsto z=0 $ 
    \item If $ i'\neq \bar{0}, i=\bar{0} \mapsto z=\vec u_{i'}/\beta $ 
    \item If $ i'= \bar{0}, i\neq\bar{0} \mapsto z=\vec u_{i}/ \beta$     
    \item If $ i'\neq \bar{0}, i\neq \bar{0} \mapsto z=0 $ 
\end{itemize}

\end{proof}

\section{Quantum implementation toolset}
\subsection{Minimum finding}

\begin{theorem}[\cite{lin2020near} Theorem 8]\label{theorem:min_finding}
Suppose we have Hamiltonian $H=\sum_k \lambda_k\left|\psi_k\right\rangle\left\langle\psi_k\right| \in$ $\mathbb{C}^{N \times N}$, where $\lambda_k \leq \lambda_{k+1}$, given through its $(\beta, a, 0)$-block-encoding\footnote{%
This refers to a $(\beta, 0)$-block-encoding implemented with $a$ ancilla qubits, following the notation convention used throughout the document.%
}
 $U_H$. Also suppose we have an initial state $\left|\phi_0\right\rangle$ prepared by circuit $U_I$, as well as the promise $|\braket{\phi_0|\psi_0}|\geq \nu$, where $\nu>0$. Then the ground energy can be estimated to precision $\eta_\lambda$ with probability $1-\eta$ with the following costs:
\begin{enumerate}
    \item Query complexity: 
    \[
    \mathcal{O}\left(\frac{\beta}{\nu \eta_\lambda} \log \left(\frac{\beta}{\eta_\lambda}\right) \log \left(\frac{1}{\nu}\right) \log \left(\frac{\log (\beta / \eta_\lambda)}{\eta}\right)\right) \text{ queries to } U_H
    \]
    \[
    \text{and } \mathcal{O}\left(\frac{1}{\nu} \log \left(\frac{\beta}{\eta_\lambda}\right) \log \left(\frac{\log (\beta / \eta_\lambda)}{\eta}\right)\right) \text{ queries to } U_I,
    \]
    % We are not using this fact \item Number of qubits: 
    % $\mathcal{O}\left(n+a+\log \left(\frac{1}{\nu}\right)\right)$,
    \item Other one- and two-qubit gates: 
    $\mathcal{O}\left(\frac{a \beta}{\nu \eta_\lambda} \log \left(\frac{\beta}{\eta_\lambda}\right) \log \left(\frac{1}{\nu}\right) \log \left(\frac{\log (\beta / \eta_\lambda)}{\eta}\right)\right)$.
\end{enumerate}

\end{theorem}

\subsection{QSVT}

Similarly, some of the following statements have been drawn from \cite{vanapeldoorn19} and \cite{Gilyn2019}.

\begin{theorem} [Polynomial eigenvalue transformation,  \citep{zerosumLP} Theorem 6]\label{QSVT_poly_eig_trans_theorem}
Suppose that $U$ is an $a$-qubit block-encoding of a Hermitian matrix $A$, and $P \in \mathbb{R}[x]$ is a degree-d polynomial satisfying that\\

(i) for all $x \in[-1,1]:|P(x)| \leq \frac{1}{2}$, or

(ii) for all $x \in[-1,1]:|P(x)| \leq 1$ and $|P(x)|=|P(-x)|$.\\

Then there is a quantum circuit $\tilde{U}$, which is an $(a+2)$-qubit block-encoding of $P(A)$, and which consists of $d$ applications of $U$ and $U^{\dagger}$, (and in case (i) a single application of controlled $U^{ \pm 1}$ ) and $\mathcal{O}((a+1) d)$ other one- and two-qubit gates. 
\end{theorem}

\begin{lemma}[Polynomial approximations, modified \citep{zerosumLP} Lemma 7]\label{Lemma_poly_approx} Let $2\beta \geq 1, \varepsilon \leq 1 / 2$. There exist a polynomial $\tilde{P}$ such that
$$\forall x \in[-1,1]:|\tilde{P}(x)| \leq \frac{1}{2} \text{ and for all } x \in[-1,0]:\left|\tilde{P}(x)-e^{2\beta x} / 4\right| \leq \varepsilon$$
    
moreover $\operatorname{deg}(\tilde{P}) =\mathcal{O}(\beta \log (1 / \varepsilon))$
\end{lemma}

\begin{theorem}[Fixed-point amplitude amplification, modified from theorem 27 \citep{Gilyn2019}]\label{theorem:fixed_point_aa}
    Let $U$ be a unitary and $\Pi$ be an orthogonal projector such that $a\left|\psi_G\right\rangle=\Pi U\left|\psi_0\right\rangle$, and $a>\delta>0$. There is a unitary circuit $\tilde{U}$ such that $\|\left|\psi_G\right\rangle-\tilde{U}\left|\psi_0\right\rangle \| \leq \varepsilon$, which uses a single ancilla qubit and consists of $\mathcal{O}\left(\frac{\log (1 / \varepsilon)}{\delta}\right) U, U^{\dagger}, C_{\Pi} N O T$, $C_{\left|\psi_0\right\rangle\left\langle\psi_0\right|} N O T$ and $e^{i \phi \sigma_z}$ gates.
\end{theorem}

\begin{corollary}[Quantum signal processing using reflections, \citep{Gilyn2019} Corollary 8 ]\label{corollary:8_qsvt}Let $P \in \mathbb{C}[x]$ be a degree-d polynomial, such that
\begin{itemize}
    \item $P$ has parity- $(d \bmod 2)$,
    \item $\forall x \in[-1,1]:|P(x)| \leq 1$,
    \item $\forall x \in(-\infty,-1] \cup[1, \infty):|P(x)| \geq 1$,
    \item if $d$ is even, then $\forall x \in \mathbb{R}: P(i x) P^*(i x) \geq 1$.
\end{itemize}
Observe that $c \geq 1$ for all $a, b \geq 0$ and thus $\sqrt{c^2-1} \in \mathbb{R}$. Then there exists $\Phi \in \mathbb{R}^d$ such that

$$
\prod_{j=1}^d\left(e^{i \phi_j \sigma_z} R(x)\right)=\left[\begin{array}{cc}
P(x) & \cdot \\
\cdot & \cdot
\end{array}\right]
$$

\noindent Moreover for $x \in\{-1,1\}$ we have that $P(x)=x^d \prod_{j=1}^d e^{i \phi_j}$, and for $d$ even $P(0)=e^{-i \sum_{j=1}^d(-1)^j \phi_j}$.
\end{corollary} 

\begin{lemma}[Robustness of singular value transformation, \citep{Gilyn2019} Lemma 22 ]\label{lemma:robustness_qsvt}If $P \in \mathbb{C}[x]$ is a degree-n polynomial satisfying the requirements of \cref{corollary:8_qsvt}, moreover $A, \tilde{A} \in \mathbb{C}^{\tilde{d} \times d}$ are matrices of operator norm at most 1, then we have that
$$
\left\|P^{(S V)}(A)-P^{(S V)}(\tilde{A})\right\| \leq 4 n \sqrt{\|A-\tilde{A}\|}
$$
\end{lemma}

\begin{lemma}\label{Lemma:ApproxHeaviside}
    (Approximation of Heaviside step function) Given a diagonal block-encoding $A = \sum_j a_j \ket{j}\bra{j}$, there is a quantum circuit that implements a block-encoding of $f(A)$ defined such that $|f(x) - \Theta(X)| \leq \delta$ for $x \in [-1, -c] \cup [c, 1]$. The circuit makes $\mathcal{O}\left(\frac{1}{c} \log\left(\frac{1}{\delta} \right) \right)$ calls to the block-encoding of $A$.
\end{lemma}
\begin{proof}
    Apply QSVT using the approximation of the Heaviside step function defined in \cite{Martyn2021}.
\end{proof}

\begin{lemma}\label{Lemma:ImperfectAABE}
    (Fixed point amplitude amplification of imperfect block-encoding) Given a $\delta$- approximate block-encoding $\widetilde{W}$ of an operator $X$, a state preparation unitary for state $\ket{\psi_0}$, and a known lower bound $\Lambda \leq ||X \ket{\psi_0}||$, there exists a quantum circuit that prepares $\ket{\psi_g} = \frac{\ket{0}X \ket{\psi_0}}{||X\ket{\psi_0}||}$ to trace-distance bounded by $\frac{4\sqrt{\delta}}{\Lambda}\log\left(\frac{1}{\omega}\right) + \omega + \sqrt{2\omega}$ which makes $\mathcal{O}\left(\frac{1}{\Lambda
    } \log\left(\frac{1}{\omega}\right) \right)$ calls to $\widetilde{W}$. 
\end{lemma}
\begin{proof}
    Denote the circuit by $U_{AA}$. The circuit applies $\widetilde{W}$ to $\ket{0} \ket{\psi_0}$ and amplifies the success probability with fixed-point amplitude amplification, which requires $\mathcal{O}\left(\frac{1}{\Lambda} \log\left(\frac{1}{\omega}\right) \right)$ calls to $\widetilde{W}$~\cite[Theorem 27]{Gilyn2019}. Both $\widetilde{W}$ and fixed-point amplitude amplification are imperfect, and we bound their errors here. In the absence of errors on $W$, it forms a projected unitary encoding of a state preparation operator
    \begin{equation}
       \left( \ket{0}\bra{0}\otimes I \right) W \left( \ket{0}\bra{0}\otimes \ket{\psi_0}\bra{\psi_0} \right) = \ket{0}\bra{0}\otimes a \ket{\psi_g}\bra{\psi_0} := A
    \end{equation}
where $a = ||X\ket{\psi_0}||$ is the success amplitude. Similarly, $\widetilde{W}$ is a projected unitary encoding of $ \tilde{A} = \ket{0}\bra{0}\otimes \tilde{a} \ket{\psi_g'}\bra{\psi_0}$, where $\ket{\psi_g'}$ is the state resulting from applying $\widetilde{W}$ to $\ket{0}\ket{\psi_0}$, and post-selecting on $\ket{0}$. We can bound
\begin{align}
    &||A - \tilde{A}|| \\
    = &|| \left( \ket{0}\bra{0}\otimes I \right) W \left( \ket{0}\bra{0}\otimes \ket{\psi_0}\bra{\psi_0} \right) - \left( \ket{0}\bra{0}\otimes I \right) \widetilde{W} \left( \ket{0}\bra{0}\otimes \ket{\psi_0}\bra{\psi_0} \right)  || \\
    \leq &|| \left( \ket{0}\bra{0}\otimes I \right) W \left( \ket{0}\bra{0}\otimes I \right) - \left( \ket{0}\bra{0}\otimes I \right) \widetilde{W} \left( \ket{0}\bra{0}\otimes I \right)  || \\
    \leq &\delta
\end{align}
by definition.

Likewise, $U_{AA}$ is a projected unitary encoding of $Q(\tilde{A})$, where $Q(\cdot)$ is the fixed point amplitude amplification polynomial. Thus, $U_{AA}\ket{0}  \ket{\psi_0} = \ket{0}Q(\tilde{A})\ket{\psi_0} + \ket{\perp}$. We have $|| Q(\tilde{A})\ket{\psi_0} - \ket{\psi_g'}|| \leq \omega$ by choice of fixed point amplitude amplification polynomial. Hence $|| Q(\tilde{A})\ket{\psi_0}|| \geq 1-\omega$, and $||\ket{\perp} || \leq \sqrt{2\omega}$.
    
The trace distance between $\ket{0}\ket{\psi_g}$ and $U_{AA} \ket{0}\ket{\psi_0} $ is bounded by
\begin{align}
    &|| \ket{0}\ket{\psi_g} - U_{AA} \ket{0}\ket{\psi_0}|| \\
    = &|| \ket{0}\ket{\psi_g} - \ket{0}Q(\tilde{A})\ket{\psi_0} + \ket{\perp} || \\
    \leq &||\ket{0}\ket{\psi_g} - \ket{0} Q(A)\ket{\psi_0} + \ket{0} Q(A)\ket{\psi_0} - \ket{0}Q(\tilde{A})\ket{\psi_0} || + ||\ket{\perp} || \\
    \leq & \omega + ||Q(A) - Q(\tilde{A})|| + \sqrt{2\omega}
\end{align}
The middle term can be bounded by applying the robustness of QSVT:
\begin{align}
    &||Q(A) - Q(\tilde{A}) || \\
    \leq &4 \cdot \mathrm{deg}(Q)\cdot \sqrt{||A - \tilde{A}||} \\
    \leq &\frac{4\sqrt{\delta}}{\Lambda}\log\left(\frac{1}{\omega}\right)
\end{align}
where in the final line we used the degree of the fixed point amplitude amplification polynomial and the bound on $||A-\tilde{A}||$ derived above.

\end{proof}

\begin{lemma}[Polynomial approximation of $\tanh$]\label{lem:polynomial_for_tanh}
Given a real number $u > 0$ and error parameter $\varepsilon_{\rm tanh} > 0$, there exists a polynomial $p(x)$ such that whenever $|x|\leq 1$, $|p(x) - \tanh(ux)| \leq \varepsilon_{\rm tanh}$. Furthermore, if $u \leq 1$, then the degree of $p$ is
$d = \bigO( \log(1/\varepsilon_{\rm tanh})/\log(1/u))$ (as $u \rightarrow 0$), and if $u \geq 1$, the degree is $d = \bigO(u \log(u/\varepsilon_{tanh}))$ (as $u \rightarrow \infty$).  
\end{lemma}
\begin{proof}
    The $\tanh$ function satisfies $|\tanh(z)|\leq 1$ for all real $z$. Prior work has derived a polynomial approximation for $\tanh(uz)$ by truncating its Taylor series, which has exponential convergence with the truncation degree as long as $z$ is within the radius of convergence of the Taylor series. In our application, this is an effective approach when $|u|\leq 1$. In that case, we can extend the calculation from \cite[Appendix C]{guo2024nonlinearTransformation} and \cite[Lemma 14]{rattew2023nonlinearTransformationsQuantumAmplitudes}, which was performed at $u=1$. Namely, they show that if we take $p(x)$ to be the degree-$d$ truncation of the Taylor series for $\tanh(ux)$, then we have
    \begin{align}
        |p(x) - \tanh(ux)| \leq \left\lvert \sum_{j = d+1}^{\infty} \alpha_j (ux)^{2j-1} \right\rvert
    \end{align}
    with 
    \begin{align}
        |\alpha_j | \leq 5 \left(\frac{2}{\pi}\right)^j
    \end{align}
    Evaluating the sum and imposing $|x|\leq 1$ gives
    \begin{align}
        |p(x) - \tanh(ux)| \leq  \frac{5}{u}\sum_{j = d+1}^{\infty} \left(\frac{2u^2}{\pi}\right)^j \leq 14\left(\frac{2}{\pi}\right)^{d+1} u^{2d-1}
    \end{align}
    where we have used $5/(1-(2u^2/\pi))\leq 14$ when $|u| \leq 1$. Thus, to achieve error $\varepsilon_{\tanh}$, it suffices to take $d= \Theta(\log(1/\varepsilon_{\rm tanh})/\log(1/u))$. 
    
    Since $\tanh(uz)$ is not analytic on the entire complex plane---it has poles at $z=i\pi(\ell + 1/2)/u$ for integer $\ell$ (the nearest poles to 0 are $\pm i \pi/(2u)$---this method is not suitable for us when $u$ is large enough for the radius of convergence to fall below 1.  Instead, for the regime $u > 1$,  we use results from approximation theory that show the existence of a good approximating polynomial for functions that are analytic on the interior of a Bernstein ellipse, defined as the set $E_{\rho} = \{\frac{1}{2}(v + v^{-1}) \colon v \in \mathbb{C}, |v| = \rho\}$, with $\rho>1$. Here, we choose $\rho = 1 + \pi/(4u)$. We note that if $v$ is purely imaginary, that is $v = \pm i\rho$, then $\frac{1}{2}(v + v^{-1}) = \pm iY$, where $Y=\frac{(\pi/4u) + (\pi^2/32 u^2)}{1+\pi/4u} \leq \pi/4u$. Thus, for this value of $\rho$, the poles of $\tanh(uz)$ lie outside of the Bernstein ellipse, and $\tanh(uz)$ is analytic on the interior of $E_{\rho}$. We now compute an $M$ for which there exists an upper bound $|\tanh(uz)|\leq M$ which holds for all $z$ in the interior of $E_{\rho}$. Above, we have established that for all such $z$, we have $|\Im(uz)| \leq \pi/4$, so in particular $\Re(e^{i \Im(uz)}) \geq \cos(\pm \pi/4) = 1/\sqrt{2}$. Thus, we may bound 
    \begin{align}
        \Re(\cosh(uz)) &= \frac{1}{2}(\Re(e^{uz}) + \Re(e^{-uz})) \\
        &= \frac{1}{2}(e^{\Re(uz)} \Re(e^{i \Im(uz)}) +e^{-\Re(uz)} \Re(e^{-i \Im(uz)})) \geq \frac{e^{u |\Re(z)|}}{2\sqrt{2}}
    \end{align}
    and thus $|\cosh(uz)|\geq \frac{e^{u |\Re(z)|}}{2\sqrt{2}}$.  Furthermore, we have $|\sinh(uz)| \leq e^{u|\Re(z)|}$. Thus, in the interior of the Bernstein ellipse, we have $|\tanh(uz)|\leq M$ with $M=2\sqrt{2}$. From \cite[Theorem 20]{tang2024CSguideQSVT} and references therein, there is a degree-$d$ polynomial $p(x)$ formed as a Chebyshev series for which
    \begin{align}
        |p(x) - f(x)| \leq \frac{2M}{\rho^d (\rho-1)} = \frac{16\sqrt{2} u}{\pi} (1+\frac{\pi}{4u})^{-d} 
    \end{align}
    From this equation, we see that given target error $\varepsilon_{\rm tanh}$, it suffices to take $d = \Theta(u\log(u/\varepsilon_{\tanh})$.
\end{proof}

\end{document}